\documentclass[onecolumn,journal,12pt]{IEEEtran}

\ifCLASSINFOpdf
\else
\fi
\usepackage{amssymb, amsmath, latexsym}
\usepackage[dvips]{graphicx}
\usepackage{color}
\usepackage{comment}

\newtheorem{thm} {\textcolor{black}{Theorem}}
\newtheorem{lem}[thm] {\textcolor{black}{Lemma}}
\newtheorem{prp}[thm] {\textcolor{black}{Proposition}}
\newtheorem{rem}[thm] {\textcolor{black}{Remark}}
\newtheorem{df}[thm]{\textcolor{black}{Definition}}
\def\QED{\mbox{\rule[0pt]{1.5ex}{1.5ex}}}

\def\endproof{\hspace*{\fill}~\QED\par\endtrivlist\unskip}

\def\R{{\mathbb{R}}}

\def\N{{\mathbb{N}}}

\def\C{{\mathbb{C}}}

\def\R{\mathbb{R}}
\def\C{\mathbb{C}}

\newcommand{\calP}{{\cal P}}
\newcommand{\calX}{{\cal X}}
\newcommand{\calY}{{\cal Y}}

\def\QED{\mbox{\rule[0pt]{1.5ex}{1.5ex}}}

\def\endproof{\hspace*{\fill}~\QED\par\endtrivlist\unskip}

\def\Label#1{\label{#1}\ [\ #1\ ]\ }
\def\Label{\label}

\begin{document}
\title{Second-Order Asymptotics of Conversions of Distributions and Entangled States
Based on Rayleigh-Normal Probability Distributions}

\author{Wataru~Kumagai,~
        Masahito~Hayashi~\IEEEmembership{Fellow, IEEE.}
\thanks{W. Kumagai is with Faculty of Engineering, Kanagawa University. e-mail: kumagai@kanagawa-u.ac.jp}
\thanks{M. Hayashi is with Nagoya University and National University of Singapore. e-mail: masahito@math.nagoya-u.ac.jp}
\thanks{This paper was presented in part at ISIT 2013.}
}

\maketitle

\begin{abstract}

We discuss the asymptotic behavior of conversions between two independent and identical distributions up to the second-order conversion rate when the conversion is produced by a deterministic function from the input probability space to the output probability space.
To derive the second-order conversion rate,
we introduce new probability distributions named Rayleigh-normal distributions.
The family of Rayleigh-normal distributions includes a Rayleigh distribution and coincides with the standard normal distribution in the limit case.
Using this family of probability distributions, we represent the asymptotic second-order rates for the distribution conversion.
As an application,
we also consider the asymptotic behavior of conversions between the multiple copies of two pure entangled states in quantum systems 
when only local operations and classical communications (LOCC) are allowed.
This problem contains entanglement concentration, entanglement dilution and a kind of cloning problem with LOCC restriction as special cases.


\end{abstract}

\begin{IEEEkeywords}
Random number conversion, 
LOCC conversion, 
Second-order asymptotics, 
Rayleigh-normal distribution. 
\end{IEEEkeywords}

\IEEEpeerreviewmaketitle

\section{Introduction}\label{sec:introduction}

The topics of this paper cover three different areas:
probability theory,
information theory
and quantum information theory.
We give below concise reasons why our results affect the three areas.
First, we derive a new family of probability distributions called Rayleigh-normal distributions as a solution of a simple optimal approximation condition of a normal distribution. 
The Rayleigh-normal distributions are parameterized by a positive real number and connect a Rayleigh distribution and the standard normal distribution via these parameters.
Since this distribution family is a new object and these properties are meaningful to handle these two kinds of distributions, 
this result has importance in probability theory. 
Second, thanks to the central limit theorem, 
an optimal conversion problem for probability distributions can be translated into the approximation condition of a normal distribution. 
Then, 
as a contribution to information theory, we find that the Rayleigh-normal distributions 
determine the second-order asymptotic behavior of conversion between probability distributions.
Third, conversions between quantum pure entangled states can be reduced to conversions between probability distributions.
In particular, as a contribution to quantum information theory, the Rayleigh-normal distributions determine the second-order asymptotic behavior of conversion between pure entangled states.

In the rest of this section,
we describe our contributions to each area and their relation  in more detail.


\subsection{Contribution to Probability Theory}
To characterize the second-order asymptotics of conversions between probability distributions and entangled states,
we introduce a family of new probability distributions on real numbers called a Rayleigh-normal distribution.
Besides its operational meaning,
it has its own interesting mathematical properties. 
The family of Rayleigh-normal distributions
is parameterized by a positive real value $v$,
and contains a Rayleigh distribution with a specific parameter
at $v=1$ 
and the standard normal distribution at $v=0$.
Also, it coincides with
the standard normal distribution in the limit as the parameter $v$ tends to infinity.
That is,
the family connects a Rayleigh distribution and the standard normal distribution, which is the origin of the name "Rayleigh-normal distribution."
The Rayleigh-normal distribution is defined as the solution of an optimization problem for continuous probability distributions on real numbers as in (\ref{eq:rn}). 
While the definition seems very abstract, 
we give the explicit form of its cumulative distribution function $Z_v$ for each parameter $0\le v<\infty$ in Theorem \ref{Zform}.
The explicit form of the Rayleigh-normal distribution is numerically computable and has four different expressions 
depending on the cases when $v=0$, $0<v<1$, $v=1$ and $v>1$.
Then, we can plot the graphs of the cumulative distribution functions of the Rayleigh-normal distributions as in Fig. \ref{RN},
and it is shown that  
the family of Rayleigh-normal distributions has
a kind of symmetry with respect to the parameter $v$
and some useful properties.

\subsection{Contribution to Conventional Information Theory}
In information theory, statistics and computer science, 
it is an important task to generate a random number, 
which is required for  stochastic simulation and information-theoretic security.
To guarantee its quality, we need to avoid pseudo-random random numbers.
In this case, we need to convert a physically generated random number to another random number.
Hence, the problem of random number conversion 
has been studied as one of the main topics in information theory \cite{VV95,VKV98,Han05,Hay08,NH13}.
The most important task is the conversion from a non-uniform random number to a uniform random number.
This problem has been discussed in information-theoretic security
because uniform random numbers are used as a resource for information-theoretic security.
However, for stochastic simulation, the required random number is not necessarily the uniform random number 
because many stochastic simulations use random seeds subject to non-uniform distributions, which depend on the purpose of the simulation \cite{BFS12,AT89}. 
For this demand, so many algorithms were developed to generate non-uniform random numbers \cite{Dev86,HLD04}.
Thus, this paper addresses conversions of general random numbers that are not necessarily uniform.
In the following, for a precise description of the problem, we discuss this problem as a conversion of distributions
because the difficulty of the problem depends on the distributions of the initial and target random numbers.

In this paper, we focus on conversions between the $n$-fold independent and identical distributions of two different distributions, and investigate its asymptotic conversion rate.
While the first-order conversion rate is known to be the ratio of the Shannon entropies \cite{Han03}, 
the second-order conversion rate has not been revealed.
One of our aims is showing the attainability and the optimality up to the second-order conversion rate. 
In the following, given a map $W$ from a finite set $\mathcal{X}$ to another finite set $\mathcal{Y}$,
we define a conversion $W$ from the set of probability distributions on $\mathcal{X}$ to that on $\mathcal{Y}$
as $W(P)(y):=P(W^{-1}(y))$ for a probability distribution $P$ on $\mathcal{X}$.
This conversion is called the {\it deterministic} conversion induced by $W$,
where the word 'deterministic' comes from the non-probabilistic property of $W$.
That is, a deterministic conversion describes our possible operation for conversion.

If we need to show only the attainability, it is enough to 
simply discuss only the class of deterministic conversions.
However, to show the optimality, we need to consider a larger class of conversions that contains the class of 
deterministic conversions.
A map $W'$ from the set of probability distributions on $\mathcal{X}$ to that on $\mathcal{Y}$
is called a {\it majorization} conversion 
if the majorization relation $P\prec W'(P)$ holds for any probability distribution $P$ on $\mathcal{X}$.
Interestingly, the precision of the class of majorization conversions is more easily upper bounded
than that of the class of deterministic conversions
because the property of majorization effectively works for the evaluation of the optimality.
Since the class of majorization conversions contains the class of deterministic conversions,
we focus on the class of majorization conversions in the proof of the optimality.
Further, the class of majorization conversions plays an important role in quantum information theory as well
because a quantum operation called an LOCC conversion for pure states is mathematically reduced to the majorization conversion of probability distribution. 

Throughout this paper, we consider both kinds of conversions 
between two independent and identical distributions of two given distributions $P$ and $Q$,
and employ the fidelity (or Bhattacharyya coefficient) $F$ as the measure of conversion accuracy.
When $P^n$ denotes the $n$-fold independent and identical distribution of 
the distribution $P$,
we mainly focus on the following integers, i.e.,
the maximum conversion number 
from $P$ to $Q$ under a permissible accuracy $0<\tau<1$ 
by deterministic conversions
\begin{eqnarray*}
L^{\cal D}_n(P, Q|\tau)
:=\max\{L\in\N~|
~\exists~ \hbox{deterministic~conversion } W \hbox{ s.t. }
F(W(P^n), Q^L)\ge\tau\},
\end{eqnarray*}
and
that by majorization conversions 
\begin{eqnarray*}
L^{\cal M}_n(P, Q|\tau)
:=\max \{L\in\N~|
~\exists~ \hbox{majorization~conversion } W' \hbox{ s.t. }F(W'(P^n), Q^L)\ge\tau\}.
\end{eqnarray*}
Those numbers represent how many copies of the target probability distribution $Q$ can be generated from the initial  probability distribution $P^n$ under the accuracy constraint $\tau$.
It is known that the first order coefficient of $L^{\cal D}_n(P, Q|\tau)$
is the ratio of the Shannon entropies $H(P)$ and $H(Q)$ \cite{Han03}
and does not depend on the accuracy $\tau$. 
Recently, as a more precise asymptotic characterization,
the second-order asymptotics attracts much attention \cite{Hay09,PPV10,Hay08}.
When either initial or target probability distribution is uniform,
the asymptotic expansions of these numbers are solved up to the second-order $\sqrt{n}$, whose coefficient depends on the constraint of the accuracy \cite{Hay08,NH13}.
However, the derivation of the second-order conversion rate has remained open for the non-uniform case (i.e. neither given nor target probability distribution is uniform).

In this paper, we show that the second-order asymptotics of two kinds of conversions can be essentially reduced into an optimal approximation problem of a normal distribution.
Moreover, we reveal that the Rayleigh-normal distribution is obtained by the solution of the optimal approximation problem
and
the asymptotic behavior of the maximum conversion numbers is described by the inverse function of a cumulative Rayleigh-normal distribution $Z_{v}$ with certain constants $D_{P,Q}$ and $v=C_{P,Q}$ as follows:
\begin{eqnarray}
L^{\cal D}_n(P, Q|\tau)
~\cong~ L^{\cal M}_n(P, Q|\tau)
~=~ \frac{H(P)}{H(Q)}n + \frac{Z_{C_{P,Q}}^{-1}(1-\tau^2)}{D_{P,Q}}\sqrt{n}
+o(\sqrt{n}),
\Label{exp.gen0}
\end{eqnarray}
where 
$\cong$ shows that the difference between the left and the right side terms is $o(\sqrt{n})$.
The asymptotic expansion (\ref{exp.gen0}) gives an operational meaning of the Rayleigh-normal distribution, i.e., it characterizes how the second-order conversion rate depends on the constraint for the accuracy of the conversion.

\subsection{Contribution to Quantum Information Theory}

In quantum information theory,
various quantum tasks have been proposed
and a specific entangled state is often required to implement those tasks.
In such a situation,
maximally entangled states are used as typical resource of entanglement.
However, other kinds of entangled states can be also used for efficient quantum tasks.
For example, in port-based teleportation, the optimal entangled state to be  used as the resource is different from the maximally entangled state \cite{IH09}.
As other examples, measurement based quantum computation \cite{GFE09} 
and quantum channel estimation \cite{Hay11}
require entangled states that are not necessarily maximally entangled.

When some distant parties want to implement some quantum tasks,
they have to prepare the desired entangled state in advance.
Then,
the distant parties can perform only restricted operations named LOCC.
Here, LOCC is a combination of local operations (LO) and classical communication (CC),
where 
LO represents quantum operations on each individual party and CC represents sharing of classical information described by bits between parties. 
LOCC is a fundamental method to convert a given entangled state into a desired entangled state shared between distant places.

Based on the motivation,
we consider LOCC conversion between multiple copies of general pure entangled states on bipartite systems in this paper.
We especially focus on the following integer, i.e., the maximum conversion number from $\psi$ to $\omega$ by LOCC under a permissible accuracy $0<\tau<1$
\begin{eqnarray}
L_n(\psi, \omega|\tau)
&:=&\max\{L\in\N~|
~\exists~ \hbox{LOCC } \Gamma \hbox{ s.t. }
F(\Gamma(\psi^{\otimes n}), \omega^{\otimes L})\ge\tau\}\Label{L0},
\end{eqnarray}
where $F$ is the fidelity between quantum states.
This number represents how many copies of the target entangled state $\omega$ can be generated from a given entangled state $\psi^{\otimes n}$ by LOCC under the accuracy constraint.
As a fundamental result of LOCC conversion,
Bennett {\it et. al.} \cite{BBPS96} showed that the first-order optimal LOCC conversion rate from a pure entangled state $\psi$ to another one $\omega$ is the ratio  of von Neumann entropies $S_{\psi}$ and $S_{\omega}$ of their reduced density matrices.
Moreover,
Kumagai and Hayashi \cite{KH13} derived 
the second-order conversion rate for entanglement dilution and entanglement concentration, which corresponds to the case when either $\psi$ or $\omega$ is the maximally entangled state.
The result of \cite{KH13} implies that the second-order asymptotic expansion of $L_n(\psi, \omega|\tau)$ in entanglement dilution and entanglement concentration are represented by the cumulative standard normal distribution function $\Phi$ as 
\begin{eqnarray}
L_n(\psi, \omega|\tau)
~=~ \frac{S_{\psi}}{S_{\omega}}n+ {const}_{\psi,\omega} ~\Phi^{-1}(1-\tau^2)\sqrt{n}+o(\sqrt{n}),
\Label{L1}
\end{eqnarray}
where 
$ {const}_{\psi,\omega}$ is a constant given in (\ref{qua.dil}) and (\ref{qua.con}).
In fact, 
besides entanglement dilution and entanglement concentration,
the cumulative standard normal distribution function $\Phi$ commonly appears in 
the second-order rates for typical quantum information-processing tasks including quantum hypothesis testing \cite{Li12,TH12}, 
classical-quantum channel coding \cite{TT13},
quantum fixed-length source coding \cite{TH12,DL14}, 
data compression with quantum side information \cite{TH12},
randomness extraction against quantum side information \cite{TH12}
and noisy dense coding \cite{DL14}.

In this paper,
we consider LOCC conversion when $\omega$ or $\psi$ are not necessarily maximally entangled.
This setting is more important when the entangled states are used for quantum tasks which require non-maximally entangled states such as measurement-based quantum computation \cite{GFE09} and quantum channel estimation \cite{Hay11}.
Thus, these tasks require us to efficiently generate non-maximally entangled states by LOCC conversion.
Surprisingly,
it is shown that the second-order optimal LOCC conversion rate between general pure states $\psi$ and $\omega$ cannot be represented by the cumulative distribution function of the standard normal distribution but by that of  the Rayleigh-normal distribution as follows:
\begin{eqnarray}
L_n(\psi, \omega|\tau)
~=~ \frac{S_{\psi}}{S_{\omega}}n + \frac{Z_{C_{\psi,\omega}}^{-1}(1-\tau^2)}{D_{\psi,\omega}}\sqrt{n} +o(\sqrt{n}),
\Label{exp.LOCC0}
\end{eqnarray}
where $D_{\psi,\omega}$ and $C_{\psi,\omega}$ are certain constants.
Since (\ref{exp.LOCC0}) is not contained in (\ref{L1}) in general,
our result is different from conventional behavior of second-order rates
and is quite nontrivial.
In particular,
it is clarified that 
the Rayleigh-normal distribution has an operational meaning also in quantum information theory from (\ref{exp.LOCC0}).
When either the initial state $\psi$ or the target entangled state $\omega$ is a maximally entangled state,
the cumulative Rayleigh-normal distribution function coincides with the cumulative standard normal distribution function $\Phi$
and the above expansion (\ref{exp.LOCC0}) recovers (\ref{L1}).
The asymptotic expansion (\ref{exp.LOCC0}) is similar to the form in (\ref{exp.gen0}).
In fact, 
it is shown that (\ref{exp.LOCC0}) is essentially equivalent to (\ref{exp.gen0}) in Section \ref{sec:application}.

Next, as a special situation of LOCC conversion, we focus on the case when the target
entangled state is the same with the given entangled state  (i.e. $\omega=\psi$).
In this special case, the formula (3) can be simplified to 
\begin{eqnarray}
L_n(\psi,\psi|\tau)
~=~ n+{\frac{\sqrt{8V_{\psi}\mathrm{log}\tau^{-1}}}{S_{\psi}}}\sqrt{n} +o(\sqrt{n}),
\Label{exp.clone0}
\end{eqnarray}
where $V_{\psi}$ is a constant depending on $\psi$.
This case can be regarded as a special type of asymptotic cloning problem, which garnered some interest recently.
However, our problem is different from conventional settings of cloning
in the following points.
While the knowledge of the state to be cloned is not perfect in the conventional setting, 
our setting assumes the perfect knowledge for the entangled state to be cloned. 
The essential point of our setting is
that our operations are restricted to LOCC operations and  
no additional entangled resource are used.
We note that the papers \cite{ACP04,OH06} also treat cloning problems under LOCC operations,
however,  their setting assumes an imperfect knowledge for the entangled state to be cloned and additional limited entangled resource unlike our setting.
To distinguish their setting, we call our setting the LOCC cloning with perfect knowledge, and call their setting the LOCC cloning with imperfect knowledge.

To characterize the performance of cloning, 
Chiribella et al.  \cite{CYY13} introduced the replication rate as the order of the number of the incremental copies after cloning.
When we apply their definition to the case of LOCC cloning with imperfect knowledge although they discussed the replication rate in the case of another type of cloning,
the replication rate is the order of the number $L_n(\psi,\psi|\tau)-n$ of the incremental copies in the optimal LOCC cloning. 
Then, the formula (\ref{exp.clone0}) shows that 
the replication rate of the LOCC cloning with perfect knowledge is 1/2.

\subsection{Outline of This Paper}

The paper is organized as follows.
In Section \ref{sec:family}, 
we introduce a new family of probability distributions on real numbers.
It describes the optimal conversion rate under the accuracy constraint in Section \ref{sec:asymptotic}.
In Section \ref{sec:non-asymptotic}, 
as problems in conventional information theory,
we formulate two kinds of approximate conversion problems between two probability distributions by using the deterministic transformation and  the majorization condition, respectively.
Then, we define the maximum conversion numbers 
and describe their properties in non-asymptotic setting.
In Section \ref{sec:asymptotic}, 
we derive the asymptotic expansion of these numbers up to the second-order $\sqrt{n}$.
In these derivations, we divide our setting into two cases: uniform case and non-uniform case. The non-uniform case itself
does not contain the uniform case; however, we show that the results in the
uniform case can be regarded as the limit of the results in the non-uniform case.
In Section \ref{sec:quantum}, 
we apply the results of Sections \ref{sec:asymptotic} to the LOCC conversion.
Then, we obtain the optimal LOCC conversion rate between general pure states up to the second-order $\sqrt{n}$. 
As a special case, we derive the rate of the incremental copies and the optimal coefficient for the LOCC cloning with the perfect knowledge.
In Section \ref{sec:conclusion}, 
we give the conclusion.

We give an outline of relations of our results.
In conversion to or from uniform distributions,
only quantile function of an initial or target distribution is important.
However, in conversion between general probability distributions,
we have to focus on the total behavior of distributions.
Applying the central limit theorem,
the problem in the second-order asymptotics can be reduced into an optimal approximation problem of  a normal distribution given by (\ref{eq:rn}). 
Then, we define new probability distributions called Rayleigh-normal distributions by (\ref{eq:rn}) 
and show their essential properties which are inevitable to discuss how the second-order asymptotics for conversion of probability distributions can be reduced to the optimal approximation problem  in (\ref{eq:rn}).
In particular, 
we show that the second-order performance of conversion is described by the Rayleigh-normal distribution.

\section{Rayleigh-Normal Distribution}\Label{sec:family}

We treat an optimal approximation problem 
with the continuous fidelity as an approximation measure 
and define a new class of probability distributions called the Rayleigh-normal distributions in (\ref{eq:rn}) of subsection \ref{sec:intro}.
As shown in Section \ref{sec:asymptotic},
the second-order asymptotics of conversions between probability distributions can be reduced to the optimal approximation problem of a normal distribution 
by the central limit theorem.
Since properties of the Rayleigh-normal distributions essentially determine the second-order conversion rate,
we give some properties in subsection \ref{sec:property}.

\subsection{Introduction of Rayleigh-Normal Distribution}
\Label{sec:intro}

In this subsection,
we introduce a new probability distribution family on $\R$ with one parameter which connects the standard normal distribution and a Rayleigh distribution with a specific parameter.
A function $Z$ on $\R$ is generally called a cumulative distribution function if $Z$ is right continuous, monotonically increasing  and satisfies $\displaystyle\lim_{x\to-\infty} Z(x)=0$ and $\displaystyle\lim_{x\to \infty} Z(x)=1$.
Then, there uniquely exists a probability distribution on $\R$ whose cumulative distribution coincides with $Z$.
That is, 
given a cumulative distribution function in the above sense,
it determines a probability distribution on $\R$.
To define the new probability distribution family,
we give its cumulative distribution functions.

For $\mu\in\R$ and $v>0$,
let $\Phi_{\mu,v}$ and $\phi_{\mu,v}$ be the cumulative distribution function and the probability density function of the normal distribution with the mean $\mu$ and the variance $v$.
We denote $\Phi_{0,1}$ and $\phi_{0,1}$ simply by $\Phi$ and $\phi$.
Using the continuous fidelity (or the Bhattacharyya coefficient) for continuous probability distributions $p$ and $q$ on $\R$ defined by 
\begin{eqnarray}
{\cal F}(p,q):=\int_{\R}\sqrt{p(x)q(x)}dx,
\end{eqnarray} 
we define the following function.
\begin{df}\Label{rn}
For $v>0$, 
the Rayleigh-normal distribution function $Z_v$ on $\R$ is defined by
\begin{eqnarray}
Z_{v}(\mu)
=1-\sup_{A}{\cal F}\left(\frac{dA}{dx}, \phi_{\mu,v}\right)^2,
\Label{eq:rn}
\end{eqnarray}
where $A:\R\to[0,1]$ runs over continuously differentiable monotone increasing functions satisfying $\Phi\le A\le1$ in the right hand side.
For $v=0$, 
the Rayleigh-normal distribution function $Z_0$ on $\R$ is defined to be the cumulative distribution function $\Phi$ of the standard normal distribution.
\end{df}

The Rayleigh-normal distribution function is proven to be a cumulative distribution function later, 
and thus, 
it determines a probability distribution on $\R$.
In addition, the right continuity of $Z_{v}$ for $v$ at $v=0$ is also shown latter.
The graphs of the Rayleigh-normal distribution functions can be described as in Fig. \ref{RN} by Proposition \ref{Zform}.
\begin{figure}[t]
 \begin{center}
 \hspace*{0em}\includegraphics[width=100mm, height=60mm]{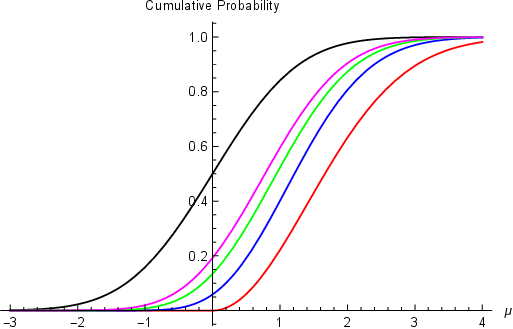}
 \end{center}
 \caption{
The graphs of the Rayleigh-normal distribution functions.
The black, purple, green, blue and red lines 
are displayed from the upper to the lower and
represent the Rayleigh-normal distributions with parameter $v=0$, $1/10$, $1/6$, $1/3$ and $1$.
}
 \Label{RN}
\end{figure}

The second order optimization problem discussed in Section \ref{sec:asymptotic}
is essentially reduced to 
the problem with respect to normal distributions.
Hence, our asymptotic conversion problem 
is essentially reduced to the conversion problem between two normal distributions, which is the right hand side of (\ref{eq:rn}).
Therefore, our Rayleigh-normal distribution function $Z_v$ plays an essential role in 
second-order asymptotics of conversion of distributions.

We note that Rayleigh distributions are included in Weibull distributions and Weibull-normal distribution is already proposed \cite{BSC14}, 
however, our Rayleigh-normal distribution is different from the Weibull-normal distribution in \cite{BSC14}
because a Rayleigh distribution with a specific scale parameter is included in the family of Rayleigh-normal distributions and is not in that of the Weibull-normal distributions.
Thus, the notion of the Rayleigh-normal distribution is first introduced in this paper.

\subsection{Properties of Rayleigh-Normal Distribution}
\Label{sec:property}

In this subsection,
we solve the optimization approximation problem of  a normal distribution in (\ref{eq:rn})
and give some useful properties of the Rayleigh-normal distributions.
In particular, 
we show how the family of Rayleigh-normal distributions connects a Rayleigh distribution and the standard normal distribution.

To give an explicit form of the Rayleigh-normal distribution functions,
we prepare four lemmas.
Their proofs are given in Appendix \ref{app.RN}.

\begin{lem}\Label{sol2}
When $0<v<1$,
the equation with respect to $x$
\begin{eqnarray}\Label{threshold2}
\frac{\phi(x)}{\phi_{\mu,v}(x)}=\frac{1-\Phi(x)}{1-\Phi_{\mu,v}(x)}
\end{eqnarray}
has the unique solution $\beta_{\mu,v}$.
Then, the inequality $\beta_{\mu,v}<\frac{\mu}{1-v}$ holds and $\beta_{\mu,v}$ is differentiable and monotonically increasing with respect to $\mu$. 
\end{lem}

\begin{lem}\Label{sol1}
When $v>1$, 
the equation with respect to $x$
\begin{eqnarray}
\frac{\phi(x)}{\phi_{\mu,v}(x)}=\frac{\Phi(x)}{\Phi_{\mu,v}(x)}\Label{threshold1}
\end{eqnarray}
has the unique solution $\alpha_{\mu,v} $.
Then, the inequality $\alpha_{\mu,v} >\frac{\mu}{1-v}$ holds and $\alpha_{\mu,v}$ is differentiable and monotonically decreasing with respect to $\mu$. 
\end{lem}

\begin{lem}\Label{monotone}
For $v>0$ and $\mu \in \mathbb{R}$,
the ratio $\frac{\phi(x)}{\phi_{\mu,v}(x)}$ is strictly monotonically decreasing only on the interval $\mathcal{I}_{\mu,v}$ 
defined by 
\begin{eqnarray}
\mathcal{I}_{\mu,v}
=\left\{
\begin{array}{cll}
\R&\hbox{ if }&v=1~ {\rm and}~ \mu>0\\
 \emptyset &\hbox{ if }&v=1~ {\rm and}~ \mu\le 0\\
(\frac{\mu}{1-v},\infty)&\hbox{ if }&v>1\\
(-\infty,\frac{\mu}{1-v})&\hbox{ if }&v<1,
\end{array}
\right.
\Label{area}
\end{eqnarray}
where $ \emptyset$ is the empty set.
\end{lem}

\begin{lem}\Label{lem.converse}
Assume that real numbers $t\le t'$ satisfy the following condition {\rm ($\star$)}:\\
{\rm ($\star$)} There exist $s$ and $s'$ that satisfy the following three conditions:
\begin{eqnarray}
&{\rm (I)}&s\le t\le t'\le s',\label{I}\nonumber\\
&{\rm (II)}&\frac{{\Phi}(t)}{{\Phi}_{\mu,v}(t)}=\frac{\phi(s)}{\phi_{\mu,v}(s)},~~\frac{1-{\Phi}(t')}{1-{\Phi}_{\mu,v}(t')}=\frac{\phi(s')}{\phi_{\mu,v}(s')},\label{II}\nonumber\\
&{\rm (III)}&\frac{\phi(x)}{\phi_{\mu,v}(x)}~is~strictly ~monotonically~decreasing~on~the~interval~(s,s').\label{III}\nonumber
\end{eqnarray}
Then the following inequality holds
\begin{eqnarray}
&&\sup_{A}{\cal F}\left(\frac{dA}{dx},\phi_{\mu,v}\right)\nonumber\\
&\le&
\sqrt{{\Phi}(t)} \sqrt{{\Phi}_{\mu,v}(t)}
+
\int_{t}^{t'}\sqrt{\phi(x)} \sqrt{\phi_{\mu,v}(x)}dx 
+\sqrt{1-{\Phi}(t')} \sqrt{1-{\Phi}_{\mu,v}(t')},\Label{lem.con.ineq2}
\end{eqnarray}
where $A:\R\to[0,1]$ in the left hand side runs over continuously differentiable monotone increasing functions satisfying $\Phi\le A\le1$.
\end{lem}

Here we introduce a function $A_{\mu,v}:\R\to[0,1]$ with parameters $\mu\in\R$ and $v>0$ which is separately defined with respect to the value of $v$ as follows.
When $v=1$,
\begin{eqnarray}
&&A_{\mu,v}(x) ~=~ A_{\mu,1}(x) ~:=~  \left\{
\begin{array}{ll}
\Phi_{\mu,1}(x) &\hbox{ if }~\mu<0 \\
\Phi(x) &\hbox{ if }~\mu\ge0.
\end{array}
\right. \Label{A_1}
\end{eqnarray}
When $v>1$,
\begin{eqnarray}
&A_{\mu,v}(x)
~:=~ \left\{
\begin{array}{cl}
\frac{\Phi(\alpha_{\mu,v})}{\Phi_{\mu,v}(\alpha_{\mu,v})}\Phi_{\mu,v}(x) & \hbox{ if }~x\le \alpha_{\mu,v} \\
\Phi(x) & \hbox{ if }~\alpha_{\mu,v}\le x.
\end{array}
\right.&
\Label{A_2}
\end{eqnarray}
When $0<v<1$,
\begin{eqnarray}
&A_{\mu,v}(x)
~:=~ \left\{
\begin{array}{cl}
\Phi(x) &\hbox{ if }~x\le \beta_{\mu,v} \\
1-\frac{1-\Phi(\beta_{\mu,v})}{1-\Phi_{\mu,v}(\beta_{\mu,v})}(1-\Phi_{\mu,v}(x)) & \hbox{ if }~\beta_{\mu,v}\le x.
\end{array}
\right.&
\Label{A_3}
\end{eqnarray}
The function $A_{\mu,v}$ is represented in Figs. \ref{Gauss1} and \ref{Gauss2}.
\begin{figure}[t]
 \begin{center}
 \hspace*{0em}\includegraphics[width=100mm, height=60mm]{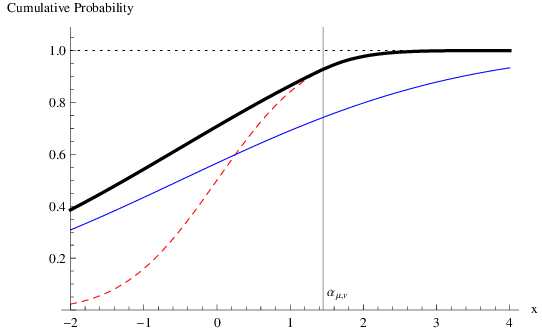}
 \end{center}
 \caption{
Let $v>1$. The dashed, the normal and the thick lines show $\Phi$, $\Phi_{\mu,v}$ and $A_{\mu,v}$, respectively. 
Then,
$A=A_{\mu,v}$ attains the supremum in (\ref{eq:rn}).
}
 \Label{Gauss1}
\end{figure}
%
%
\begin{figure}[t]
 \begin{center}
 \hspace*{0em}\includegraphics[width=100mm, height=60mm]{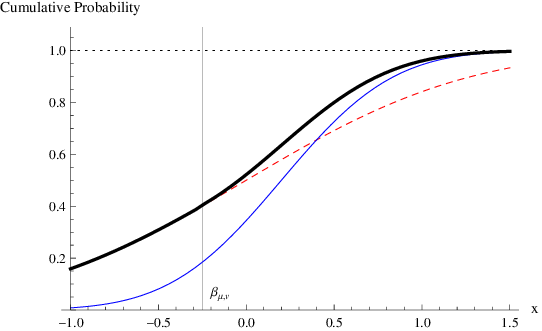}
 \end{center}
 \caption{
Let $0<v<1$. The dashed, the normal and the thick lines show $\Phi$, $\Phi_{\mu,v}$ and $A_{\mu,v}$, respectively. 
Then,
$A=A_{\mu,v}$ attains the supremum in (\ref{eq:rn}).
}
 \Label{Gauss2}
\end{figure}
%
\begin{lem}\Label{lem.converse2}
For an arbitrary $\epsilon>0$,
there exist real numbers $t\le t'$ which satisfy the condition {\rm ($\star$)} in Lemma \ref{lem.converse} and the following inequality
\begin{eqnarray}
&&\sqrt{{\Phi}(t)} \sqrt{{\Phi}_{\mu,v}(t)}
+
\int_{t}^{t'}\sqrt{\phi(x)} \sqrt{\phi_{\mu,v}(x)}dx +\sqrt{1-{\Phi}(t')} \sqrt{1-{\Phi}_{\mu,v}(t')}\nonumber\\
&\le&{\cal F}\left(\frac{dA_{\mu,v}}{dx},\phi_{\mu,v}\right) +\epsilon 
\Label{lem.con.ineq3}
\end{eqnarray}
\end{lem}

We denote the cumulative distribution function of the Rayleigh distribution with scale parameter $\sigma>0$ by
\begin{eqnarray}
R_{\sigma}(x)
~=~ \left\{
\begin{array}{ll}
1-e^{-\frac{x^2}{2\sigma^2}} &\hbox{ if } ~x>0\\
0 &\hbox{ if } ~x\le0.
\end{array}
\right.&
\Label{Ray}
\end{eqnarray}
Then,
a family of Rayleigh-normal distribution functions is represented as follows.
In particular, it includes the Rayleigh distribution with scale parameter $\sigma=\sqrt{2}$ when $v=1$.
\begin{thm}\Label{Zform}
For $v \ge 0$,
the following holds
\begin{eqnarray}
Z_{v}(\mu)
~=~ \left\{
\begin{array}{ll}
\Phi(\mu) &\hbox{ if }~~v=0 \\
1-(\sqrt{1-\Phi(\beta_{\mu,v})}\sqrt{1-\Phi_{\mu,v}(\beta_{\mu,v})} + I_{\mu,v}(\beta_{\mu,v}))^2 &\hbox{ if }~~0< v<1 \\
R_{\sqrt{2}}(\mu) &\hbox{ if }~~ v=1\\
1-(\sqrt{\Phi(\alpha_{\mu,v}) \Phi_{\mu,v}(\alpha_{\mu,v})} 
 + I_{\mu,v}(\infty) - I_{\mu,v}(\alpha_{\mu,v}))^2 &\hbox{ if }~~ v>1,
\end{array}
\right.&
\Label{Z}
\end{eqnarray}
where 
\begin{eqnarray}
\hspace{-2em}I_{\mu,v}(x)
&:=&\sqrt{\frac{2\sqrt{v}}{1+v}}e^{-\frac{\mu^2}{4(1+v)}} \Phi_{\frac{\mu}{1+v}, \frac{2v}{1+v}}\left(x\right),
\Label{Ix}\\
\hspace{-2em}I_{\mu,v}(\infty)
&:=&\lim_{x\to\infty} I_{\mu,v}(x)
=\sqrt{\frac{2\sqrt{v}}{1+v}}e^{-\frac{\mu^2}{4(1+v)}}.
\Label{Iinfty}
\end{eqnarray}
\end{thm}

\begin{proof}
Since the case of $v=0$ is trivial from the definition, we discuss the other cases.
The function $A_{\mu,v}$ defined in (\ref{A_1}), (\ref{A_2}) or (\ref{A_3}) is a continuous differentiable monotone increasing function satisfying $\Phi\le A\le1$.
Thus, we obtain 
\begin{eqnarray}
\sup_{A}{\cal F}\left(\frac{dA}{dx},\phi_{\mu,v}\right) 
&=&{\cal F}\left(\frac{dA_{\mu,v}}{dx},\phi_{\mu,v}\right) 
\Label{f-eq}
\end{eqnarray}
by Lemmas \ref{monotone}, \ref{lem.converse} and \ref{lem.converse2}, 
where $A:\R\to[0,1]$  in the left hand side runs over continuously differentiable monotone increasing functions satisfying $\Phi\le A\le1$.
From a direct calculation,
we can verify that 
\begin{eqnarray}
{\cal F}\left(\frac{dA_{\mu,v}}{dx},\phi_{\mu,v}\right) 
~=~ \left\{
\begin{array}{ll}
\sqrt{1-\Phi(\beta_{\mu,v})}\sqrt{1-\Phi_{\mu,v}(\beta_{\mu,v})} + I_{\mu,v}(\beta_{\mu,v}) &\hbox{ if }~~ 0<v<1 \\
1 &\hbox{ if }~~ v=1~\hbox{ and }~\mu\le0\\
e^{-\frac{\mu^2}{4\sigma^2}} &\hbox{ if }~~ v=1~\hbox{ and }~\mu>0\\
\sqrt{\Phi(\alpha_{\mu,v}) \Phi_{\mu,v}(\alpha_{\mu,v})} 
 + I_{\mu,v}(\infty) - I_{\mu,v}(\alpha_{\mu,v}) &\hbox{ if }~~ v>1.
\end{array}
\right.\nonumber
\end{eqnarray}
Therefore, the proof is completed.
\end{proof}

Then, we show some properties of the Rayleigh-normal distribution functions.
From Theorem \ref{Zform},
we can show a kind of symmetry of the family of the Rayleigh-normal distribution functions about the inversion of $v$ as follows.
These propositions are proven in Appendix \ref{app.properties}.
\begin{prp}\Label{Sym}
The following equation holds for $\mu\in\R$ and $v>0$:
\begin{eqnarray}
Z_v(\mu)=Z_{v^{-1}}\left(\frac{\mu}{\sqrt{v}}\right).
\Label{sym}
\end{eqnarray}
\end{prp}
By Proposition \ref{Sym},
the behavior of the Rayleigh-normal distribution function $Z_v$ for $v>1$ can be represented by that for $0< v<1$.

Next we show that the family of Rayleigh-normal distribution function includes the standard normal distribution function as its extreme case.
\begin{prp}\Label{F3lim}
The following equation holds for $\mu\in\R$:
\begin{eqnarray}
\lim_{v\to +0}Z_{v}(\mu)
~=~ \lim_{v\to \infty}Z_{v}(\sqrt{v}\mu)
~=~ \Phi\left(\mu\right).
\Label{Z-lim0}
\end{eqnarray}
\end{prp}
Thus, the Rayleigh-normal distribution function $Z_{v}$ is right continuous with respect to $v$ at $v=0$.

Finally, we give the following most basic property of the Rayleigh-normal distribution function.
\begin{prp}\Label{cum}
The Rayleigh-normal distribution function $Z_v$ is a cumulative distribution function for each $v\ge 0$.
\end{prp}
%
%
By Proposition \ref{cum}, 
the set of the functions $Z_v$ determines a family of probability distributions on $\R$ 
with one parameter $v\ge 0$.
We call the probability distribution determined by $Z_v$ the Rayleigh-normal distribution.
As shown in Theorem \ref{thm.gen},
the family of probability distribution functions can represent the optimal conversion rate in the second-order asymptotics.

\section{Conversions for Probability Distributions: Non-Asymptotic Setting}\Label{sec:non-asymptotic}

In Sections \ref{sec:non-asymptotic} and \ref{sec:asymptotic},
we focus on information-theoretic aspects for two kinds of conversions called deterministic conversion and majorization conversion for probability distributions.
Their roles in quantum information theory will be explained in Section \ref{sec:application}.

\subsection{Deterministic Conversion}\Label{sec:deterministic}

Let $\calP(\calX)$ be the set of all probability distributions on a finite set $\calX$.
For $P\in\calP(\calX)$ and a map $f:\mathcal{X}\to\mathcal{Y}$, 
we define the probability distribution $W_f(P)\in\calP(\calY)$ by 
\begin{eqnarray}
W_f(P)(y):=\sum_{x\in W^{-1}(y)}P(x).
\Label{eq:deterministic}
\end{eqnarray}
We call a map $W_f:\calP(\mathcal{X}) \to \calP(\mathcal{Y})$ defined in (\ref{eq:deterministic}) 
a {\it deterministic conversion}.
Here, the word 'deterministic' comes from the non-probabilistic property of $f$.

In order to treat the quality of conversion,
we introduce the fidelity (or the Bhattacharyya coefficient) $F$ between two probability distributions over the same discrete set $\mathcal{Y}$ as
\begin{eqnarray}
F(Q, Q'):=\sum_{y\in\mathcal{Y}}\sqrt{Q(y)}\sqrt{Q'(y)}.
\end{eqnarray}
Since this value $F(Q, Q')$ relates to the Hellinger distance $d_H$ as $d_H(Q, Q')=\sqrt{1-F(Q, Q')}$ \cite{Vaa98},
it represents how close two probability distributions $Q$ and $Q'$.
When a permissible accuracy $0<\tau<1$ is fixed,
we define the maximal conversion number $L$ of copies of $Q$ 
by deterministic conversions with the initial distribution $P$ as
\begin{eqnarray*}
L^{\cal D}(P, Q|\tau)
:=\max
\{L\in\N ~|~
\exists~ f:\mathcal{X}\to\mathcal{Y}^L,~
F(W_f(P), Q^L)\ge\nu
\}.
\end{eqnarray*}
One of the main topics of the paper is to analyze the above maximum conversion number by deterministic conversions.
When we define the maximal fidelity $F^{\cal D}$ from $P\in\calP(\calX)$ to $Q\in\calP(\calY)$ 
among deterministic conversions by
\begin{eqnarray}
F^{\cal D}(P\to Q)
&:=&\sup_{W:\calP(\calX) \to \calP(\calY)}
\{F(W(P), Q)~|~ W~ \text{is a deterministic conversion}\}\\
&=&\sup_{f:\mathcal{X}\to\mathcal{Y}}
F(W_f(P), Q),\nonumber
\end{eqnarray}
the maximum conversion number $L^{\cal D}$ is rewritten as
\begin{eqnarray}
L^{\cal D}(P, Q|\tau)
~=~ \max\{L\in\N~|~F^{\cal D}(P\to Q^L)\ge\tau\}.
\end{eqnarray}
We denote the maximum conversion number from $n$-i.i.d. $P^n$ to i.i.d of $Q$ with a permissible accuracy $0<\tau<1$ by deterministic conversions as
\begin{eqnarray*}
L^{\cal D}_n(P, Q|\tau):=L^{\cal D}(P^n, Q|\tau).
\end{eqnarray*}

\subsection{Majorization Conversion}\Label{sec:majorization}

In order to relax the condition for the deterministic conversion, 
we introduce the concept of a majorization conversion.
For a probability distribution $P$ on a finite set $\mathcal{X}$, 
let $P^{\downarrow}=\{P^{\downarrow}_i\}_{i=1}^{\infty}$ be a probability distribution on the set $\N$ of natural numbers 
where 
$|\mathcal{X}|$ represents the cardinality of the set $\mathcal{X}$ and $P^{\downarrow}_i$ is the $i$-th largest element of $\{P(x)\}_{x\in\mathcal{X}}$ 
for $1\le i\le|\mathcal{X}|$
and 
$0$ for $i>|\mathcal{X}|$.
When probability distributions $P$ and $Q$ satisfy $\sum_{i=1}^lP^{\downarrow}_i \le \sum_{i=1}^lQ^{\downarrow}_i$ for any $l\in\N$, it is said that $P$ is majorized by $Q$ and written as $P\prec Q$.
Here, we note that the sets where $P$ and $Q$ are defined do not necessarily coincide with each other, and the majorization relation is a partial order on a set of probability distributions on finite sets \cite{MO79,Arn86}. 
Then a map $W'$ from the set of probability distributions on $\mathcal{X}$ to that on $\mathcal{Y}$
is called a majorization conversion 
if it satisfies $P\prec W'(P)$ for any probability distribution $P$ on $\mathcal{X}$\footnote{
We note that a majorization conversion is a generalization of deterministic conversion that prohibits a probabilistic mixture such as a doubly stochastic map. 
This is because a doubly stochastic map is NOT a majorization conversion in general as follows.
The input distribution is majorized by the output distribution in a majorization conversion while the output distribution is majorized by the input distribution in a doubly stochastic map.
}.

Majorization conversions have an operational meaning in secret correlation manipulation \cite{CP02}.
In the setting,
two parties secretly have a copy of a random variable distributed according to a probability distribution $P$ in the beginning 
and wish to generate a random variable distributed according to another one $Q$ 
without leaking any information about the generated random variable to an adversary. 
When they are allowed to use unlimited public communication, 
they can succeed at the above task with certainty if and only if $P$ is majorized by $Q$.
Majorization conversions have an operational meaning also in quantum settings as we will see in Section \ref{sec:quantum}.

We give two important remarks on the majorization. 
The first one is that 
a deterministic conversion is a majorization conversion, i.e., 
a deterministic conversion by a map $W:\mathcal{X}\to\mathcal{Y}$ satisfies 
the majorization relation $P\prec W(P)$ for any probability distribution $P$ on a finite set $\mathcal{X}$.
The second one is that, when the support size of a probability distribution $P$ is less than or equal to $m$ and $U_m$ is the uniform distribution with support size $m$, we have $U_m \prec P$.
This fact is necessary in the analysis for the quantum operation called entanglement concentration which will be treated in Section \ref{sec:quantum}.

Here, we define the maximum conversion number $L$ of $Q^L$ which can be approximated from $P$ under a permissible accuracy $0<\tau<1$ among majorization conversions as
\begin{eqnarray}
L^{\cal M}(P, Q|\tau)
&:=&\max \{L\in\N~|
~\exists~ \hbox{majorization~conversion } W \hbox{ s.t. }
F(W(P), Q^L)\ge\tau\}\\
&=&\max \{L\in\N~|
~\exists~ \hbox{probability~distribution } P'\succ P \hbox{ s.t. }
F(P', Q^L)\ge\tau\}.
\end{eqnarray}
The equality between (27) and (28) can be shown as follows.
As explained in the beginning of this subsection,
a majorization conversion is given as a assigning a probability distribution that majorizes the original distribution.
Hence, a majorization conversion $W$ of (27) to attain the maximum in (28)
is given by assigning the probability distribution $P' \succ P$ to attain the maximum in (28).
A probability distribution $P' \succ P$ to attain the maximum in (27)
is given by $W(P)$ when the majorization conversion $W$ attains the maximum in (27).
This discussion shows the equality between (27) and (28).

To analyze the above maximum conversion number by majorization conversion is also one of the main topics in this paper beside to treat that by deterministic conversions.
When we introduce the maximum fidelity among  the majorization conversions as
\begin{eqnarray}
F^{\cal M}(P\to Q)
~:=~ \max\{F(P', Q)~|~ P' \hbox{~is~a~probability~distribution~on~ }\mathcal{Y}, P'\succ P\}
\Label{FM}
\end{eqnarray}
where $P$ and $Q$ are probability distributions on $\mathcal{X}$ and $\mathcal{Y}$ respectively\footnote{We note that (\ref{FM}) can be reduced to a convex optimization problem in the following way.
First, we reorder the entries of $Q$ in decreasing order.
Since we discuss the maximum of $F(P',Q)$, we can restrict $P' \in\{P':P'\succ P\}$ to a distribution whose entries are in decreasing order.
Because the set of such distributions $P'$ is convex and the fidelity $F$ is concave with respect to each component,
(\ref{FM}) can be regarded as a convex optimization problem.},
the maximum conversion number $L^{\cal M}$ is rewritten as
\begin{eqnarray}
L^{\cal M}(P, Q|\tau) =\max\{L\in\N~|~F^{\cal M}(P\to Q^L)\ge\tau\}.
\end{eqnarray}
We also denote the maximum conversion number from $n$-i.i.d. of $P$ to i.i.d. of $Q$ under a permissible accuracy $0<\tau<1$ by majorization conversions as
\begin{eqnarray*}
L^{\cal M}_n(P, Q|\tau):=L^{\cal M}(P^n, Q|\tau).
\end{eqnarray*}
Then,
since a deterministic conversion is a kind of majorization conversion,
the following relations are derived:
\begin{eqnarray}
F^{\cal M}(P\to Q) &\ge& F^{\cal D}(P\to Q),\Label{fidelity ineq} \\
L^{\cal M}(P, Q|\tau) &\ge& L^{\cal D}(P, Q|\tau), \\
L^{\cal M}_n(P, Q|\tau) &\ge& L^{\cal D}_n(P, Q|\tau).\Label{number ineq}
\end{eqnarray}
These inequalities play an essential role in the asymptotics of the maximum conversion numbers by those conversions.

Next, we prepare two propositions for discussions in latter parts.
The following lemma gives the optimal majorization conversion accuracy from a uniform distribution to an arbitrary distribution.
\begin{prp}\cite{VJN00}\Label{opt dil}
For a probability distribution $Q$ and a natural number $m$, let 
${\cal D}_m(Q)$ be defined as follows:
\begin{eqnarray}
{\cal D}_m(Q)(j)
~:=~
\left\{\begin{array}{ccc}
\frac{Q^{\downarrow}(j)}{\sum_{i=1}^m Q^{\downarrow}(i)}&\hbox{ if }&1\le j\le m\\
0&\hbox{ if }&m+1\le j.
\end{array}
\right.
\Label{opt dil prob}
\end{eqnarray}
Then, for the uniform distribution $U_m$ whose support size is $L$, 
the following equation holds:
\begin{eqnarray}\Label{Deq}
F^{\mathcal{M}}(U_m\to Q)
~=~ F({\cal D}_m(Q), Q^{\downarrow})
~=~ \sqrt{\sum_{i=1}^m Q^{\downarrow}(i)}.
\end{eqnarray}
\end{prp}
Proposition \ref{opt dil} is easily proven by using the Schwarz inequality.

The following lemma gives the optimal majorization conversion accuracy from an arbitrary distribution to a uniform distribution.
\begin{prp}\Label{opt con}
For a probability distribution $P$ and a natural number $m$, 
we define the following distribution ${\cal C}_m(P)$ 
on $\{1, \ldots, m\}$ as a distribution approximating the uniform distribution:
\begin{eqnarray}
{\cal C}_m(P)(j)
~:=~ \left\{
\begin{array}{lll}
P^{\downarrow}(j)&\hbox{ if }&1\le j\le J_{P,m}-1\\
\frac{\sum_{i=J_{P,m}}^{|\mathcal{X}|} P^{\downarrow}(i)}{m+1-J_{P,m}}&\hbox{ if }&~J_{P,m}\le j\le m
\end{array}
\right.\Label{PL}
\end{eqnarray}
where 
\begin{eqnarray}
J_{P,m}
&:=&
\left\{
\begin{array}{ll}
1&\hbox{if } P^{\downarrow}(1)\le \frac{1}{m}\\
\max\left\{j \in \{2, \ldots, m\} 
\left|\frac{\sum_{i=j}^{|\mathcal{X}|} P^{\downarrow}(i)}{m+1-j}<P^{\downarrow}(j-1)\right.\right\} &\hbox{otherwise.}
\end{array}\right.
\Label{J}
\end{eqnarray}
Then, 
$P\prec {\cal C}_m(P)$ and 
the following equation hold:
\begin{eqnarray}
\hspace{-1.5em}F^{\mathcal{M}}(P\to U_m)
&\hspace{-0.1em}=&F({\cal C}_m(P), U_m)\nonumber\\
&\hspace{-0.1em}=&\sqrt{\frac{1}{m}}\left(\sum_{j=1}^{J_{P,m}-1}\sqrt{P^{\downarrow}(j)}+\sqrt{(m+1-J_{P,m})\sum_{i=J_{P,m}}^{|\mathcal{X}|} P^{\downarrow}(i)}\right).
\end{eqnarray}
\end{prp}
The proof of Proposition \ref{opt con} is given in Appendix \ref{app.OPT}.

We give an intuitive explanation for why ${\cal C}_m(P)$ appropriately approximates the target uniform distribution. 
To well approximate the uniform distribution $U_m$, we first reorder the probability weights such that they are in decreasing order.
Then, we need to move the probability weights smaller than the $m$-th largest weight because such events do not contribute to the fidelity with the uniform distribution.
To increase the fidelity, it is better to add these leftover weights to 
events with smaller weights so that
the resultant distribution is closer to the uniform distribution.
The best way is the following. We find a suitable threshold event, whose weight is the $J_{P,m}$-th largest. Then, we move the above leftover weights to the events from the $J_{P,m}$-th largest weight to the $m$-th largest weight so that the resultant weights are uniform on this part. Due to the choice of $J_{P,m}$, this conversion is available by majorization conversion.


\section{Conversions for Probability Distributions: Asymptotic Setting} \Label{sec:asymptotic}

We will derive the asymptotic expansion formulas for $L^{\cal D}_n(P, Q|\tau)$ and $L^{\cal M}_n(P, Q|\tau)$
up to the second-order term $\sqrt{n}$, which are called the second-order asymptotic expansions in information theory. 
Since the first and second orders are $n$ and $\sqrt{n}$,
their coefficients are called the first-order rate and the second-order rate, respectively. 
To begin with, we note that the first-order asymptotics of maximum conversion numbers  are potentially done as follows:
\begin{eqnarray}\Label{first}
\lim_{n\to\infty}\frac{L^{\cal D}_n(P, Q|\tau)}{n}
~=~ \lim_{n\to\infty}\frac{L^{\cal M}_n(P, Q|\tau)}{n}
~=~ \frac{H(P)}{H(Q)},
\end{eqnarray}
where $H(P)$ is the entropy of $P$, i.e., $H(P):=-\sum_{x\in{\cal X}} P(x)\log_2P(x)$.
The logarithmic function always has the base $2$ through this paper and we denote it simply by $\log$ in the following.
Throughout this paper, we assume a probability distribution has non-zero entropy, or equivalently, the support size of the probability distribution is not $1$.
The equation between the left hand side and the right hand side is obtained from the results about the intrinsic randomness and the resolvability in conventional information theory \cite{Han03}.
Similarly, the second equation is obtained from the results about entanglement concentration and dilution in quantum information theory \cite{BBPS96}.

\subsection{Asymptotic Expansion Formula}\Label{subsec:Gen}
To describe the asymptotic expansion formula,
we introduce three parameters as
\begin{eqnarray}
V(P)&:=&\sum_{x\in{\cal X}} P(x)(-\log P(x)-H(P))^2,\\
D_{P,Q}&:=&\frac{H(Q)}{\sqrt{V(P)}},\\
C_{P,Q}&:=&\frac{H(P)}{V(P)}\left(\frac{H(Q)}{V(Q)}\right)^{-1}.
\end{eqnarray}
In particular,
we call $C_{P,Q}$ the {\it conversion characteristics} between probability distributions $P$ and $Q$ in the following
because the explicit form of the second-order conversion rate and its derivation in Theorem \ref{thm.gen} differ depending on $C_{P,Q}$.

If both $P$ and $Q$ are non-uniform distributions,
the symmetry of Rayleigh-normal distributions represented by Proposition \ref{Sym} 
yields the relation
\begin{eqnarray}
\frac{Z_{C_{P,Q}}^{-1}(1-\tau^2)}{D_{P,Q}}
&=& \left(\frac{H(P)}{H(Q)}\right)^{\frac{3}{2}}
\frac{Z_{C_{Q,P}}^{-1}(1-\tau^2)}{D_{Q,P}}
\Label{another}
\end{eqnarray}
by using the equation $C_{P,Q}^{-1}=C_{Q,P}$.
Note that the left hand side of (\ref{another}) is well-defined when $Q$ is uniform (i.e. $C_{P,Q}=0$) since $Z_0$ is defined as $\Phi$
while it cannot be defined when $P$ is uniform (i.e. $C_{P,Q}=\infty$).
On the other hand, the right hand side of (\ref{another}) is well-defined when $P$ is uniform (i.e. $C_{Q,P}=0$).
Using these quantities and the relation \eqref{another},
as the main theorem, we obtain the following asymptotic expansion of two maximum conversion numbers with an accuracy constraint $\tau$,
which enables us to highly accurately evaluate the quantity $L_n(P,Q|\tau)$
although its direct calculation is very hard for a large number $n$. 

\begin{thm}\Label{thm.gen}
Let $P$ and $Q$ be arbitrary probability distributions on finite sets.
Then, 
the following asymptotic expansion holds for an arbitrary $\tau\in(0,1)$:
\begin{eqnarray}
L^{\cal D}_n(P, Q|\tau)
~\cong~ L^{\cal M}_n(P, Q|\tau)
~= 
\left\{
\begin{array}{ll}
\frac{H(P)}{H(Q)}n + \frac{Z_{C_{P,Q}}^{-1}(1-\tau^2)}{D_{P,Q}}\sqrt{n} +o(\sqrt{n})
&\hbox{ if } V(P)\neq 0 \\
\frac{H(P)}{H(Q)}n + 
\left(\frac{H(P)}{H(Q)}\right)^{\frac{3}{2}}
\frac{Z_{C_{Q,P}}^{-1}(1-\tau^2)}{D_{Q,P}} \sqrt{n} +o(\sqrt{n})
&\hbox{ if } V(Q)\neq 0 ,
\end{array}
\right.
\Label{exp.gen}
\end{eqnarray}
where $\cong$ shows that the difference between the left and the right side terms 
is $o(\sqrt{n})$ at most.
\end{thm}

The graphs of the second-order rates of maximum conversion numbers are described as in Fig.  \ref{case3-1} for different conversion characteristics $C_{P,Q}$.
When $V(P)$ nor $V(Q)$ are not zero, 
both expressions are valid in the equation \eqref{exp.gen} due to the equation \eqref{another}.
Here, the conversion characteristics $C_{P,Q}$ corresponds to the parameter $v$ of the Rayleigh-normal distribution $Z_v$.
Even though we focus on the first case, i.e., $V(P)\neq 0$, 
the expression of the second-order rate is split to four cases depending on $C_{P,Q}$ as was shown in Theorem \ref{Zform}.
That is, we have five cases depending on the value of the conversion characteristics, totally.
In particular,
the case when $C_{P,Q}=1$ is singular since the second-order rate is non-negative even if the accuracy $\tau$ tends to $1$ unlike another parameter as shown in Fig. \ref{case3-1}.
On the other hand, the case when $C_{P,Q}=0$ is also singular in the sense that the second-order rate decays fastest as required accuracy gets higher.
Thus, the conversion characteristics $C_{P,Q}$ reflects the balance between $P$ and $Q$ and it can be regarded as a new kind of information quantity.
\begin{figure}[t]
 \begin{center}
 \hspace*{0em}\includegraphics[width=90mm, height=70mm]{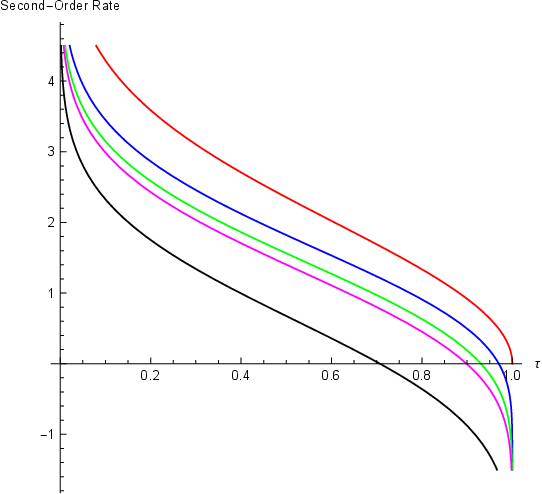}
 \end{center}
 \caption{
The second-order rates in (\ref{another}) of $L^{\cal D}_n(P, Q|\tau)$ behaves as above with respect to accuracy $\tau\in(0,1)$.
The black, purple, green, blue and red lines are displayed from the upper to the lower and  correspond to the cases when $C_{P,Q}=0$, $1/10$, $1/6$, $1/3$ and $1$ under $D_{P,Q}=1$.
The case $C_{P,Q}>1$
can be transformed to that with parameter $C_{P,Q}<1$
by Proposition \ref{Sym}.
Only when $C_{P,Q}=1$,  the second-order rate  is always non-negative and goes to $0$ when $\tau$ tends to $1$.
On the other hand, when $C_{P,Q}\ne 1$, the second-order rate goes to $-\infty$ when $\tau$ tends to $1$.
}
 \Label{case3-1}
\end{figure}

We give concrete forms of (\ref{exp.gen}) for three specific cases.
Here, let $U_m$ be the uniform distribution on the set $\{1, \ldots, m\}$.
When $P$ is the uniform distribution $U_m$, 
the second expression of (\ref{exp.gen}) coincides with the following concrete form by Theorem $\ref{Zform}$ with $v=0$:
\begin{eqnarray}
L^{\cal D}_n(U_m,Q|\tau)
~\cong~ L^{\cal M}_n(U_m,Q|\tau)
~\cong~ \frac{\log m}{H(Q)}n + \sqrt{\frac{V(Q)\log m}{H(Q)^3}}\Phi^{-1}(1-\tau^2)\sqrt{n}.
\Label{cla.dil}
\end{eqnarray}
When $Q$ is the uniform distribution $U_m$,
the first expression of (\ref{exp.gen}) coincides with the following concrete form by Theorem $\ref{Zform}$ with $v=0$:
\begin{eqnarray}
L^{\cal D}_n(P, U_m|\tau)
~\cong~ L^{\cal M}_n(P, U_m|\tau)
~\cong~ \frac{H(P)}{\log m}n + \frac{\sqrt{V(P)}}{\log m}\Phi^{-1}(1-\tau^2)\sqrt{n}.
\Label{cla.con}
\end{eqnarray}
When $Q$ is equal to $P$,
since $Z_{1}$ is the cumulative distribution function $R_{\sqrt{2}}$ of the Rayleigh distribution  by Theorem $\ref{Zform}$ with $v=1$ and $R_{\sqrt{2}}^{-1}(1-\tau^2)=\sqrt{8\log \tau^{-1}}$,
the formula (\ref{exp.gen}) is simplified to
\begin{eqnarray}
L^{\cal D}_n(P,P|\tau)
~\cong~ L^{\cal M}_n(P,P|\tau)
~\cong~ n + \frac{\sqrt{8V(P)\log \tau^{-1}}}{H(P)}\sqrt{n}.
\Label{cla.clone}
\end{eqnarray}
%


In the remaining part of this subsection,
we prepare a key lemma and give a proof of Theorem \ref{thm.gen}.
For this purpose, we introduce a notation 
\begin{eqnarray}
F_{P,Q}^i(b)
:=\lim_{n\to\infty} F^i\left(P^{n} \to Q^{\frac{H(P)}{H(Q)}n+b\sqrt{n}}\right)
 \Label{FPQ}
\end{eqnarray}
for $i={\cal D}$ or ${\cal M}$.
Then, (\ref{fidelity ineq}) implies that
\begin{eqnarray}
F_{P,Q}^{\cal D}(b) 
~\le~ F_{P,Q}^{\cal M}(b). \Label{H3}
\end{eqnarray}
We prepare the following lemma, which will be shown in the remaining subsections of this section.
\begin{lem}\Label{lem.gen}
Let $P$ and $Q$ be arbitrary probability distributions on finite sets.
Then, the following relations hold for any $b \in \R$:
\begin{eqnarray}
F_{P,Q}^{\cal D}(b) 
& \ge& 
\left\{
\begin{array}{ll}
\sqrt{1-Z_{C_{P,Q}}(bD_{P,Q})}
& \hbox{ if } V(P)\neq 0 \\
\sqrt{1-Z_{C_{Q,P}}\left(b\sqrt{\frac{H(Q)^3}{H(P)V(Q)}}\right)}
& \hbox{ if } V(Q)\neq 0 ,
\end{array}
\right.  \Label{H1}\\
F_{P,Q}^{\cal M}(b) &\le &
\left\{
\begin{array}{ll}
\sqrt{1-Z_{C_{P,Q}}(bD_{P,Q})}
& \hbox{ if } V(P)\neq 0 \\
\sqrt{1-Z_{C_{Q,P}}\left(b\sqrt{\frac{H(Q)^3}{H(P)V(Q)}}\right)}
& \hbox{ if } V(Q)\neq 0.
\end{array}
\right.
\Label{H2}
\end{eqnarray}
\end{lem}
When $V(Q)$ nor $V(P)$ is not zero,
the equation \eqref{another} guarantees that 
\begin{eqnarray*}
Z_{C_{P,Q}}(bD_{P,Q})
=Z_{C_{Q,P}}\left(b\sqrt{\frac{H(Q)^3}{H(P)V(Q)}}\right), 
\end{eqnarray*}
and thus, both expressions in the right hand sides of \eqref{H1} and \eqref{H2} give the same value. 

\noindent{\bf Proof of Theorem \ref{thm.gen}:}
Combining \eqref{H3}, \eqref{H1}, and \eqref{H2}, we obtain
\begin{eqnarray}
F_{P,Q}^{\cal D}(b)
~=~F_{P,Q}^{\cal M}(b)
~=~
\left\{
\begin{array}{ll}
\sqrt{1-Z_{C_{P,Q}}(bD_{P,Q})}
& \hbox{ if } V(P)\neq 0 \\
\sqrt{1-Z_{C_{Q,P}}\left(b\sqrt{\frac{H(Q)^3}{H(P)V(Q)}}\right)}
& \hbox{ if } V(Q)\neq 0 .
\end{array}
\right.
 \Label{2nd}
\end{eqnarray}
Theorem \ref{thm.gen} is obtained from \eqref{2nd} as follows.
For arbitrary $\tau\in(0,1)$,
we have 
\begin{eqnarray}
F_{P,Q}^{i~-1}(\tau)
=\frac{Z_{C_{P,Q}}^{-1}(1-\tau^2)}{D_{P,Q}}
\end{eqnarray}
for $i={\cal D}$ and ${\cal M}$
from (\ref{2nd}).
Since
\begin{eqnarray}
\lim_{n\to\infty} F^i\left(P^{n} \to Q^{\frac{H(P)}{H(Q)}n+\left(\frac{Z_{C_{P,Q}}^{-1}(1-\tau^2)}{D_{P,Q}}-\epsilon\right)\sqrt{n}}\right)
~=~F_{P,Q}^{i}(F_{P,Q}^{i~-1}(\tau)-\epsilon)
~>~ F_{P,Q}^{i}(F_{P,Q}^{i~-1}(\tau))
~=~ \tau
\end{eqnarray}
holds for $i={\cal D}$ and ${\cal M}$ and $\epsilon>0$,
$L^i_n(P, Q|\tau)$ is greater than or equal to the right side in (\ref{exp.gen}).
Similarly, 
\begin{eqnarray}
\lim_{n\to\infty} F^i\left(P^{n} \to Q^{\frac{H(P)}{H(Q)}n+\left(\frac{Z_{C_{P,Q}}^{-1}(1-\tau^2)}{D_{P,Q}}+\epsilon\right)\sqrt{n}}\right)
~=~ F_{P,Q}^{i}(F_{P,Q}^{i~-1}(\tau)+\epsilon)
~<~ F_{P,Q}^{i}(F_{P,Q}^{i~-1}(\tau))
~=~ \tau
\end{eqnarray}
holds.
Thus, $L^i_n(P, Q|\tau)$ for $i={\cal D}$ and ${\cal M}$ is less than or equal to the right hand side in (\ref{exp.gen}). 
Therefore, Theorem \ref{thm.gen} is obtained.
\endproof

From the above discussion, 
all we have to do is to show Lemma \ref{lem.gen}.
We separately prove the inequalities (\ref{H1}) and (\ref{H2}) 
for the uniform cases (i.e. $P$ or $Q$ is uniform) and the non-uniform case (i.e.
both $P$ and $Q$ are non-uniform).

\subsection{Limit of Tail Probability}

Before going to the proof of (\ref{H1}) and (\ref{H2}),
we prepare two useful lemmas to derive (\ref{H1}) and (\ref{H2}) in this subsection.
Let 
\begin{eqnarray}
S_n^Q(x)
:=\{1, 2, 3,  ..., \lceil 2^{H(Q)n+x\sqrt{n}} \rceil\}
\Label{S}
\end{eqnarray}
 and $S_n^Q(x, x'):= S_n^Q(x')\setminus S_n^Q(x)$.
Then, $ Q^{n\downarrow}(S_n^Q(x))$ represents the upper tail probability of $Q^{n\downarrow}$.

\begin{lem}\Label{lem.central}
For a non-uniform distribution $Q$ on a finite set and an arbitrary real number $x\in\R$,
\begin{eqnarray}
\displaystyle\lim_{n\to\infty} Q^{n\downarrow}(S_n^Q(x))
&=&\Phi\left(\frac{x}{\sqrt{V(Q)}}\right).
\Label{central1}
\end{eqnarray}
\end{lem}

The following lemma is a generalization of Lemma \ref{lem.central}.
\begin{lem}\Label{lem.central2}
Let $Q$ be a non-uniform distribution on a finite set.
For an arbitrary distribution $P$ on a finite set and arbitrary real numbers $x$ and $b\in\R$,
we have
\begin{eqnarray}
\displaystyle\lim_{n\to\infty} Q^{\frac{H(P)}{H(Q)}n+b\sqrt{n}\downarrow}(S_n^P(x))
&=&\Phi\left(\sqrt{\frac{H(Q)}{H(P)V(Q)}}(x-bH(Q))\right).
\Label{central2}
\end{eqnarray}
In particular, when $P$ is a non-uniform distribution as well as $Q$,
\begin{eqnarray}
\displaystyle\lim_{n\to\infty} Q^{\frac{H(P)}{H(Q)}n+b\sqrt{n}\downarrow}(S_n^P(x))
&=&\Phi_{P,Q,b}\left(\frac{x}{\sqrt{V(P)}}\right),
\Label{central3}
\end{eqnarray}
where $\Phi_{P,Q,b}:=\Phi_{b D_{P,Q}, C_{P,Q}}$.
\end{lem}
The proofs of Propositions \ref{lem.central} and \ref{lem.central2} are given in Appendix \ref{app.tail}.

\subsection{Uniform Distribution Cases}\Label{sec:uniform}

In this subsection, we prove Lemma \ref{lem.gen}, i.e., (\ref{H1}) and (\ref{H2}), 
for the uniform cases (i.e. $P$ or $Q$ is uniform).

\subsubsection{Source Distribution $P$ is Uniform}\Label{source}
We consider the case when 
$P$ is the uniform distribution $U_m$ and $Q$ is a non-uniform probability distribution on a finite set. 

\noindent{\it Sketch of proof of (49):}\quad
Firstly, we will show the existence of a sequence $\{f_n\}_{n=1}^{\infty}$ of
maps which satisfies 
\begin{eqnarray}
\liminf_{n\to\infty} F^{\mathcal{D}}(W_{f_n}(U_m^n), Q^{ \frac{\log m}{H(Q)}n+b\sqrt{n}\downarrow})
&\ge&
\sqrt{1-Z_{C_{Q,U_m}}\left(b\sqrt{\frac{H(Q)^3}{H(P)V(Q)}}\right)}.
\Label{Pineq1}
\end{eqnarray}
Since the definition of $F^{\cal D}$ implies 
\begin{eqnarray}
F^{\mathcal{D}}_{U_m,Q}(b)
~\ge~ F^{\mathcal{D}}(W_{f_n}(U_m^n), Q^{ \frac{\log m}{H(Q)}n+b\sqrt{n}\downarrow}),\Label{Pineq2}
\end{eqnarray}
the combination of (\ref{Pineq1}) and (\ref{Pineq2}) implies (\ref{H1}). 
Hence, the proof of (\ref{H1}) is reduced to the proof of (\ref{Pineq1}).

\noindent{\it Detailed proof of (49):}\quad
From the above sketch of proof, all we have to do is to show the existence of a sequence $\{W_n\}_{n=1}^{\infty}$ which satisfies (\ref{Pineq1}).
To show (\ref{Pineq1}), we prepare the following lemma.
\begin{lem}\Label{Wlem}
Let $S_1$ and $S_2$ be subsets of  the set $\N$ of natural numbers.
Suppose that $B=\{B(i)\}_{i\in S_1}$ and $C=\{C(j)\}_{j\in S_2}$ are non-negative real numbers in decreasing order  and  their sums of all components coincide with each other, i.e.  $\sum_{i\in S_1}B(i)=\sum_{i\in S_2}C(j)$.
Then, there exists a map $f:S_1\to S_2$ such that 
\begin{eqnarray}
B(i) &\le& W_f(C)(i) + \max_{j\in S_2}C(j)
\Label{Wmap2}
\end{eqnarray}
for any $i\in S_1$
where $W_f(C)(i):=\sum_{j\in f^{-1}(i)}C(j)$.
\end{lem}
Lemma \ref{Wlem} is proven in Appendix \ref{app.Exist}.

From Lemma \ref{Wlem},
there exists a map $f_n$ such that
\begin{eqnarray}
 Q^{ \frac{\log m}{H(Q)}n+b\sqrt{n}\downarrow}(j)
&\le& W_{f_n}(U_m^n)(j)+2^{-(\log m)n}
\Label{Wmap}
\end{eqnarray}
for $j\in S_n^{U_m}(0)=\{1,2,...,m^n\}$.
Since $S_n^{U_m}(-\gamma)\subset S_n^{U_m}(0)$ for an arbitrary $\gamma>0$,
the following inequalities are derived by the property (\ref{Wmap}):
\begin{eqnarray}
&&F^{\mathcal{D}}(W_{f_n}(U_m^n), Q^{ \frac{\log m}{H(Q)}n+b\sqrt{n}\downarrow})\nonumber\\
&\ge&\sum_{j\in S_n^{U_m}(-\gamma)}
\sqrt{W_{f_n}(U_m^{n})(j)} \sqrt{Q^{ \frac{\log m}{H(Q)}n+b\sqrt{n}\downarrow}(j)}
\nonumber\\
&\ge&\sum_{j\in S_n^{U_m}(-\gamma)}
\sqrt{\max\{Q^{ \frac{\log m}{H(Q)}n+b\sqrt{n}}(j)-2^{-(\log m)n}, 0\}} \sqrt{Q^{ \frac{\log m}{H(Q)}n+b\sqrt{n}\downarrow}(j)}
\nonumber\\
&\ge&\sum_{j\in S_n^{U_m}(-\gamma)}
\sqrt{Q^{ \frac{\log m}{H(Q)}n+b\sqrt{n}}(j)} \sqrt{Q^{ \frac{\log m}{H(Q)}n+b\sqrt{n}\downarrow}(j)}\Label{HH63}\\
&&-\sum_{j\in S_n^{U_m}(-\gamma)}
\sqrt{2^{-(\log m)n}} \sqrt{Q^{ \frac{\log m}{H(Q)}n+b\sqrt{n}\downarrow}(j)}
\nonumber\\
&\ge&
{Q^{ \frac{\log m}{H(Q)}n+b\sqrt{n}}(S_n^{U_m}(-\gamma))} 
-\sqrt{2^{-(\log m)n}} 
\sqrt{|S_n^{U_m}(-\gamma)|} 
\sqrt{Q^{ \frac{\log m}{H(Q)}n+b\sqrt{n}\downarrow}(S_n^{U_m}(-\gamma))}\Label{sch1}\\
&\ge&
{Q^{ \frac{\log m}{H(Q)}n+b\sqrt{n}}(S_n^{U_m}(-\gamma))} 
-\sqrt{2^{-(\log m)n}} 
\sqrt{2^{(\log m)n-\gamma\sqrt{n}}}\nonumber\\
&=&
{Q^{ \frac{\log m}{H(Q)}n+b\sqrt{n}}(S_n^{U_m}(-\gamma))} 
- \sqrt{2^{-\gamma\sqrt{n}}},
\Label{60}
\end{eqnarray}
where
 (\ref{HH63}) follows from $\sqrt{x-y} \ge \sqrt{x}- \sqrt{y}$ for any $x\ge y \ge 0$
 and the inequality (\ref{sch1}) is obtained by the Schwarz inequality.
Since the second term in (\ref{60}) goes to $0$ as $n$ tends to $\infty$,
\begin{eqnarray}
\liminf_{n\to\infty} F^{\mathcal{D}}(U_m^n\to Q^{ \frac{\log m}{H(Q)}n+b\sqrt{n}})
&\ge&
\lim_{\gamma\to0}\liminf_{n\to\infty} {Q^{ \frac{\log m}{H(Q)}n+b\sqrt{n}}(S_n^{U_m}(-\gamma))} \nonumber\\
&=&
\lim_{\gamma\to0}\Phi\left(\sqrt{\frac{H(Q)}{V(Q)\log m}}(\gamma-bH(Q))\right) \Label{philim1}\\
&=&\Phi\left(-\sqrt{\frac{H(Q)^3}{V(Q)\log m}}b\right)\nonumber\\
&=&\sqrt{1-Z_{C_{Q,U_m}}\left(b\sqrt{\frac{H(Q)^3}{H(P)V(Q)}}\right)},\Label{Z11}
\end{eqnarray}
where we used Lemma \ref{lem.central2} in (\ref{philim1}) and the definition $Z_{C_{Q,U_m}}=Z_0=\Phi$ in (\ref{Z11}).

\noindent{\it Proof of (\ref{H2})}:\quad
\begin{eqnarray}
\lim_{n\to\infty} F^{\cal M}\left(U_m^n\to Q^{ \frac{\log m}{H(Q)}n+b\sqrt{n}}\right)^2
&=&\lim_{n\to\infty} F^{\cal M}\left(U_2^{(\log m)n}\to Q^{ \frac{(\log m)n}{H(Q)}+\frac{b}{\sqrt{\log m}}\sqrt{(\log m)n}}\right)^2\Label{1}\\
&=&\lim_{k\to\infty} F^{\cal M}\left(U_2^{H(Q)k-\sqrt{\frac{H(Q)^3}{\log m}}b\sqrt{k}+o(\sqrt{k})}\to Q^{k}\right)^2
\Label{2}\\
&=&\lim_{k\to\infty} Q^{k\downarrow}\left(S^Q_k\left(-\sqrt{\frac{H(Q)^3}{\log m}}b\right)\right)
\Label{3}\\
&=&\Phi\left(-\sqrt{\frac{H(Q)^3}{V(Q)\log m}}b\right),\Label{4}\\
&=&\sqrt{1-Z_{C_{Q,P}}\left(b\sqrt{\frac{H(Q)^3}{H(P)V(Q)}}\right)}\Label{5}
\end{eqnarray}
where we used Proposition \ref{opt dil} in (\ref{3}) and Lemma \ref{lem.central2} in (\ref{4}).
In (\ref{2}), we replace the exponent of $Q$ with $k$  and represent the exponent $(\log m)n$ of $U_2$ by $k$.


~

\subsubsection{Target Distribution $Q$ is Uniform}\Label{target}
We consider the case when 
$Q$ is the uniform distribution $U_m$ and $P$ is a non-uniform probability distribution on a finite set. 

\noindent{\it Sketch of proof of (49):}\quad
We will first construct a sequence $\{P'_{n}\}_{n=1}^{\infty}$ of probability distributions such that
\begin{eqnarray}
\liminf_{n\to\infty} F(P'_{n}, U_m^{ \frac{H(P)}{\log m}n+b\sqrt{n}\downarrow})
&\ge&
\sqrt{1-Z_{C_{P,U_m}}\left(bD_{P,U_m}\right)}.
\Label{Qineq1}
\end{eqnarray}
Then, we will show the existence of a sequence $\{f_n\}_{n=1}^{\infty}$ of maps which satisfies 
\begin{eqnarray}
\liminf_{n\to\infty} F^{\mathcal{D}}(W_{f_n}(P^{n\downarrow}), U_m^{ \frac{H(P)}{\log m}n+b\sqrt{n}\downarrow})
&\ge&
\liminf_{n\to\infty} F(P'_{n}, U_m^{ \frac{H(P)}{\log m}n+b\sqrt{n}\downarrow}).
\Label{Qineq2}
\end{eqnarray}
Since we have the following inequality from the definition
\begin{eqnarray}
F^{\mathcal{D}}_{P,U_m}(b)
&\ge& F^{\mathcal{D}}(W_{f_n}(P^n), U_m^{ \frac{H(P)}{\log m}n+b\sqrt{n}\downarrow}),\Label{Qineq3}
\end{eqnarray}
the inequality (\ref{H1}) is derived by (\ref{Qineq1}), (\ref{Qineq2}) and (\ref{Qineq3}).

\noindent{\it Detailed proof of (\ref{H1}):}\quad
From the above sketch of proof, all we have to do  is to show (\ref{Qineq1}) and (\ref{Qineq2}).
We first construct a sequence $\{P'_{n}\}_{n=1}^{\infty}$ of probability distributions which satisfies (\ref{Qineq1}).
For an arbitrary $\epsilon>0$, 
we define a probability distribution $P'_{n}$ 
satisfying that
\begin{eqnarray}
P'_{n}(j)
&=&P^{n\downarrow}(S_n^P(b\log m+\epsilon, \infty)) U_m^{\frac{H(P)}{\log m}n+b\sqrt{n}}(j)\\
&=&P^{n\downarrow}(S_n^P(b\log m+\epsilon, \infty)) m^{-(\frac{H(P)}{\log m}n+b\sqrt{n})}
\Label{newprob-1}
\end{eqnarray}
for any $j\in S_n^P(b\log m)$.
Here, there is no constraint for $P'_{n}(j)$ with $j\in\N\setminus S_n^P(b\log m)$ 
as long as $P'_{n}$ is a probability distribution.
Note that $P'_{n}(j)$ is uniform on $S_n^P(b\log m)$. 
Then, we obtain
\begin{eqnarray}
\sum_{j \in S_n^P(b\log m)}P'_{n}(j)
~= \sum_{k \in S_n^P(b\log m+\epsilon, \infty)} P^{n\downarrow}(k).
\Label{assum1}
\end{eqnarray}

Then, the constructed sequence $\{P'_n\}_{n=1}^{\infty}$ satisfies (\ref{Qineq1}) as follows:
\begin{eqnarray}
&&\liminf_{n\to\infty} F(P'_{n}, U_m^{ \frac{H(P)}{\log m}n+b\sqrt{n}\downarrow})\\
&\ge&\liminf_{n\to\infty} 
\sum_{j\in S_n^P(b\log m)}
\sqrt{P'_{n}(j)}\sqrt{U_m^{ \frac{H(P)}{\log m}n+b\sqrt{n}\downarrow}(j)}
\nonumber\\
&=&\liminf_{n\to\infty} 
\sqrt{P^{n\downarrow}(S_n^P(b\log m+\epsilon, \infty))}
\sqrt{U_m^{ \frac{H(P)}{\log m}n+b\sqrt{n}\downarrow}(S_n^P(b\log m))}
\Label{H6-3}\\
&=&\liminf_{n\to\infty} 
\sqrt{P^{n\downarrow}(S_n^P(b\log m+\epsilon, \infty))}\\
&=&
\sqrt{1-\Phi\left(\frac{b\log m+\epsilon}{\sqrt{V(P)}}\right)} \Label{lim30}\\
&\overset{\epsilon\to0}{\to}&
\sqrt{1-\Phi\left(\frac{b\log m}{\sqrt{V(P)}}\right)} \\
&=&\sqrt{1-Z_{C_{P,U_m}}(bD_{P,U_m})},
\Label{lim31}
\end{eqnarray}
where (\ref{lim30}) follows from Lemma \ref{central1} and  (\ref{lim31}) follows from the definition $Z_{C_{P,U_m}}=Z_0=\Phi$.

Then, we will show the existence of a sequence $\{f_{n}\}_{n=1}^{\infty}$ of maps which satisfies (\ref{Qineq2}). 
From Lemma \ref{Wlem} and (\ref{assum1}), we choose a  map $f_{n}$ 
such that 
\begin{eqnarray}
&f_{n}(S_n^P(b\log m+\epsilon, \infty)) ~\subset~ S_n^P(b\log m),&
\Label{W2-1}\\
&P'_{n}(j)
~\le~ 
W_{f_n}(P^{n\downarrow})(j) 
+\max_{k\in S_n^P(b\log m+\epsilon, \infty)}P^{n\downarrow}(k)&
\Label{W1-1}
\end{eqnarray}
for any $j\in S_n^P(b\log m)$.
Since
\begin{eqnarray}
\max_{k\in S_n^P(b\log m+\epsilon, \infty)}P^{n\downarrow}(k)
~\le~ \min_{k\in S_n^P(b\log m+\epsilon)}P^{n\downarrow}(k),
\end{eqnarray}
we can evaluate as 
\begin{eqnarray}
W_{f_n}(P^{n\downarrow})(j) 
~\ge~
P'_{n}(j) - \min_{k\in S_n^P(b\log m+\epsilon)}P^{n\downarrow}(k).
\end{eqnarray}
Then, we have the following inequalities:
\begin{eqnarray}
&&F(W_{f_n}(P^{n\downarrow}), U_m^{ \frac{H(P)}{\log m}n+b\sqrt{n}\downarrow})\nonumber\\
&\ge&\sum_{j\in S_n^P(b\log m)}
\sqrt{W_{f_n}(P^{n\downarrow})(j)} \sqrt{U_m^{ \frac{H(P)}{\log m}n+b\sqrt{n}\downarrow}(j)}
\nonumber\\
&\ge&\sum_{j\in S_n^P(b\log m)}
\sqrt{\max\{P'_{n}(j)-\min_{k\in S_n^P(b\log m+\epsilon)}P^{n\downarrow}(k),0\}} \sqrt{U_m^{ \frac{H(P)}{\log m}n+b\sqrt{n}\downarrow}(j)}
\nonumber\\
&\ge&
\sum_{j\in S_n^P(b\log m)}
\sqrt{P'_{n}(j)}\sqrt{U_m^{ \frac{H(P)}{\log m}n+b\sqrt{n}\downarrow}(j)}
\Label{H6-2}\\
&&-
\sum_{j\in S_n^P(b\log m)}
\sqrt{\min_{k\in S_n^P(b\log m+\epsilon)}P^{n\downarrow}(k)} 
\sqrt{U_m^{ \frac{H(P)}{\log m}n+b\sqrt{n}\downarrow}(j)}\nonumber\\
&=&
F\left(P'_{n}, U_m^{ \frac{H(P)}{\log m}n+b\sqrt{n}\downarrow}\right)
-
\sum_{j\in S_n^P(b\log m)}
\sqrt{\min_{k\in S_n^P(b\log m+\epsilon)}P^{n\downarrow}(k)} 
\sqrt{U_m^{ \frac{H(P)}{\log m}n+b\sqrt{n}\downarrow}(j)}\\
&\ge&
F\left(P'_{n}, U_m^{ \frac{H(P)}{\log m}n+b\sqrt{n}\downarrow}\right)
-
\sqrt{\min_{k\in S_n^P(b\log m+\epsilon)}P^{n\downarrow}(k)} 
\sqrt{|S_n^P(b\log m)|},
\Label{HH90}
\end{eqnarray}
where (\ref{H6-2}) follows from $\sqrt{x-y} \ge \sqrt{x}- \sqrt{y}$ for any $x\ge y \ge 0$
and (\ref{HH90}) follows from the following inequality:
\begin{eqnarray*}
\sum_{j\in S_n^P(b\log m)} \sqrt{U_m^{ \frac{H(P)}{\log m}n+b\sqrt{n}\downarrow}(j)}
&\le&\sum_{j\in S_n^P(b\log m)} \sqrt{U_m^{ \frac{H(P)}{\log m}n+b\sqrt{n}\downarrow}(1)}\\
&=& |S_n^P(b\log m)| \sqrt{m^{-(\frac{H(P)}{\log m}n+b\sqrt{n})}}\\
&=& |S_n^P(b\log m)| \sqrt{2^{-(H(P)n+b\log m\sqrt{n})}}\\
&=& \sqrt{|S_n^P(b\log m)|}.
\end{eqnarray*}

To show (\ref{Qineq2}),
it is enough to show that the second term in (\ref{HH90}) goes to $0$ as $n$ goes to infinity.
To evaluate the second term in (\ref{HH90}),
we prepare the following lemma.
\begin{lem}\Label{alpha}
Let $P$ be a non-uniform distribution and $A$ be a continuous differentiable monotone increasing function satisfying $\Phi\le A\le1$.
When we set functions $y_{P,A}(x):\R\to\R$ and $\alpha_{n}^P(x)$ as 
\begin{eqnarray}
y_{P,A}(x)&:=&\sqrt{V(P)}\Phi^{-1}\left(A\left(\frac{x}{\sqrt{V(P)}}\right)\right),
\Label{y}\\
\alpha_{n}^P(x)
&:=&\min_{k \in S_n^P(y_{P,A}(x))} P^{n\downarrow}(k)
=P^{n\downarrow}(\lceil 2^{H(P)n+y_{P,A}(x)\sqrt{n}} \rceil),
\Label{def-alpha}
\end{eqnarray}
we have the following for $\epsilon>0$
\begin{eqnarray}
\Label{H10}
\alpha_{n}^P(x+\epsilon) |S_n^P(x)|
~\le~ 2^{-\epsilon\sqrt{n }}.
\end{eqnarray}
\end{lem}
\begin{proof}
The definition of $\alpha_n^P$ implies that
\begin{align}
\alpha_{n}(x+\epsilon)
~\le~ 2^{-(n H(P)+\sqrt{n}y_{P,A}(x+\epsilon))}.
\end{align}
Since 
$|S_n^P(x)|
\le 2^{n H(P)+ \sqrt{n}x}$,
we obtain 
\begin{eqnarray}
\alpha_{n}(x+\epsilon) |S_n^P(x)|
~\le~ 2^{-\sqrt{n}(y_{P,A}(x+\epsilon)-x)}.
\Label{HH96}
\end{eqnarray}
Since $A \ge \Phi$, we have
$x \le y_{P,A}(x) $.
Hence, we have
\begin{eqnarray}
2^{-\sqrt{n}(y_{P,A}(x+\epsilon)-x)}
~\le~ 2^{-\sqrt{n}((x+\epsilon)-x)}
= 2^{-\epsilon\sqrt{n}}.
\Label{HH97}
\end{eqnarray}
The inequalities (\ref{HH96}) and (\ref{HH97}) derives (\ref{H10}).
\end{proof}
By Lemma \ref{alpha} with $A=\Phi$,
the second term in (\ref{HH90}) is evaluated as 
\begin{eqnarray}
\sqrt{\min_{k\in S_n^P(b\log m+\epsilon)}P^{n\downarrow}(k)} 
\sqrt{|S_n^P(b\log m)|}
=\sqrt{\alpha_n^P(b\log m+\epsilon)
|S_n^P(b\log m)|}
\le \sqrt{2^{-\epsilon\sqrt{n}}}
\end{eqnarray}
and thus,
it goes to zero as $n$ goes to infinity.
%

\noindent{\it Sketch of proof of (50):}\quad
%
We introduce a notation for a real number $y$ as 
\begin{eqnarray}
\{P^{\downarrow}>y\}:=\{i\in\N~|~P^{\downarrow}(i)> y\}.
\end{eqnarray} 
Then, we will show the following inequality for $m_n=m^{\frac{H(P)}{\log m}n+b\sqrt{n}}$ and $m'_n=m^{\frac{H(P)}{\log m}n+(b-\lambda)\sqrt{n}}$ with $b\in\R$ and $\lambda>0$:
\begin{eqnarray}
F^{\cal M}(P^n\to U_m^{\frac{H(P)}{\log m}n+b\sqrt{n}})
&=&F^{\cal M}(P^n\to U_{m_n})\nonumber\\
&\le& \sqrt{\frac{|\{P^{n \downarrow}> 1/{m'_n}\}|-1}{m_n}}\sqrt{P^{n \downarrow}(\{P^{n \downarrow}>{1}/{m'_n}\})}\Label{finitelem3}\\
&&\hspace{0em}+\sqrt{1+\frac{1-|\{P^{n \downarrow}>{1}/{m'_n}\}|}{m_n}} \sqrt{1-P^{n \downarrow}(\{P^{n \downarrow}>{1}/{m'_n}\})}.
\nonumber
\end{eqnarray}
Moreover, we will show that
\begin{eqnarray}
\lim_{n\to\infty}\frac{|\{P^{n \downarrow}>{1}/{m'_n}\}|}{m_n} &=&0,\Label{lim1}\\
\lim_{\lambda\to+0}\lim_{n\to\infty} P^{n \downarrow}(\{P^{n \downarrow}>{1}/{m'_n}\}) &=&\Phi\left(\frac{b\log{m}}{\sqrt{V(P)}}\right). \Label{lim2}
\end{eqnarray}
Then, from (\ref{finitelem3}), (\ref{lim1}) and (\ref{lim2}),
we obtain (\ref{H2}) as follows:
\begin{eqnarray}
\lim_{n\to\infty}F^{\cal M}(P^n\to U_m^{\frac{H(P)}{\log m}n+b\sqrt{n}})
=\sqrt{1-\Phi\left(\frac{b\log{m}}{\sqrt{V(P)}}\right)}
=\sqrt{1-Z_{C_{P,U_m}}(bD_{P,U_m})},
\end{eqnarray}
where the last equality follows from the definition $Z_{C_{P,U_m}}=Z_0=\Phi$.

\noindent{\it Detailed proof of (50):}\quad
From the above sketch of proof,
all we have to do is to show (\ref{finitelem3}), (\ref{lim1}) and (\ref{lim2}).
First, we show (\ref{finitelem3}).
For an arbitrary positive integer ${L}$, 
we define as  $I_{P,m}:=|\{P>1/L\}|$.
Here, we show $I_{P,m}<J_{P,m}$.
To do so,
we assume that $I_{P,m}\ge J_{P,m}$ and derive a contradiction in the following.
Since 
\begin{eqnarray}
P^{\downarrow}(J_{P,m})
&\le& \frac{\sum_{i=J_{P,m}+1}^{|\mathcal{X}|} P^{\downarrow}(i)}{L-J_{P,m}}
\nonumber
\end{eqnarray}
holds by the definition of $J_{P,m}$,
we have
\begin{eqnarray}
P^{\downarrow}(J_{P,m})
&\le& \frac{\sum_{i=J_{P,m}}^{|\mathcal{X}|} P^{\downarrow}(i)}{m+1-J_{P,m}}
~=~ {\cal C}_m(P)(j)
\Label{Cineq}
\end{eqnarray}
for $J_{P,m}\le j \le L$.
Since we assume that $I_{P,m}\ge J_{P,m}$,
$P^{\downarrow}(J_{P,m})> L^{-1}$ holds by the definition of $I_{P,m}$
and it follows that all components of ${\cal C}_m(P)$ are strictly greater than  $L^{-1}$ by (\ref{Cineq}).
Then, ${\cal C}_m(P)$ cannot be a probability distribution because the total sum of its components is greater than $1$
and this is a contradiction.
For arbitrary positive integers $m\ge m'$,  
the inequalities $I_{P,m'}\le I_{P,m}<J_{P,m}$ hold by the definition of $I_{P,m}$.
Then, the following inequality holds:
\begin{eqnarray}
F^{\cal M}(P\to U_m)
&=&\sqrt{\frac{1}{m}}\left(\sum_{j=1}^{J_{P,m}-1}\sqrt{P^{\downarrow}(j)}+\sqrt{(m+1-J_{P,m})\sum_{i=J_{P,m}}^{|\mathcal{X}|} P^{\downarrow}(i)}\right)\Label{103}\\
&=&\sqrt{\frac{1}{m}}\left(\sum_{j=1}^{I_{P,m'}-1}\sqrt{P^{\downarrow}(j)} + \sum_{j=I_{P,m'}}^{J_{P,m}-1}\sqrt{P^{\downarrow}(j)} + \sqrt{(m+1-J_{P,m})\sum_{i=J_{P,m}}^{|\mathcal{X}|} P^{\downarrow}(i)}\right)\nonumber\\
&\le&\sqrt{\frac{1}{m}}\left(\sqrt{I_{P,m'}-1}\sqrt{\sum_{j=1}^{I_{P,m'}-1}P^{\downarrow}(j)} 
+\sqrt{m+1-I_{P,m'}} \sqrt{\sum_{i=I_{P,m'}}^{|\mathcal{X}|} P^{\downarrow}(i)}\right)\Label{104}\\
&=&\sqrt{\frac{1}{m}} \Big(\sqrt{|\{P^{\downarrow}> 1/{m'}\}|-1}\sqrt{P^{\downarrow}(\{P^{\downarrow}>{1}/{m'}\})}\nonumber\\
&&\hspace{2em}+\sqrt{m+1-|\{P^{\downarrow}>{1}/{m'}\}|} \sqrt{1-P^{\downarrow}(\{P^{\downarrow}>{1}/{m'}\})}\Big)\nonumber\\
&=& \sqrt{\frac{|\{P^{\downarrow}> 1/{m'}\}|-1}{m}}\sqrt{P^{\downarrow}(\{P^{\downarrow}>{1}/{m'}\})}\Label{spectrum}\\
&&\hspace{1em}+\sqrt{1+\frac{1-|\{P^{\downarrow}>{1}/{m'}\}|}{m}} \sqrt{1-P^{\downarrow}(\{P^{\downarrow}>{1}/{m'}\})},
\nonumber
\end{eqnarray}
where the equality (\ref{103}) is due to Proposition \ref{opt con} and the inequality (\ref{104}) is obtained by the Schwarz inequality.
Thus, we obtain (\ref{finitelem3}) by substituting $P=P^n$, $m=m_n$ and $m'=m'_n$ into (\ref{spectrum}).

Second, we  show (\ref{lim1}).
If $|\{P^{n \downarrow}>{1}/{m'_n}\}| > m'_n$,
the total sum of components $P^{n \downarrow}$ is greater than $1$ and it is a contradiction to the property of the probability distribution $P^{n \downarrow}$.
Thus, $|\{P^{n \downarrow}>{1}/{m'_n}\}| \le m'_n$ holds.
Then, (\ref{lim1}) follows as 
\begin{eqnarray}
\lim_{n\to\infty}\frac{|\{P^{n \downarrow}>{1}/{m'_n}\}|}{m_n} 
&\le& \lim_{n\to\infty}\frac{m'_n}{m_n}
~=~ \lim_{n\to\infty}m^{-\lambda\sqrt{n}}
~=~ 0.
\end{eqnarray}

Third, we show (\ref{lim2}).
Here, we have
\begin{eqnarray}
\displaystyle\lim_{n\to\infty} P^{n\downarrow}(\{P^{n \downarrow}>{1}/{m'_n}\})
&=&\displaystyle\lim_{n\to\infty} P^{n\downarrow}(\{P^{n \downarrow}>m^{-\frac{H(P)}{\log m}n-(b-\lambda)\sqrt{n}}\})\nonumber\\
&=&\displaystyle\lim_{n\to\infty} P^{n\downarrow}(\{P^{n \downarrow}>2^{-H(P)n-(b-\lambda)\log{m}\sqrt{n}}\})\nonumber\\
&=&\Phi\left(\frac{(b-\lambda)\log{m}}{\sqrt{V(P)}}\right)\Label{lamda-lim}
\end{eqnarray}
where the last equality will be proven in (\ref{central1'}) of Proof of Proposition \ref{lem.central}.
Thus, we obtain (\ref{lim2}) from (\ref{lamda-lim}).

\begin{rem}
For probability distributions $P$ and $Q$ on finite sets $\mathcal{X}$ and $\mathcal{Y}$, 
we have discussed the approximate conversion problem from the i.i.d. of $P$ to that of $Q$.
In particular, when $P$ or $Q$ is a uniform distribution, 
the problems have been well-known as the resolvability problem and the intrinsic randomness problem respectively \cite{Han03}.
Hayashi \cite{Hay08} treated the intrinsic randomness and Nomura and Han \cite{NH13} treated the resolvability besides the intrinsic randomness in the framework of the second-order asymptotics.
Their formulation is different from our formulation
because their analyses are based on the total variation distance not on the Hellinger distance,
which has the one-to-one correspondence to the fidelity.
Hence, our results of these special cases are not contained in their results.
On the other hand,
Tomamichel and Hayashi \cite{TH12} considered randomness extraction against quantum side information in the second-order asymptotics and adopted the fidelity to measure accuracy of the operations.
Since the intrinsic randomness in this paper is regarded as randomness extraction without quantum side information in \cite{TH12},
(\ref{cla.con}) can be directly obtained from Lemma $16$ in \cite{TH12}.

\end{rem}

\subsection{Non-Uniform Distribution Case}

In this subsection, we prove Lemma \ref{lem.gen}, i.e. (\ref{H1}) and (\ref{H2}), 
for the non-uniform case (i.e.
both $P$ and $Q$ are non-uniform).
For a preparation of our proof, depending on distributions $P$, $Q$ and a real number $b$,
we choose the probability density function of the normal distribution 
$\phi_{P, Q, b}:=\phi_{b D_{P,Q}, C_{P,Q}}$.
Then,
note that  the right hand side of (\ref{H1}) has another form as
\begin{eqnarray}
\sqrt{1-Z_{C_{P,Q}}(bD_{P,Q})}
=\sup_{A}{\cal F}\left(\frac{dA}{dx}, \phi_{P, Q, b}\right)
\Label{d-c5}
\end{eqnarray}
by the definition of the Rayleigh-normal distribution,
where supremum is taken over the functions satisfying the conditions in Definition \ref{rn}. 
Instead of the left hand side of (\ref{d-c5}),
we evaluate the right hand side of (\ref{d-c5}) in the proofs of (\ref{H1}) and (\ref{H2}).

\subsubsection{Direct Part} \Label{sec:reduction.dir}
In this subsection,
we prove (\ref{H1}) for the non-uniform case.
\par
\noindent{\it Sketch of proof of (49):}\quad
To prove (\ref{H1}),
it is enough to show 
\begin{align}
F_{P,Q}^{\cal D}(b) 
~\ge~ {\cal F}\left(\frac{dA}{dx}, \phi_{P, Q, b}\right) - \epsilon
\Label{ineq109}
\end{align}
for an arbitrary continuous differentiable monotone increasing function $A$ satisfying $\Phi\le A\le1$ and an arbitrary $\epsilon>0$
because of (\ref{d-c5}).

Given a continuous differentiable monotone increasing function $A$ satisfying $\Phi \le A \le 1$,
we will first construct a sequence $\{P'_{n,I}\}_{n=1}^{\infty}$ of probability distributions  for each $I\in\N$ such that
\begin{eqnarray}
\liminf_{n\to\infty} F(P'_{n,I}, Q^{ \frac{H(P)}{H(Q)}n+b\sqrt{n}\downarrow})
&\ge&
{\cal F}\left(\frac{dA}{dx}, \phi_{P, Q, b}\right) - \epsilon.
\Label{Gineq1}
\end{eqnarray}
Then, we will show the existence of a sequence $\{f_n\}_{n=1}^{\infty}$ of maps which satisfies 
\begin{eqnarray}
\liminf_{n\to\infty} F^{\mathcal{D}}(W_{f_n}(P^{n\downarrow}), Q^{ \frac{H(P)}{H(Q)}n+b\sqrt{n}\downarrow})
&\ge&
\liminf_{n\to\infty} F(P'_{n,I}, Q^{ \frac{H(P)}{H(Q)}n+b\sqrt{n}\downarrow}).
\Label{Gineq2}
\end{eqnarray}
Since we have the inequality from the definition
\begin{eqnarray}
F_{P,Q}^{\cal D}(b) 
\ge F^{\mathcal{D}}(W_{f_n}(P^n), Q^{ \frac{H(P)}{H(Q)}n+b\sqrt{n}\downarrow}),\Label{Gineq3}
\end{eqnarray}
the inequality (\ref{ineq109}) is derived by (\ref{Gineq1}), (\ref{Gineq2}) and (\ref{Gineq3}).

\noindent{\it Detailed proof of (49):}\quad
From the above sketch of proof, all we have to do  is to show (\ref{Gineq1}) and (\ref{Gineq2}).
%
We first construct a sequence $\{P'_{n,I}\}_{n=1}^{\infty}$ of probability distributions which satisfies (\ref{Gineq1}).
For an arbitrary $\epsilon>0$, we choose $\lambda>0$ which satisfies
\begin{eqnarray}
\int_{(-\infty,-\lambda)\cup(\lambda,\infty)}\sqrt{\frac{dA}{dx}(x)}\sqrt{\phi_{P,Q,b}(x)}dx
&\le & \epsilon.
\Label{bishou0}
\end{eqnarray}
For $I\in\N$, $0\le i\le I$, and $\lambda > 0$, 
we set sequences as 
\begin{eqnarray}
x_{i}^I:=\sqrt{V(P)}\left(-\lambda+\frac{2\lambda}{I}i\right), 
~~y_i^I:=y_{P,A}(x_{i}^I),
\end{eqnarray}
where the function $y_{P,A}(x)$ was defined in (\ref{y}).
Here we introduce a probability distribution $P'_{n, I}$.
For any $j\in\cup_{i=1}^I S_n^P(x_{i-1}^I, x_{i}^I)$, 
we note that there uniquely exists $i$ such that $j\in S_n^P(x_{i-1}^I, x_{i}^I)$.
Then we define $P'_{n, I}$ as 
\begin{eqnarray}
P'_{n, I}(j)
&=&\frac{P^{n\downarrow}(S_n^P(y_{i+1}^I, y_{i+2}^I))}{Q^{\frac{H(P)}{H(Q)}n+b\sqrt{n}\downarrow}( S_n^P(x_{i-1}^I, x_{i}^I))}Q^{\frac{H(P)}{H(Q)}n+b\sqrt{n}\downarrow}(j)
\Label{newprob}
\end{eqnarray}
for $1\le i\le I-2$ and $j\in S_n^P(x_{i-1}^I, x_{i}^I)$.
%
%
Here, there is no constraint for $P'_{n, I}(j)$ with $j\in\N\setminus S_n^P(x_{0}^I, x_{I-2}^I)$ as long as $P'_{n, I}$ is a probability distribution.
Then, the following holds:
\begin{eqnarray}
\sum_{j \in S_n^P(x_{i-1}^I,x_{i}^I)}P'_{n, I}(j)
~= \sum_{k \in S_n^P(y_{i+1}^I, y_{i+2}^I)} P^{n\downarrow}(k)
\Label{assum4}
\end{eqnarray}
for $1\le i\le I-2$.
Using the definition (\ref{newprob}) of $ P'_{n, I}(j)$, 
we have (\ref{Gineq1}) as follows:
\begin{eqnarray}
&&\liminf_{n\to\infty} F(P'_{n,I}, Q^{ \frac{H(P)}{H(Q)}n+b\sqrt{n}\downarrow})\nonumber\\
&\ge&\liminf_{n\to\infty}
\sum_{i=1}^{I-2}\sum_{j\in S_n^P(x_{i-1}^I, x_{i}^I)}
\sqrt{P'_{n, I}(j)}\sqrt{Q^{ \frac{H(P)}{H(Q)}n+b\sqrt{n}\downarrow}(j)}
\nonumber \\
&=&\liminf_{n\to\infty}
\sum_{i=1}^{I-2}\sum_{j\in S_n^P(x_{i-1}^I, x_{i}^I)}
\sqrt{
\frac{P^{n\downarrow}(S_n^P(y_{i+1}^I, y_{i+2}^I))}{Q^{ \frac{H(P)}{H(Q)}n+b\sqrt{n}\downarrow}( S_n^P(x_{i-1}^I, x_{i}^I))}}
Q^{ \frac{H(P)}{H(Q)}n+b\sqrt{n}\downarrow}(j) \nonumber \\
&=&\liminf_{n\to\infty}
\sum_{i=1}^{I-2}
\sqrt{
\frac{P^{n\downarrow}(S_n^P(y_{i+1}^I, y_{i+2}^I))}
{Q^{ \frac{H(P)}{H(Q)}n+b\sqrt{n}\downarrow}( S_n^P(x_{i-1}^I, x_{i}^I))}}
Q^{ \frac{H(P)}{H(Q)}n+b\sqrt{n}\downarrow}( S_n^P(x_{i-1}^I, x_{i}^I))
\nonumber \\
&=&\liminf_{n\to\infty}
\sum_{i=1}^{I-2}\sqrt{P^{n\downarrow}(S_n^P(y_{i+1}^I, y_{i+2}^I))} 
\sqrt{Q^{ \frac{H(P)}{H(Q)}n+b\sqrt{n}\downarrow}( S_n^P(x_{i-1}^I, x_{i}^I))}\Label{H6-16-8}\\
&=&\sum_{i=1}^{I-2}\sqrt{\Phi\left(\frac{y_{i+2}^I}{\sqrt{V(P)}}\right)-\Phi\left(\frac{y_{i+1}^I}{\sqrt{V(P)}}\right)} \sqrt{\Phi_{P,Q,b}\left(\frac{x_{i}^I}{\sqrt{V(P)}}\right)-\Phi_{P,Q,b}\left(\frac{x_{i-1}^I}{\sqrt{V(P)}}\right)}\nonumber\\
&&\Label{lim.ineq}\\
&=&\sum_{i=1}^{I-2}\sqrt{A\left(-\lambda+\frac{2\lambda}{I}(i+2)\right)-A\left(-\lambda+\frac{2\lambda}{I}(i+1)\right)} \nonumber\\
&&~~~~\times
\sqrt{\Phi_{P,Q,b}\left(-\lambda+\frac{2\lambda}{I}i\right)-\Phi_{P,Q,b}\left(-\lambda+\frac{2\lambda}{I}(i-1)\right)}
\nonumber\\
&=&
\sum_{i=1}^{I-2}\sqrt{\int_{-\lambda+\frac{2\lambda}{I}i}^{-\lambda+\frac{2\lambda}{I}(i+1)}
\frac{dA}{dx}\left(x+\frac{2\lambda}{I}\right)dx} 
\sqrt{\int_{-\lambda+\frac{2\lambda}{I}i}^{-\lambda+\frac{2\lambda}{I}(i+1)}
\phi_{P,Q,b}\left(x-\frac{2\lambda}{I}\right)dx}\nonumber \\
&\ge&\sum_{i=1}^{I-2}
\int_{-\lambda+\frac{2\lambda}{I}i}^{-\lambda+\frac{2\lambda}{I}(i+1)}
\sqrt{\frac{dA}{dx}\left(x+\frac{2\lambda}{I}\right)}
\sqrt{\phi_{P,Q,b}\left(x-\frac{2\lambda}{I}\right)}
dx\Label{Schwarz}
\end{eqnarray}
\begin{eqnarray}
&=&
\int_{-\lambda+\frac{2\lambda}{I}}^{-\lambda+\frac{2\lambda}{I}(I-1)}
\sqrt{\frac{dA}{dx}\left(x+\frac{2\lambda}{I}\right)}
\sqrt{\phi_{P,Q,b}\left(x-\frac{2\lambda}{I}\right)}
dx\nonumber \\
&\overset{I\to\infty}{\longrightarrow}&\int_{-\lambda}^{\lambda}\sqrt{\frac{dA}{dx}(x)}\sqrt{\phi_{P,Q,b}(x)}dx\nonumber\\
&\ge& {\cal F}\left(\frac{dA}{dx}, \phi_{P, Q, b}\right) - \epsilon,
\Label{H6-16-9},
\end{eqnarray}
where 
(\ref{lim.ineq}) follows from Lemma \ref{lem.central},
and
(\ref{H6-16-9})  follows from (\ref{bishou0}).

Then, we show the existence of a sequence $\{f_n\}_{n=1}^{\infty}$ of maps which satisfies (\ref{Gineq2}).
From Lemma \ref{Wlem},
we choose a map $f_{n,I}:\N\to\N$ 
for $2<I\in\N$ and $n\in\N$ such that 
\begin{eqnarray}
W_{f_{n,I}}(S_n^P(y_{i+1}^I, y_{i+2}^I))
&\subset& S_n^P(x_{i-1}^I, x_{i}^I),
\Label{W2}\\
P'_{n, I}(j)
&\le& W_{f_{n,I}}(P^{n\downarrow})(j) 
+\max_{k\in S_n^P(y_{i+1}^I, y_{i+2}^I)}P^{n\downarrow}(k)
\Label{W1}
\end{eqnarray}
for any $1\le i\le I-2$ and $j\in S_n^P(x_{i-1}^I, x_{i}^I)$.
Then,
combining (\ref{W1}) with the inequality
\begin{eqnarray}
\max_{k\in S_n^P(y_{i+1}^I, y_{i+2}^I)}P^{n\downarrow}(k)
~\le~ \min_{k\in S_n^P(y_{i+1}^I)}P^{n\downarrow}(k)
= \alpha_n^P(x_{i+1}^I),
\end{eqnarray}
where $\alpha_n^P$ was defined in (\ref{def-alpha}), 
we have
\begin{eqnarray}
&& F(W_{f_{n,I}}(P^{n\downarrow}), Q^{ \frac{H(P)}{H(Q)}n+b\sqrt{n}\downarrow})\nonumber\\
&\ge&\sum_{i=1}^{I-2}\sum_{j\in S_n^P(x_{i-1}^I, x_{i}^I)}
\sqrt{W_{f_{n,I}}(P^{n\downarrow})(j)} \sqrt{Q^{ \frac{H(P)}{H(Q)}n+b\sqrt{n}\downarrow}(j)}
\nonumber\\
&\ge&\sum_{i=1}^{I-2}\sum_{j\in S_n^P(x_{i-1}^I, x_{i}^I)}
\sqrt{\max\{P'_{n, I}(j)-\alpha_n^P(x_{i+1}^I),0\}} \sqrt{Q^{ \frac{H(P)}{H(Q)}n+b\sqrt{n}\downarrow}(j)}
\nonumber\\
&\ge&
\sum_{i=1}^{I-2}\sum_{j\in S_n^P(x_{i-1}^I, x_{i}^I)}
\sqrt{P'_{n, I}(j)}\sqrt{Q^{ \frac{H(P)}{H(Q)}n+b\sqrt{n}\downarrow}(j)}\Label{H6-16-4}\\
&&-
\sum_{i=1}^{I-2}\sum_{j\in S_n^P(x_{i-1}^I, x_{i}^I)}
\sqrt{\alpha_n^P(x_{i+1}^I)} 
\sqrt{Q^{ \frac{H(P)}{H(Q)}n+b\sqrt{n}\downarrow}(j)}\nonumber\\
&=&
F(P'_{n,I}, Q^{ \frac{H(P)}{H(Q)}n+b\sqrt{n}\downarrow})\Label{H6-16-5}\\
&&-
\sum_{i=1}^{I-2}\sum_{j\in S_n^P(x_{i-1}^I, x_{i}^I)}
\sqrt{\alpha_n^P(x_{i+1}^I)} 
\sqrt{Q^{ \frac{H(P)}{H(Q)}n+b\sqrt{n}\downarrow}(j)},\nonumber
\end{eqnarray}
where (\ref{H6-16-4}) follows from (\ref{W1}) and  the last inequality follows from $\sqrt{x-y} \ge \sqrt{x}- \sqrt{y}$ for any $x\ge y \ge 0$.
%
Using the Schwarz inequality,
the second term of (\ref{H6-16-4}) can be evaluated as follows:
\begin{eqnarray}
&&\sum_{i=1}^{I-2}\sum_{j\in S_n^P(x_{i-1}^I, x_{i}^I)}
\sqrt{\alpha_n^P(x_{i+1}^I)} 
\sqrt{Q^{ \frac{H(P)}{H(Q)}n+b\sqrt{n}\downarrow}(j)}
\nonumber\\
&\le&
\sum_{i=1}^{I-2}
\sqrt{\alpha_n^P(x_{i+1}^I)} 
\sqrt{|S_n^P(x_{i-1}^I, x_{i}^I)|}
\sqrt{
\sum_{j\in S_n^P(x_{i-1}^I, x_{i}^I)}
Q^{ \frac{H(P)}{H(Q)}n+b\sqrt{n}\downarrow}(j)}
\nonumber\\
&\le&
\sum_{i=1}^{I-2}
\sqrt{\alpha_n^P(x_{i+1}^I)} 
\sqrt{|S_n^P(x_{i}^I)|}
\nonumber\\
&\le & 
\sum_{i=1}^{I-2}
\sqrt{ 2^{-2\lambda\sqrt{V(P)n}/I} }
=
(I-2) 2^{-\lambda\sqrt{V(P)n}/I} 
\Label{H6-16-7}\\
&\overset{n\to\infty}{\to}&0, \Label{H6-16-9}
\end{eqnarray}
where the inequality (\ref{H6-16-7}) follows from Lemma \ref{alpha}.
Thus, we obtain (\ref{Gineq2}) from (\ref{H6-16-5}) and (\ref{H6-16-9}).

\subsubsection{Converse Part} \Label{sec:reduction.con}

In this subsection,
we prove (\ref{H2}) for the non-uniform case.
From (\ref{d-c5}), it is enough to show 
\begin{align}
F_{P,Q}^{\cal M}(b) 
 ~\le~\sup_{A} {\cal F}\left(\frac{dA}{dx}, \phi_{P, Q, b}\right) +\epsilon
 \Label{HH134}
\end{align}
for an arbitrary $\epsilon>0$.
To show (\ref{HH134}), we prepare the following lemma.
\begin{lem}\Label{20}
Assume that real numbers $t\le t'$ satisfy the condition {\rm ($\star$)}
in Lemma \ref{lem.converse}.
Then the following inequality holds
\begin{eqnarray}
F_{P,Q}^{\cal M}(b) 
&\le& \sqrt{\Phi(t)}\sqrt{\Phi_{P, Q, b}(t)}
+\int_{t}^{t'}\sqrt{\phi(x)} \sqrt{\phi_{P, Q, b}(x)}dx \nonumber\\
&&+\sqrt{1-{\Phi}(t')} \sqrt{1-{\Phi}_{P, Q, b}(t')}.\nonumber
\end{eqnarray}
\end{lem}
The proof of Lemmas \ref{20} is given in Appendix \ref{app.Converse}.
Then we obtain (\ref{HH134}) as follows:
\begin{align}
F_{P,Q}^{\cal M}(b) 
 ~\le~ {\cal F}\left(\frac{dA_{bD_{P,Q},C_{P,Q}}}{dx}, \phi_{P, Q, b}\right) +\epsilon
 ~\le~\sup_{A} {\cal F}\left(\frac{dA}{dx}, \phi_{P, Q, b}\right) +\epsilon,
\end{align}
where the function $A_{\mu,v}$ was defined in (\ref{A_1}), (\ref{A_2}) and (\ref{A_3}) and the first inequality follows from Lemmas \ref{lem.converse2} and \ref{20}.

\section{Application to Quantum Information Theory} \Label{sec:quantum}

In this section, we apply the second-order asymptotics to the approximate conversion between two bipartite pure entangled states by LOCC and the cloning for a known entangled pure state by LOCC.

\subsection{Entangled State and LOCC Conversion} \Label{sec:application}

We first briefly introduce some notions of quantum information theory which are used in this section.
In quantum information theory,
a quantum system is described by a Hilbert space.
Then, 
a quantum state on the quantum system is defined by a density operator on the Hilbert space, 
i.e. a positive semidefinite operator whose trace is one,
and in particular, a quantum state whose rank is one is called a pure state.
A collection of some quantum systems is called
a composite system and is described by the tensor product of Hilbert spaces of the quantum systems which constitute the composite system. 
Then, a tensor product state is defined by a quantum state which is represented by a tensor product of density operators on each quantum systems
and 
a quantum state is called separable when the quantum state can be represented by a convex combination of separable states.
On the other hand,
a quantum state which is not separable is said to be entangled.
In this paper,
we treat only finite-dimensional bipartite composite systems
and assume that quantum states are pure entangled states.

Entanglement is used in several quantum informational operations \cite{HHHH09,OP93,BZ06,BDSW96,DHR02,BS98}.
Pure entangled states can be expressed by the Schmidt decomposition 
and the coefficients are called the Schmidt coefficients of the entangled state.
Squared Schmidt coefficients consist of a probability distribution from the property of a pure entangled state 
and is helpful to describe characteristics about entanglement.
In various quantum operations,
most entangled states are often focused on.
Such a state is called a maximally entangled state 
and defined by an entangled state whose Schmidt coefficients are all equivalent to each other.
In particular, a maximally entangled state on a two-qubit system is called the EPR state.
Entanglement of a pure state $\psi$
is characterized by the von Neumann entropy $S_{\psi}$ of its partial density matrix,
which coincides with the Shannon entropy of squared Schmidt coefficients of $\psi$.
For example, the pure state $\psi$ is entangled if and only if $S_{\psi}\ne 0$.
Since several values cannot be defined for the singular case $S_{\psi}= 0$, 
we assume that pure states are entangled in this section.

The conversion of entangled states by LOCC has been studied in both the non-asymptotic case \cite{Vid99,Nie99,VJN00} and the asymptotic case \cite{BBPS96,BPRST00,HHT01}.
In this section, as an application to quantum information theory, we treat problems of the approximate conversion between pure entangled states by LOCC.
As a typical LOCC conversion, 
we focus on entanglement concentration, in which, 
an i.i.d. pure state of $\psi$ is converted to multiple copies of the EPR state.
It is known that the optimal first-order conversion rate is the von Neumann entropy $S_{\psi}$ of its partial density matrix \cite{BBPS96}.
Then, 
it is possible to approximately generate multiple copies $\psi_{EPR}^{S_{\psi}n+o(n)}$ of the EPR state from  $n$-copies $\psi^{\otimes n}$ of a given state $\psi$
under the condition that
the fidelity between the generated state and the target state asymptotically goes to $1$.
However, the converse does not holds, that is, 
even when the number of EPR states to be generated has the asymptotic expansion of the form of $S_{\psi}n+o(n)$,
it is not necessarily possible to generate them from  $n$-copies $\psi^{\otimes n}$ 
under the condition that the fidelity between the generated state and the target state asymptotically goes to $1$.
In order to treat the error of LOCC conversion more precisely, we need to deal with the second-order asymptotics.
That is, the asymptotically achievable fidelity between the generated state and the target state
depends on the coefficient of the order $\sqrt{n}$.
A similar problem occurs in entanglement dilution, in which, the multiple copies of the EPR state are 
converted to the multiple copies of a target pure entangled state.
That is, in entanglement dilution,
the asymptotically achievable fidelity between the generated state and the target state 
also depends on the coefficient of the order $\sqrt{n}$.
Such relations in entanglement concentration and dilution were studied in \cite{HL04,HW03,KH13}.
However, the existing studies dealt with 
the relation between the asymptotic fidelity and the coefficient of the order $\sqrt{n}$
only when the initial or the target state is the EPR state, and
thus they did not investigate this relation when both of the initial and the target states are non-EPR states.
In the following, we treat more general LOCC conversions including entanglement concentration 
and entanglement dilution under the fidelity constraint, 
and clarify the relation between the second-order rate of the conversion and the asymptotically achievable fidelity between the generated state and the target state.


Before going to the asymptotics of LOCC conversion, 
we give some notations and remarks.
In the following, we employ the fidelity 
\begin{eqnarray}
F(\psi, \omega) = \langle \psi, \omega  \rangle
\Label{fidelity}
\end{eqnarray}
to describe the accuracy of LOCC conversions, where the right hand side of (\ref{fidelity}) is the inner product between pure states $\psi$ and $\omega$. 
The following value represents the maximum fidelity of LOCC conversion for states $\psi$ and $\omega$:
\begin{eqnarray}
F(\psi\to\omega)
:=\max\{F(\Gamma(\psi), \omega) ~|~
 \Gamma \hbox{: LOCC}
\}.
\end{eqnarray}
Let $P_{\psi}$ and $P_{\omega}$ be the probability distributions which consist of the squared Schmidt coefficients for pure entangled states $\psi$ and $\omega$, respectively.
Then, it was shown in Lemma $1$ of \cite{VJN00} that  the fidelity $F(\psi, \omega)$ between pure entangled states relates with the fidelity $F(P_\psi, P_\omega)$ between probability distributions as
\begin{align}
F(P_\psi^{\downarrow}, P_\omega^{\downarrow})= 
\max
\{F((U_A \otimes U_B) \psi, \omega)~|~
U_A,U_B \hbox{: unitary~operations}\}.
\label{6-14-1}
\end{align}
Since $\psi$ is transformed to $\omega$ by LOCC if and only if $P_{\psi}\prec P_{\omega}$ \cite{Nie99} where $\prec$ is the majorization relation given in Subsection 2.2, 
(\ref{6-14-1}) implies the following relation for pure states $\psi$ and $\omega$:
\begin{eqnarray}
F(\psi\to\omega)=F^{\cal M}(P_{\psi}\to P_{\omega}).\Label{dd}
\end{eqnarray}
We define the maximum conversion number for $\omega$ from $n$-copies of $\psi$ by LOCC under a permissible accuracy $0<\tau<1$ as follows:
\begin{eqnarray}
L_n(\psi, \omega|\tau)
&:=&\max\{L\in\N~|
~\exists~ \hbox{LOCC } \Gamma \hbox{ s.t. }
F(\Gamma(\psi^{\otimes n}), \omega^{\otimes L})\ge\tau\}\Label{L}\\
&=&\max\{L\in\N~|~F(\psi^{\otimes n}\to \omega^{\otimes L})\ge\tau\}.
\end{eqnarray}
Since $P_{\psi^{\otimes n}}=P_{\psi}^n$, the following holds:
\begin{eqnarray}
L_n(\psi, \omega|\tau)
=L^{\cal M}_n(P_{\psi}, P_{\omega}|\tau).\Label{ML}
\end{eqnarray}
Let $V_{\psi}:=V(P_{\psi})$,
$D_{\psi,\omega}:=D_{P_\psi,P_\omega}$ and $C_{\psi,\omega}:=C_{P_\psi,P_\omega}$.
Then we call $C_{\psi,\omega}$ the {\it LOCC conversion characteristics} between pure states $\psi$ and $\omega$. 

Since $H(P_{\psi})=S_{\psi}$ and $H(P_{\omega})=S_{\omega}$,
the asymptotic expansion of the maximum conversion number $L_n(\psi, \omega|\tau)$ is obtained from Theorem \ref{thm.gen} as follows.
\begin{thm}
\Label{thm.gen.LOCC}
\begin{eqnarray}
L_n(\psi, \omega|\tau)
&=& \frac{S_{\psi}}{S_{\omega}}n + \frac{Z_{C_{\psi,\omega}}^{-1}(1-\tau^2)}{D_{\psi,\omega}}\sqrt{n} +o(\sqrt{n}).
\Label{exp.LOCC}
\end{eqnarray}
\end{thm}
In particular,
when the initial state is the maximally entangled state $\psi_m^{max}$ on $\C^m\otimes\C^m$,
Theorem \ref{thm.gen.LOCC} implies
\begin{eqnarray}
L_n(\psi_m^{max},\omega|\tau)
&=& \frac{\log m}{S_{\omega}}n + \sqrt{\frac{V_{\omega}\log m}{S_{\omega}^3}}\Phi^{-1}(1-\tau^2)\sqrt{n}  +o(\sqrt{n}).
\Label{qua.dil}
\end{eqnarray}
Similarly,
when the target state is the maximally entangled state $\psi_m^{max}$,
we have 
\begin{eqnarray}
L_n(\psi,\psi_m^{max}|\tau)
&=& \frac{S_{\psi}}{\log m}n + \frac{\sqrt{V_{\psi}}}{\log m}\Phi^{-1}(1-\tau^2)\sqrt{n} +o(\sqrt{n}).
\Label{qua.con}
\end{eqnarray}

\begin{rem}
Bennett et al.  \cite{BBPS96} gave the first-order conversion rate of $L_n(\psi, \omega|\tau)$.
Moreover, when $\psi$ or $\omega$ is the EPR state (i.e. the cases of entanglement dilution or entanglement concentration), 
Hayden and Winter \cite{HW03} and Harrow and Lo \cite{HL04} pointed out that 
 the second-order of $L_n(\psi, \omega|\tau)$ is $\sqrt{n}$ and its second-order rate depends on the permissible accuracy for those operations.
However, the explicit form of the second-order rate for entanglement dilution and concentration was not obtained in their work. 
On the other hand, 
when $\psi$ or $\omega$ is the EPR state,
Theorem \ref{thm.gen.LOCC} gives the explicit  second-order rate of $L_n(\psi, \omega|\tau)$ in (\ref{exp.gen}), 
which coincides with the result in \cite{KH13},
and hence, our results provide a refinement of the existing studies.
Moreover, we also derived the second-order rates of $L_n(\psi, \omega|\tau)$ when both $\psi$ and $\omega$ are non-EPR pure states.
Therefore, we obtain the second-order expansion of  $L_n(\psi, \omega|\tau)$ in all cases as long as both $\psi$ and $\omega$ are entangled pure state.
\end{rem}


\begin{rem}
We mention the relation between conversion problems in conventional and quantum information theory.
Due to the results of Nielsen \cite{Nie99}, the approximate conversion problem between pure states on bipartite systems 
is induced into that of probability distribution under the majorization condition by considering the squared Schmidt coefficients of the pure states.
Since the squared Schmidt coefficients of a maximally entangled state form a uniform distribution, 
in particular, those of the EPR state form the uniform distribution over $\{0,1\}$,
it is thought that entanglement dilution and concentration in quantum information theory correspond to the resolvability 
and the intrinsic randomness in conventional information theory.
\end{rem}

\begin{rem}
Here, we remark the relation with a variable-rate protocol.
In this paper, our protocol fixes the conversion rate between the numbers of initial and target states.
Hence, such a protocol is called a fixed-rate protocol.
However, we can decide the rate depending on our measurement outcome during our protocol.
When we generate the maximally entangled state from a partially entangled state, the papers \cite{BBPS96,MH07} discuss this problem.
Such a protocol is often called ``entanglement gambling" \cite{LP01}.
To interpret $C_{\psi,\omega}$ as the variance, we need to consider a variable-rate protocol.
Unfortunately, our result gives the relation between the error and the conversion rate only for a fixed-rate protocol,
and cannot be applied to a variable-rate protocol.
Such an extension remains as a future study.
\end{rem}

\subsection{LOCC Cloning with Perfect Knowledge} \Label{sec:application2}
Due to the no-cloning theorem, we cannot generate a complete copy of an unknown quantum state.
Then, in studies of the cloning of an unknown quantum state, an approximate cloning method and the evaluation of its  accuracy 
have been mainly treated \cite{Wer98,BH98,Fil04}.
On the other hand, 
even when the state to be copied is known,
it is impossible to perfectly copy the state
when the state is entangled and our operations are limited to LOCC.
In the following, we treat such a case.
Thus, we assume that we know entangled state to be copied, but, we can use only LOCC for cloning.
We note that existing studies \cite{ACP04,OH06} discussed similar cloning problem\footnote{The papers \cite{ACP04,OH06} discussed local copying,
and a limited amount of the EPR states are prepared as a resource for copying,
only LOCC is allowed for our operation,
and 
we only know that the state to be copied belongs to the set of candidate of the states.
It is required to copy the unknown state perfectly by using the same amount of the EPR states
as the number of required clones.}, however,
the setting is different from ours
because their setting assumes an imperfect knowledge for the entangled state to be cloned and additional limited entangled resource.
To distinguish their setting, we call our setting the LOCC cloning with perfect knowledge, and call their setting the LOCC cloning with imperfect knowledge.
%
\begin{figure}[t]
 \begin{center}
\vspace*{-2em}
 \hspace*{-18em}
 \mbox{\raisebox{-0mm}{
 \includegraphics[width=70mm, height=37mm, bb=20 20 330 210]{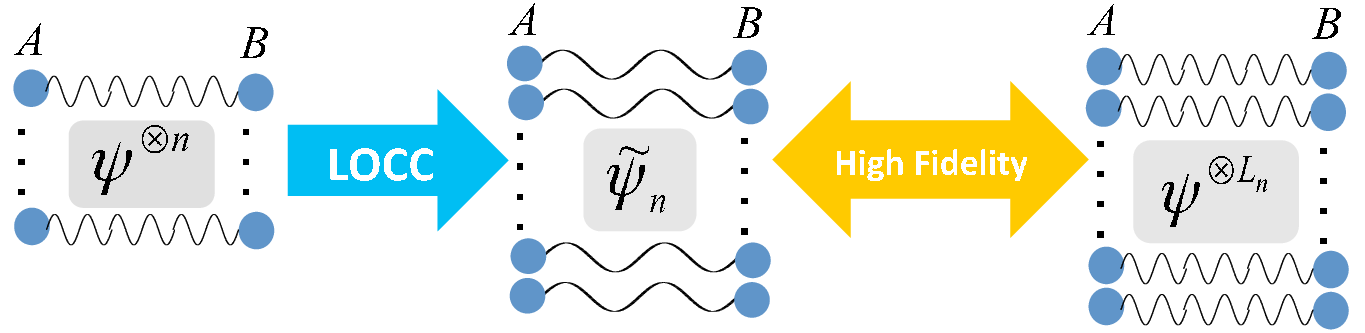}
 }}
 \end{center}
 \caption{
In approximate LOCC cloning, an i.i.d. entangled state $\psi^{\otimes n}$ is transformed by LOCC to a quantum state $\tilde{\psi}_n$ which has high fidelity for $\psi^{\otimes L_n}$.
Under the condition that the fidelity between $\tilde{\psi}_n$ and $\psi^{\otimes L_n}$ is greater than $\tau$, the maximal number of copies of  $\psi^{\otimes L_n}$ is $L_n(\psi, \psi|\tau)$.
}
 \Label{cloning}
\end{figure}
%
In this paper, 
we investigate LOCC cloning with perfect knowledge when the initial entangled state is $n$-copies of $\psi$ 
and the target state is $L_n$-copies of $\psi$ with $L_n\ge n$.
That is, we analyze how large number $L_n$ of copies we can generate under the condition that the fidelity between the transformed state from the initial state by LOCC and the target entangled state is greater than a permissible accuracy $\tau$.
The maximal number $L_n(\psi|\tau)$ of $L_n$ given above is formulated by
\begin{eqnarray}
L_n(\psi|\tau)
:=\max\{L\in\N ~|~ F(\Gamma(\psi^{\otimes n}),\psi^{\otimes L})\ge\tau,~ \Gamma:LOCC\},
\end{eqnarray}
which equals $L_n(\psi, \psi|\tau)$ by the definition in (\ref{L}).
Then we obtain the following asymptotic expansion from (\ref{cla.clone}).
\begin{thm}
For a pure entangled state $\psi$ and $0<\tau<1$,
\begin{eqnarray}
L_n(\psi|\tau)
= n+{\frac{\sqrt{8V_{\psi}\mathrm{log}\tau^{-1}}}{S_{\psi}}}\sqrt{n}  +o(\sqrt{n}).
\end{eqnarray}
\end{thm}
Thus, 
when the initial state is the i.i.d. entangled state $\psi^{\otimes n}$ of a non-maximally entangled state $\psi$,
the incremental number $L_n(\psi|\tau)-n$ of copies by LOCC cloning with perfect knowledge has the order of $\sqrt{n}$.
On the other hand, when $\psi$ is a maximally entangled state, $V_{\psi}=0$ and thus the incremental number of copies by LOCC cloning with perfect knowledge does not have the order of $\sqrt{n}$.
Indeed, since 
\begin{eqnarray}
\max_{\Gamma:LOCC}F(\Gamma(\psi_{EPR}^{\otimes n}),\psi_{EPR}^{\otimes L})=\sqrt{2^{n-L}}.
\end{eqnarray}
holds by Lemma \ref{opt dil}, we obtain
\begin{eqnarray}
L_n(\psi_{EPR}|\tau)= \lfloor n+2\log \tau^{-1} \rfloor,
\end{eqnarray}
where $\lfloor\cdot\rfloor$ represents the floor function, and the incremental number $L_n(\psi|\tau)-n$ is bounded by a constant $2\log \tau^{-1}$ for any $n$ unlike a non-maximally state.

According to Chiribella et al. \cite{CYY13}, 
we define the replication rate as the limit
\begin{eqnarray*}
r(\psi,\tau):= \displaystyle\lim_{n \to \infty}\log \frac{L_n(\psi|\tau)}{n}.
\end{eqnarray*}
Then, the rate $r(\psi,\tau)$ can be characterized as follows
\begin{align}
r(\psi,\tau)=
\left\{
\begin{array}{ccl}
1/2 & \hbox{ if }& V_{\psi}\neq 0 \\
0 & \hbox{ if }& V_{\psi}= 0.
\end{array}
\right.
\end{align}

\section{Conclusion}\Label{sec:conclusion}

We have addressed approximation conversion problems of probability distributions by deterministic and majorization conversions.
We have found that two conversion methods are related as in (\ref{fidelity ineq})-(\ref{number ineq}) 
and have derived the asymptotic expansion of the maximum conversion number up to the order $\sqrt{n}$
for the both kinds of conversion problems between two i.i.d. probability distributions.
To derive the computable form of the second-order rate of the asymptotic expansion, 
the problem has been divided into the uniform case and the non-uniform case.
However, we note that the maximum conversion numbers for two kinds of conversions are equivalent to each other in all cases up to the order $\sqrt{n}$ as stated in Theorem \ref{thm.gen}. 
A key to derive the asymptotic expansion is
to introduce Rayleigh-normal distribution and to investigate its properties.
In particular, 
the optimal second-order conversion rate is described by the Rayleigh-normal distribution function for the non-uniform case.
Thereafter, as applications to quantum information theory,
we have addressed LOCC conversion problem between bipartite pure entangled states including entanglement concentration and entanglement dilution.
Then, we have derived the asymptotic expansion of the maximum conversion number using the results for majorization conversion of probability distributions.
In particular, 
we have clarified the relation between the second-order conversion rate and the accuracy of LOCC conversion.
As a special case, we have introduced the notion of LOCC cloning with the perfect knowledge.
Using the results for LOCC conversion, we have derived the rate of the incremental copies and the optimal coefficient in this setting. 

The following problems can be considered as future problems.
First, this paper assumes the independent and identical distributed condition for the sequences of distributions or pure entangled states to be converted.
However, the actual sequences of distributions or pure entangled states might have correlation in practice.
Hence, it is an interesting open problem to extend the obtained result to the case of correlated sequences of distributions or pure entangled states  \cite{MPSVW10}, e.g., the Markovian case. 
Next, only pure states have been treated in quantum information setting 
although mixed entangled states may appear in practice.
So, an extension to the case of mixed states is required as a future study.
Finally, we point out the significance of analysis in a finite-length setting.
We have analyzed the asymptotic performance of approximate conversions in this paper.
On the other hand, we can operate an input state only with a finite length.
Therefore, it is needed to analyze the approximate conversion problems in a finite-length setting.
For entanglement dilution and the concentration,
the recent paper \cite{KH13} dealt with an analysis in a finite-length setting and derived its numerical results. 
However, 
no result investigates the finite-length setting of the case when the initial state and the target state are non-EPR states\footnote{After submitting this paper,
 \cite{Ren15} also discussed a similar problem mainly with the variational distance. See section $4$ in \cite{Ren15} in detail.}.

\section*{Acknowledgments}
We would like to thank Dr. Hiroyasu Tajima for his helpful comments, anonymous reviewers for their valuable comments and suggestions to improve the quality of the paper
and Dr. Marco Tomamichel and Mr. Christopher Thomas Chubb for pointing out the error in Fig. 1.
WK was partially supported by Grant-in-Aid for JSPS Fellows No. 233283. 
MH is partially supported by a MEXT Grant-in-Aid for Scientific Research (A) No. 23246071 and the National Institute of Information and Communication
Technology (NICT), Japan.
The Centre for Quantum Technologies is funded by the Singapore Ministry of Education and the National Research Foundation
as part of the Research Centres of Excellence programme.


\appendix

\subsection{Proof of Lemmas: The Explicit Form of The Rayleigh-normal Distributions}
\Label{app.RN}
We prove Lemmas \ref{sol2}, \ref{sol1}, \ref{monotone} and \ref{lem.converse} used to derive Theorem \ref{Zform} which shows the explicit form of the Rayleigh-normal distributions.

~


\noindent{\bf Proof of Lemma \ref{sol2}:}
We first show the inequality $\beta_{\mu,v}<\frac{\mu}{1-v}$.
The existence of the unique solution $\beta_{\mu,v}$ is equivalent to the existence of the unique zero point of the following function with respect to $x$:
\begin{eqnarray}
f_{\mu,v}(x):=(1-\Phi_{\mu,v}(x))-(1-\Phi(x))\frac{\phi_{\mu,v}(x)}{\phi(x)}.
\end{eqnarray}
Then we have
\begin{eqnarray}
\frac{\partial f_{\mu,v}}{\partial x}(x)
&=&-\frac{\partial \Phi_{\mu,v}}{\partial x}(x) + \frac{\partial \Phi}{\partial x}(x)\frac{\phi_{\mu,v}(x)}{\phi(x)} -(1-\Phi(x))\frac{\partial}{\partial x}\left(\frac{\phi_{\mu,v}}{\phi}\right)(x)\\
&=&-\phi_{\mu,v}(x)+\phi(x) \frac{\phi_{\mu,v}(x)}{\phi(x)} -(1-\Phi(x))\frac{\partial}{\partial x}\left(\frac{\phi_{\mu,v}}{\phi}\right)(x)\\
&=&-(1-\Phi)\frac{\partial}{\partial x}\left(\frac{\phi_{\mu,v}}{\phi}\right)(x)\\
&=&-(1-\Phi(x))\frac{\frac{\partial \phi_{\mu,v}}{\partial x}(x)\phi(x)-\phi_{\mu,v}(x)\frac{\partial \phi}{\partial x}(x)}{\phi(x)^2}\\
&=&-(1-\Phi(x))\frac{-\frac{x-\mu}{v}\phi_{\mu,v}(x)\phi(x)+x\phi_{\mu,v}(x)\phi(x)}{\phi(x)^2}\\
&=&-\frac{1-v}{v}\left(\frac{\mu}{1-v}-x\right)\frac{\phi_{\mu,v}(x)}{\phi(x)}(1-\Phi(x)). \Label{f-x}
\end{eqnarray}
Thus, the function $f_{\mu,v}$ is strictly monotonically decreasing when $x<\frac{\mu}{1-v}$ and is strictly monotonically increasing when $x>\frac{\mu}{1-v}$.
Since
\begin{eqnarray}
\lim_{x\to-\infty}f_{\mu,v}(x)=1, ~~\lim_{x\to\infty}f(x)=0, 
\end{eqnarray}
the function $f_{\mu,v}$ has the unique zero point $\beta_{\mu,v}<\frac{\mu}{1-v}$ due to the intermediate value theorem.

Next we show that $\beta_{\mu,v}$ is differentiable monotonically increasing with respect to $\mu$.
Since $f_{\mu,v}$ is continuously differentiable with respect to $x\in\R$ and $\mu\in\R$,
the implicit function theorem derives that $\beta_{\mu,v}$ is differentiable with respect to $\mu$ and
\begin{eqnarray}
\frac{\partial\beta_{\mu,v}}{\partial\mu}
=-\frac{\frac{\partial f_{\mu,v}}{\partial \mu}(\beta_{\mu,v})}{\frac{\partial f_{\mu,v}}{\partial x}( \beta_{\mu,v})}.
\end{eqnarray}
To show that $\beta_{\mu,v}$ is monotonically increasing with respect to $\mu$,
it is enough to prove $\frac{\partial f_{\mu,v}}{\partial x}(\beta_{\mu,v})<0$ and  $\frac{\partial f_{\mu,v}}{\partial \mu}(\beta_{\mu,v})>0$.
From (\ref{f-x}), $\beta_{\mu,v}<\frac{\mu}{1-v}$ and the definition of $\beta_{\mu,v}$,
we have
\begin{eqnarray}
\frac{\partial f_{\mu,v}}{\partial x}(\beta_{\mu,v})
=-\frac{1-v}{v}\left(\frac{\mu}{1-v}-\beta_{\mu,v}\right)(1-\Phi_{\mu,v}(\beta_{\mu,v}))
<0. 
\Label{x-ineq}
\end{eqnarray}
In addition, since  $\phi(y)+y\Phi(y)>0$ on $\R$,\footnote{The function $g(y):=\phi(y)+y\Phi(y)$ is proven to be positive on $\R$ as follows. 
Since $\frac{dg}{dy}(y)=\Phi(y)$ is strictly positive on $\R$, it is enough to show that $\displaystyle\lim_{y\to-\infty}g(y)=0$.
Using  L'H\^opital's rule, $\displaystyle\lim_{y\to-\infty}g(y)=\displaystyle\lim_{y\to-\infty}(\Phi(y)/y^{-1})=\displaystyle\lim_{y\to-\infty}(\phi(y)/(-y^{-2}))=-\displaystyle\lim_{y\to-\infty}y^2\phi(y)=0$.}
\begin{eqnarray}
\frac{\partial f_{\mu,v}}{\partial \mu}(\beta_{\mu,v})
&=&-\frac{\partial \Phi_{\mu,v}}{\partial \mu}(\beta_{\mu,v}) -(1-\Phi(\beta_{\mu,v}))\frac{\frac{\partial \phi_{\mu,v}}{\partial \mu}(\beta_{\mu,v})}{\phi(\beta_{\mu,v})} \\
&=&\frac{1}{\sqrt{v}}\phi\left(\frac{\beta_{\mu,v}-\mu}{\sqrt{v}}\right) - \frac{\beta_{\mu,v}-\mu}{v}(1-\Phi_{\mu,v}(\beta_{\mu,v}))
\\
&=&\frac{1}{\sqrt{v}}
\left(
\phi\left(\frac{\beta_{\mu,v}-\mu}{\sqrt{v}}\right)-\frac{\beta_{\mu,v}-\mu}{\sqrt{v}}\left(1-\Phi\left(\frac{\beta_{\mu,v}-\mu}{\sqrt{v}}\right)\right)\right)\\
&=&\frac{1}{\sqrt{v}}
\left(
\phi\left(-\frac{\beta_{\mu,v}-\mu}{\sqrt{v}}\right)+
\left(-\frac{\beta_{\mu,v}-\mu}{\sqrt{v}}\right)\left(\Phi\left(-\frac{\beta_{\mu,v}-\mu}{\sqrt{v}}\right)\right)\right)
>0. 
\Label{mu-ineq}
\end{eqnarray}
Thus $\beta_{\mu,v}$ is proved to be monotonically increasing with respect to $\mu$.
\endproof

~

\noindent{\bf Proof of Lemma \ref{sol1}:}
We first show the inequality $\alpha_{\mu,v} >\frac{\mu}{1-v}$.
The existence of the unique solution $\alpha_{\mu,v}$ is equivalent to the existence of the unique zero point of the function with respect to $x$
\begin{eqnarray}
f_{\mu,v}(x)=\frac{\phi(x)}{\phi_{\mu,v}(x)}\Phi_{\mu,v}(x)-\Phi(x).
\end{eqnarray}
Then, similar to the derivation of (\ref{f-x}), we have 
\begin{eqnarray}
\frac{\partial f_{\mu,v}}{\partial x}(x)
=-\frac{v-1}{v}\left(x-\frac{\mu}{1-v}\right)\frac{\phi(x)}{\phi_{\mu,v}(x)}\Phi_{\mu,v}(x).
\Label{f-x2}
\end{eqnarray}
Thus,
the function $f_{\mu,v}$ is strictly monotonically increasing when $x<\frac{\mu}{1-v}$ and is strictly monotonically decreasing when $x>\frac{\mu}{1-v}$.
Since
\begin{eqnarray}
\lim_{x\to-\infty}f(x)=0, ~~\lim_{x\to\infty}f(x)=-1, 
\end{eqnarray}
the function $f_{\mu,v}$ has the unique zero point $\alpha_{\mu,v}>\frac{\mu}{1-v}$ due to the intermediate value theorem.

Next we show that $\alpha_{\mu,v}$ is differentiable monotonically decreasing with respect to $\mu$.
Since $f_{\mu,v}$ is continuously differentiable with respect to $x\in\R$ and $\mu\in\R$,
the implicit function theorem derives that $\alpha_{\mu,v}$ is differentiable with respect to $\mu$ and
\begin{eqnarray}
\frac{\partial\alpha_{\mu,v}}{\partial\mu}
=-\frac{\frac{\partial f_{\mu,v}}{\partial \mu}(\alpha_{\mu,v})}{\frac{\partial f_{\mu,v}}{\partial x}(\alpha_{\mu,v})}.
\end{eqnarray}
To show that $\alpha_{\mu,v}$ is monotonically decreasing with respect to $\mu$,
it is enough to prove $\frac{\partial f_{\mu,v}}{\partial x}(\alpha_{\mu,v})<0$ and  $\frac{\partial f_{\mu,v}}{\partial \mu}(\alpha_{\mu,v})<0$.
From (\ref{f-x2}), $\alpha_{\mu,v}>\frac{\mu}{1-v}$ and the definition of $\alpha_{\mu,v}$,
we have
\begin{eqnarray}
\frac{\partial f_{\mu,v}}{\partial x}(\alpha_{\mu,v})
~=~ -\frac{v-1}{v}\left(\alpha_{\mu,v}-\frac{\mu}{1-v}\right)\Phi(\alpha_{\mu,v})
~<~ 0. 
\Label{x-ineq2}
\end{eqnarray}
In addition, since  $\phi(y)+y\Phi(y)>0$ for any $y\in\R$,
\begin{eqnarray}
\frac{\partial f_{\mu,v}}{\partial \mu}(\alpha_{\mu,v})
&=&\frac{-\phi(\alpha_{\mu,v})}{\sqrt{v}\phi_{\mu,v}(\alpha_{\mu,v})}
\left(
\phi\left(\frac{\alpha_{\mu,v}-\mu}{\sqrt{v}}\right)+\frac{\alpha_{\mu,v}-\mu}{\sqrt{v}}\Phi\left(\frac{\alpha_{\mu,v}-\mu}{\sqrt{v}}\right)\right)
~<~ 0. 
\Label{mu-ineq2}
\end{eqnarray}
Thus $\alpha_{\mu,v}$ is proved to be monotonically decreasing with respect to $\mu$.
\endproof

~

\noindent{\bf Proof of Lemma \ref{monotone}:}
The interval $\mathcal{I}_{\mu,v}$ is easily derived from
\begin{eqnarray}
\frac{\phi(x)}{\phi_{\mu,v}(x)}
=\left\{
\begin{array}{ll}
e^{\frac{\mu^2}{2}}e^{-\mu x}&\hbox{ if }v=1\\
\sqrt{v}e^{\frac{-\mu^2}{2(1-v)}}e^{\frac{1-v}{2v}\left(x-\frac{\mu}{1-v}\right)^2}&\hbox{ if }v\ne1.
\end{array}
\right.
\Label{Nratio}
\end{eqnarray}
\endproof

~

To show Lemma \ref{lem.converse},
we prepare the following lemma.
\begin{lem}\Label{naiseki2}
Let  ${\bf a}=\{a_i\}_{i=0}^{I}$ and ${\bf b}=\{b_i\}_{i=0}^{I}$ be probability distributions and satisfy 
\begin{eqnarray}
\frac{a_{i-1}}{b_{i-1}} ~>~ \frac{a_i}{b_i}
\Label{ratio}
\end{eqnarray}
for $i=1,2,...,I$.
When ${\bf c}=\{c_i\}_{i=0}^{I}$ is a probability distribution and satisfies 
\begin{eqnarray}
\sum_{i=0}^k a_i ~\le~ \sum_{i=0}^k c_i~~~(k=0,1,...,I)
\Label{gijimajo}
\end{eqnarray}
the following holds:
\begin{eqnarray}
\sum_{i=0}^I \sqrt{a_i}\sqrt{b_i}
~\ge~\sum_{i=0}^I \sqrt{c_i}\sqrt{b_i}.
\Label{acineq}
\end{eqnarray}
Moreover, the equation holds for ${\bf c}$ if and only if ${\bf c}={\bf a}$.
\end{lem}

\begin{proof}
Let $D({\bf a})$ be the set of probability distributions on $\{0,1,...,I\}$ whose element ${\bf c}$ satisfies (\ref{gijimajo}).
Here, we define a function as $f_b({\bf c}):=F({\bf b},{\bf c})$ on $D({\bf a})$ by the fidelity $F$.
Then, to Lemma \ref{naiseki2},
it is enough to prove that ${\bf a}$ uniquely attains the maximum of $f_b$.
Since the function $f_{\bf b}$ is continuous, there exists a maximizer of $f_b$ in $D({\bf a})$.
In the following,
we will show that an arbitrary ${\bf c}\in D({\bf a})$ which is not ${\bf a}$ does not attain the maximum of $f_{\bf b}$.
Then, it implies that ${\bf a}$ is the unique maximizer of $f_{\bf b}$.

Note that there exist two natural numbers $i_0<l_0\in\{0,1,...,I\}$ such that
$a_{i_0}<c_{i_0}$ and $a_{j_0}< c_{j_0}$ hold
since ${\bf a}$ and ${\bf c}$ are different probability distributions and satisfy (\ref{gijimajo}).
Then we have 
\begin{eqnarray}
\frac{c_{i_0}}{b_{i_0}}>\frac{a_{i_0}}{b_{i_0}}>\frac{a_{j_0}}{b_{j_0}}>\frac{c_{j_0}}{b_{j_0}}.
\Label{ineqs}
\end{eqnarray}
Hence, 
for a constant $0<\epsilon_0<\min\{c_{i_0}-a_{i_0},a_{j_0}-c_{j_0}, \frac{1}{2}(c_{i_0}-c_{j_0})\}$,
the following holds\footnote{This fact follows from that for the function $g(\epsilon):=\sqrt{c_{i_0}-\epsilon}\sqrt{b_{i_0}}+\sqrt{c_{j_0}+\epsilon}\sqrt{b_{j_0}}$, its derivative $\frac{dg}{d\epsilon}(\epsilon)=\frac{1}{2}(\sqrt{\frac{b_{j_0}}{c_{j_0}+\epsilon}}-\sqrt{\frac{b_{i_0}}{c_{i_0}-\epsilon}})$ is positive for $\epsilon\in[0,\epsilon_0]$ due to (\ref{ineqs}).}:
\begin{eqnarray}
\sqrt{c_{i_0}}\sqrt{b_{i_0}}+\sqrt{c_{j_0}}\sqrt{b_{j_0}}
~<~\sqrt{c_{i_0}-\epsilon_0}\sqrt{b_{i_0}}+\sqrt{c_{j_0}+\epsilon_0}\sqrt{b_{j_0}}.
\Label{kaizen}
\end{eqnarray}
When we set ${\bf c}'$ as $c'_{i_0}:=c_{i_0}-\epsilon_0$, $c'_{j_0}:=c_{j_0}+\epsilon_0$ and $c'_k=c_k$ for $k\ne {i_0},{j_0}$,
${\bf c}'$ is in $D({\bf a})$ and $f_{\bf b}({\bf c}')>f_{\bf b}({\bf c})$ from (\ref{kaizen}).
Therefore, ${\bf c}$ does not attain the maximum of $f_ {\bf b}$, and only ${\bf a}$ attains the maximum.
\end{proof}

~

\noindent{\bf Proof of Lemma \ref{lem.converse}:}
When we set a sequence $\{x_i^I\}_{i=0}^{I}$ for $I\in\N$ as $x_i^I:=t+\frac{t'-t}{I}i$,
we have the following inequalities for an arbitrary $A$ satisfying the conditions in definition \ref{rn}:
\begin{eqnarray}
&&{\cal F}\left(\frac{dA}{dx},\phi_{\mu,v}\right) \nonumber\\
&=&
\int_{-\infty}^{t}\sqrt{\frac{dA}{dx}(x)} \sqrt{\phi_{\mu,v}(x)}dx
+\int_{t'}^{\infty}\sqrt{\frac{dA}{dx}(x)} \sqrt{\phi_{\mu,v}(x)}dx\\
&&\sum_{i=1}^{I} \int_{x_{i-1}^I}^{x_{i}^I}\sqrt{\frac{dA}{dx}(x)} \sqrt{\phi_{\mu,v}(x)}dx\nonumber\\
&\le&
\sqrt{A(t)} \sqrt{{\Phi}_{\mu,v}(t)} + \sqrt{1-A(t')} \sqrt{1-{\Phi}_{\mu,v}(\lambda)} \Label{ine1}\\
&&+\sum_{i=1}^{I} \sqrt{A(x_i^I)-A(x_{i-1}^I)} \sqrt{\Phi_{\mu,v}(x_i^I)-\Phi_{\mu,v}(x_{i-1}^I)},\nonumber\end{eqnarray}
where 
the inequality (\ref{ine1}) is obtained from the Schwarz inequality.

To evaluate (\ref{ine1}),
we will use Lemma \ref{naiseki2}.
To do so,
we set a sequence as 
\begin{eqnarray}
a_0:=\Phi(t),& a_i:=\Phi(x_i^I)-\Phi(x_{i-1}^I),& a_{I+1}:=1-\Phi(t'),\\
b_0:=\Phi_{\mu,v}(t),& b_i:=\Phi_{\mu,v}(x_i^I)-\Phi_{\mu,v}(x_{i-1}^I),& b_{I+1}:=1-\Phi_{\mu,v}(t'),\\
c_0:=A(t),& c_i:=A(x_i^I)-A(x_{i-1}^I),& c_{I+1}:=1-A(t'),
\end{eqnarray}
for $i=1,2,...,I$.
Then, we can verify that the sequences
${\bf a}=\{a_i\}_{i=0}^{I+1}$, ${\bf b}=\{b_i\}_{i=0}^{I+1}$ and ${\bf c}=\{c_i\}_{i=0}^{I+1}$ satisfy the assumptions in Lemma \ref{naiseki2} as follows. 
First, the sequences $a, b$ and $c$ are probability distributions by the definitions.
Second, ${\bf a}$ and ${\bf c}$ satisfy (\ref{gijimajo}) since the function $A$ satisfies $\Phi\le A$ by the definition.
In the following, we show (\ref{ratio}) for $i=1,2,...,I+1$.
For $i=1$, we have
\begin{eqnarray}
\frac{a_1}{b_1}
~=~\frac{\Phi(t+\frac{t'-t}{I})-\Phi(t)}{\Phi_{\mu,v}(t+\frac{t'-t}{I})-\Phi_{\mu,v}(t)}
~=~\frac{\phi(u_1)}{\phi_{\mu,v}(u_1)}
~<~\frac{\phi(s)}{\phi_{\mu,v}(s)}
~=~\frac{\Phi(t)}{\Phi_{\mu,v}(t)}
~=~\frac{a_0}{b_0}
\end{eqnarray}
where the existence of a real number $u_1\in(t,t+\frac{t'-t}{I})$ is guaranteed by the mean value theorem in the second equality,
the inequality follows from the conditions (I) and (III) in the assumption of Lemma \ref{lem.converse},
and the third equality comes from the conditions (II).
For $i=2,3,...,I$, we have
\begin{eqnarray*}
\frac{a_i}{b_i}
~=~\frac{\Phi(x_i^I)-\Phi(x_{i-1}^I)}{\Phi_{\mu,v}(x_i^I)-\Phi_{\mu,v}(x_{i-1}^I)}
~=~\frac{\phi(u_i)}{\phi_{\mu,v}(u_i)}
~<~\frac{\phi(u_{i-1})}{\phi_{\mu,v}(u_{i-1})}
~=~\frac{\Phi(x_{i-1}^I)-\Phi(x_{i-2}^I)}{\Phi_{\mu,v}(x_{i-1}^I)-\Phi_{\mu,v}(x_{i-2}^I)}
~=~\frac{a_{i-1}}{b_{i-1}}
\end{eqnarray*}
where the existence of real numbers $u_i\in(x_{i-1}^I,x_{i}^I)$ and $u_{i-1}\in(x_{i-2}^I,x_{i-1}^I)$ is guaranteed by the mean value theorem in the second and third equality,
the inequality follows from the conditions (I) and (III) in the assumption of Lemma \ref{lem.converse}.
%
For $i=I+1$, we have
\begin{eqnarray}
\frac{a_{I+1}}{b_{I+1}}
~=~\frac{1-\Phi(t')}{1-\Phi_{\mu,v}(t')}
~=~\frac{\phi(s')}{\phi_{\mu,v}(s')}
~<~\frac{\phi(u_{I+1})}{\phi_{\mu,v}(u_{I+1})}
~=~\frac{\Phi(t')-\Phi(t'-\frac{t'-t}{I})}{\Phi_{\mu,v}(t')-\Phi_{\mu,v}(t'-\frac{t'-t}{I})}
~=~\frac{a_I}{b_I}
\end{eqnarray}
where 
the second equality comes from the conditions (II),
the existence of a real number $u_{I+1}\in(t'-\frac{t'-t}{I},t')$ is guaranteed by the mean value theorem in the third equality
and 
the inequality follows from the conditions (I) and (III) in the assumption of Lemma \ref{lem.converse}.
From the above discussion,
the sequences
${\bf a}$, ${\bf b}$ and ${\bf c}$ satisfy the assumptions in Lemma \ref{naiseki2}.

Using Lemma \ref{naiseki2},
we obtain the following ineqality:
\begin{eqnarray}
&&\sqrt{A(t)} \sqrt{{\Phi}_{\mu,v}(t)} + \sqrt{1-A(t')} \sqrt{1-{\Phi}_{\mu,v}(\lambda)} \nonumber\\
&&+\sum_{i=1}^{I} \sqrt{A(x_i^I)-A(x_{i-1}^I)} \sqrt{\Phi_{\mu,v}(x_i^I)-\Phi_{\mu,v}(x_{i-1}^I)} \nonumber\\
&=&\sum_{i=0}^{I+1} \sqrt{c_i}\sqrt{b_i}\\
&\le&\sum_{i=0}^{I+1} \sqrt{a_i}\sqrt{b_i}\\
&=&
\sqrt{{\Phi}(t)} \sqrt{{\Phi}_{\mu,v}(t)} +\sqrt{1-{\Phi}(t')} \sqrt{1-{\Phi}_{\mu,v}(t')} \Label{ine2}\\
&&+\sum_{i=1}^{I} \sqrt{\Phi(x_i^I)-\Phi(x_{i-1}^I)} \sqrt{\Phi_{\mu,v}(x_i^I)-\Phi_{\mu,v}(x_{i-1}^I)} \nonumber\\
&=&
\sqrt{{\Phi}(t)} \sqrt{{\Phi}_{\mu,v}(t)} +\sqrt{1-{\Phi}(t')} \sqrt{1-{\Phi}_{\mu,v}(t')} \\
&&+\sum_{i=1}^{I} \sqrt{\frac{\Phi(x_i^I)-\Phi(x_{i-1}^I)}{x_{i}^I-x_{i-1}^I}} \sqrt{\frac{\Phi_{\mu,v}(x_i^I)-\Phi_{\mu,v}(x_{i-1}^I)}{x_{i}^I-x_{i-1}^I}}(x_{i}^I-x_{i-1}^I) \nonumber\\
&\overset{I\to\infty}{\longrightarrow}&
\sqrt{{\Phi}(t)} \sqrt{{\Phi}_{\mu,v}(t)}
+\sqrt{1-{\Phi}(t')} \sqrt{1-{\Phi}_{\mu,v}(t')}\\
&&~~~+\int_{t}^{t'}\sqrt{\phi(x)} \sqrt{\phi_{\mu,v}(x)}dx.\nonumber
\end{eqnarray}
\endproof

~

\noindent{\bf Proof of Lemma \ref{lem.converse2}:}
We separately prove Lemma \ref{lem.converse2} for four cases; (i) $v=1$ and $\mu\le0$, (ii) $v=1$ and $\mu>0$, (iii) $v>1$ and (iv) $v<1$.

\noindent{\it Proof of the case (i):}\quad
By the Schwarz inequality,
the left hand side of (\ref{lem.con.ineq3}) is less than or equal to $1$.
When $\mu\le0$,
since $A_{\mu,1}=\Phi_{\mu,1}$ satisfies ${\cal F}\left(\frac{dA_{\mu,1}}{dx}, \phi_{\mu,1}\right)=1$,
(\ref{lem.con.ineq3}) holds.

\vspace{0.5em}

\noindent{\it Proof of the case (ii):}\quad
For an arbitrary $\epsilon>0$,
we take an arbitrary real number $\lambda>0$ which satisfies
\begin{eqnarray}
\sqrt{{\Phi}(-\lambda)}\sqrt{{\Phi}_{\mu,v}(-\lambda)}
+\sqrt{1-{\Phi}(\lambda)}\sqrt{1-{\Phi}_{\mu,v}(\lambda)}
~<~ \epsilon
\end{eqnarray}
and show that $t:=-\lambda$ and $t':=\lambda$ satisfy the condition ($\star$) in Lemma \ref{lem.converse}  and (\ref{lem.con.ineq3}).

We have (\ref{lem.con.ineq3}) as follows:
\begin{eqnarray}
&&\sqrt{{\Phi}(-\lambda)} \sqrt{{\Phi}_{\mu,v}(-\lambda)}
+
\int_{-\lambda}^{\lambda}\sqrt{\phi(x)} \sqrt{\phi_{\mu,v}(x)}dx 
+\sqrt{1-{\Phi}(\lambda)} \sqrt{1-{\Phi}_{\mu,v}(\lambda)}\nonumber\\
&\le& \int_{\R}\sqrt{\phi(x)} \sqrt{\phi_{\mu,v}(x)}dx+\epsilon\\
&=&{\cal F}\left(\frac{dA_{\mu,1}}{dx}, \phi_{\mu,v}\right)+\epsilon,\Label{Aeq}
\end{eqnarray}
where we used $A_{\mu,1}=\Phi$ in (\ref{Aeq}).

Next, we show the condition in ($\star$) of Lemma \ref{lem.converse} for $t=-\lambda$ and $t'=\lambda$.
To do so,
it is enough to show that there exist $s$ and $s'$ which satisfy (I), (II) and (III) in ($\star$) of Lemma \ref{lem.converse}.
Since $\frac{\phi(x)}{\phi_{\mu,v}(x)}$ is continuous and strictly monotonically decreasing on $\R$ from Lemma \ref{monotone} 
and equations
\begin{eqnarray}
\lim_{x\to-\infty}\frac{\phi(x)}{\phi_{\mu,v}(x)}=\infty,~~ \lim_{x\to\infty}\frac{\phi(x)}{\phi_{\mu,v}(x)}=0
\end{eqnarray}
holds for $\mu>0$,
there uniquely exist real numbers $s$ and $s'$ such that \begin{eqnarray}
&\frac{{\Phi}(-\lambda)}{{\Phi}_{\mu,v}(-\lambda)}=\frac{\phi(s)}{\phi_{\mu,v}(s)},~~\frac{1-{\Phi}(\lambda)}{1-{\Phi}_{\mu,v}(\lambda)}=\frac{\phi(s')}{\phi_{\mu,v}(s')}&
\Label{defs}
\end{eqnarray}
by the intermediate value theorem.

In the following,
we prove that the above $s$ and $s'$ satisfy (I), (II) and (III) in ($\star$) of Lemma \ref{lem.converse}.
The condition (III) in ($\star$) is verified from Lemma \ref{monotone}.
The condition (II) in ($\star$) is obtained by the definition (\ref{defs}) of $s$ and $s'$.
The condition (I) in ($\star$), i.e. $s\le-\lambda<\lambda\le s'$ is obtained as follows.
To show $s<-\lambda$, it is enough to prove that $\frac{\phi(s)}{\phi_{\mu,v}(s)}\ge\frac{\phi(-\lambda)}{\phi_{\mu,v}(-\lambda)}$ because $\frac{\phi(x)}{\phi_{\mu,v}(x)}$ is monotonically decreasing on $\R$ by Lemma \ref{monotone}.
We have
\begin{eqnarray}
\frac{\phi(s)}{\phi_{\mu,v}(s)}
&=&\frac{{\Phi}(-\lambda)}{{\Phi}_{\mu,v}(-\lambda)}\nonumber\\
&=&\lim_{w\to-\infty}\frac{{\Phi}(-\lambda)-{\Phi}(w)}{{\Phi}_{\mu,v}(-\lambda)-{\Phi}_{\mu,v}(w)}\nonumber\\
&=&\lim_{w\to-\infty}\frac{\phi(s_w)}{\phi_{\mu,v}(s_w)}\Label{mean}\\
&\ge&\frac{\phi(-\lambda)}{\phi_{\mu,v}(-\lambda)}
\Label{phiine}
\end{eqnarray}
where the  existence of $s_w\in(w,-\lambda)$ in (\ref{mean}) is guaranteed by the mean value theorem and the inequality (\ref{phiine}) holds since $\frac{\phi(x)}{\phi_{\mu,v}(x)}$ is monotonically decreasing.
Thus, $s\le-\lambda$ holds. 
Similarly, $\lambda\le s'$ can be obtained by the following inequality:
\begin{eqnarray}
\frac{\phi(s')}{\phi_{\mu,v}(s')}
&=&\frac{{1-\Phi}(\lambda)}{1-{\Phi}_{\mu,v}(\lambda)}\nonumber\\
&=&\lim_{w\to\infty}\frac{{\Phi}(w)-{\Phi}(\lambda)}{{\Phi}_{\mu,v}(w)-{\Phi}_{\mu,v}(\lambda)}\nonumber\\
&=&\lim_{w\to\infty}\frac{\phi(s'_w)}{\phi_{\mu,v}(s'_w)}\Label{mean2}\\
&\le&\frac{\phi(\lambda)}{\phi_{\mu,v}(\lambda)}
\Label{phiine2}
\end{eqnarray}
where the  existence of $s'_w\in(\lambda,w)$ in (\ref{mean2}) is guaranteed by the mean value theorem and the inequality (\ref{phiine2}) holds since $\frac{\phi(x)}{\phi_{\mu,v}(x)}$ is monotonically decreasing.
Therefore, we obtained the condiiton (I) in ($\star$) of Lemma \ref{lem.converse}.
Thus, the proof is completed for the case (ii).

\vspace{0.5em}

\noindent{\it Proof of the case (iii):}\quad
We take a constant $\lambda>0$ which satisfies $\alpha_{\mu,v}<\lambda$ and 
\begin{eqnarray}
\sqrt{1-{\Phi}(\lambda)}\sqrt{1-{\Phi}_{\mu,v}(\lambda)}
~<~ \epsilon.
\Label{epine3}
\end{eqnarray}
We show that $t:=\alpha_{\mu,v}$ and $t':=\lambda$ satisfy the condition ($\star$) of Lemma \ref{lem.converse} and (\ref{lem.con.ineq2}).

We have (\ref{lem.con.ineq2}) as follows:
\begin{eqnarray}
&&\sqrt{\Phi(\alpha_{\mu,v})}\sqrt{\Phi_{\mu,v}(\alpha_{\mu,v})}
+\int_{\alpha_{\mu,v}}^{\lambda}\sqrt{\phi(x)}\sqrt{\phi_{\mu,v}(x)}dx 
+\sqrt{1-\Phi(\lambda)}\sqrt{1-\Phi_{\mu,v}(\lambda)}\nonumber\\
&\le&\sqrt{\Phi(\alpha_{\mu,v})}\sqrt{\Phi_{\mu,v}(\alpha_{\mu,v})}
+\int_{\alpha_{\mu,v}}^{\infty}\sqrt{\phi(x)}\sqrt{\phi_{\mu,v}(x)}dx+\epsilon, \nonumber\\
&=&{\cal F}\left(\frac{dA_{\mu,v}}{dx}, \phi_{\mu,v}\right)+\epsilon,
\Label{Aeq2}
\end{eqnarray}
where the equation (\ref{Aeq2}) comes from the definition (\ref{A_2}) of $A_{\mu,v}$.

Then, we show the condition in ($\star$) of Lemma \ref{lem.converse} for $t=\alpha_{\mu,v}$ and $t'=\lambda$.
To do so,
we assume the the existence of $s'$ which satisfies
\begin{eqnarray}
\frac{1-{\Phi}(\lambda)}{1-{\Phi}_{\mu,v}(\lambda)}
~=~ \frac{\phi(s')}{\phi_{\mu,v}(s')}.
\Label{defs2}
\end{eqnarray}
Then we can easily show that $s:=t=\alpha_{\mu,v}$ and the above $s'$ satisfy (I), (II) and (III) in ($\star$) of Lemma \ref{lem.converse} as follows.
%
%
The conditions (I) and (II) in ($\star$) are obtained by the definitions (\ref{threshold1}) and  (\ref{defs2}).
The condition (III) in ($\star$) is verified since $\frac{\phi(x)}{\phi_{\mu,v}(x)}$ is monotonically decreasing on $(\alpha_{\mu,v}, s')$ from Lemma \ref{monotone} and Lemma \ref{sol1}.

Thus, all we have to do is to show (\ref{defs2}).
We have 
\begin{eqnarray}
0~<~
\frac{{1-\Phi}(\lambda)}{1-{\Phi}_{\mu,v}(\lambda)}
&=&\lim_{w\to\infty}\frac{{\Phi}(w)-{\Phi}(\lambda)}{{\Phi}_{\mu,v}(w)-{\Phi}_{\mu,v}(\lambda)}\nonumber\\
&=&\lim_{w\to\infty}\frac{\phi(s'_w)}{\phi_{\mu,v}(s'_w)}\Label{mean3}\\
&\le&\frac{\phi(\lambda)}{\phi_{\mu,v}(\lambda)}
\Label{phiine3}
\end{eqnarray}
where the  existence of $s'_w\in(\lambda,w)$ in (\ref{mean3}) is guaranteed by the mean value theorem and the inequality (\ref{phiine3}) holds since $\frac{\phi(x)}{\phi_{\mu,v}(x)}$ is monotonically decreasing on $(\lambda,\infty)$ from Lemma \ref{monotone}
and $\frac{\mu}{1-v}<0<\lambda$.
Moreover, the equation
\begin{eqnarray}
\lim_{x\to\infty}\frac{\phi(x)}{\phi_{\mu,v}(x)}=0
\Label{phieq3}
\end{eqnarray}
holds due to $v>1$.
%
Since $\frac{\phi(x)}{\phi_{\mu,v}(x)}$ is continuous and strictly monotonically decreasing on $(\lambda,\infty)$, 
there uniquely exists a real number $s'\ge \lambda$ which satisfies (\ref{defs2})
by  (\ref{phiine3}), (\ref{phieq3}) and the intermediate value theorem.
%
%
Thus, the proof is completed for the case (iii).

\vspace{0.5em}

\noindent{\it Proof of the case (iv):}\quad
We take a constant $\lambda<0$ which satisfies $\lambda<\beta_{\mu,v}$ and 
\begin{eqnarray}
\sqrt{\Phi(\lambda)}\sqrt{\Phi_{\mu,v}(\lambda)}<\epsilon.
\Label{epine4}
\end{eqnarray}
To use Lemma \ref{lem.converse}, 
we verify that $t:=\lambda$ and $t':=\beta_{\mu,v}$ satisfy the condition ($\star$) of Lemma \ref{lem.converse} and (\ref{lem.con.ineq2}) in the following.

We have (\ref{lem.con.ineq2}) as follows:
\begin{eqnarray}
&&\sqrt{\Phi(\lambda)}\sqrt{\Phi_{\mu,v}(\lambda)}
+\int_{\lambda}^{\beta_{\mu,v}}\sqrt{\phi(x)}\sqrt{\phi_{\mu,v}(x)}dx
+\sqrt{1-\Phi(\beta_{\mu,v})}\sqrt{1-\Phi_{\mu,v}(\beta_{\mu,v})}\nonumber\\
&\le&\int_{-\infty}^{\beta_{\mu,v}}\sqrt{\phi(x)}\sqrt{\phi_{\mu,v}(x)}dx
+\sqrt{1-\Phi(\beta_{\mu,v})}\sqrt{1-\Phi_{\mu,v}(\beta_{\mu,v})}+\epsilon, \nonumber\\
&=&{\cal F}\left(\frac{dA_{\mu,v}}{dx}, \phi_{\mu,v}\right)+\epsilon.
\end{eqnarray}

Then, we show the condition in ($\star$) of Lemma \ref{lem.converse} for $t=\lambda$ and $t'=\beta_{\mu,v}$.
To do so,
we assume the the existence of $s\le \lambda$ which satisfies
\begin{eqnarray}
\frac{{\Phi}(\lambda)}{{\Phi}_{\mu,v}(\lambda)}=\frac{\phi(s)}{\phi_{\mu,v}(s)}.
\Label{defs4}
\end{eqnarray}
Then we can easily show that the above $s$ and $s':=t'=\beta_{\mu,v}$ satisfy (I), (II) and (III) in ($\star$) of Lemma \ref{lem.converse} as follows.
The conditions (I) and (II) in ($\star$) are obtained by the definitions (\ref{threshold2}) and  (\ref{defs4}).
The condition (III) in ($\star$) is verified since $\frac{\phi(x)}{\phi_{\mu,v}(x)}$ is monotonically decreasing on $(s,\beta_{\mu,v})$ from Lemma \ref{monotone} and Lemma \ref{sol2}.

Thus, all we have to do is to show (\ref{defs4}).
We have 
\begin{eqnarray}
\frac{{\Phi}(\lambda)}{{\Phi}_{\mu,v}(\lambda)}
&=&\lim_{w\to-\infty}\frac{{\Phi}(\lambda)-{\Phi}(w)}{{\Phi}_{\mu,v}(\lambda)-{\Phi}_{\mu,v}(w)}\nonumber\\
&=&\lim_{w\to\infty}\frac{\phi(s_w)}{\phi_{\mu,v}(s_w)}\Label{mean4}\\
&\ge&\frac{\phi(\lambda)}{\phi_{\mu,v}(\lambda)}
\Label{phiine4}
\end{eqnarray}
where the  existence of $s_w\in(-\infty, \lambda)$ in (\ref{mean4}) is guaranteed by the mean value theorem and the inequality (\ref{phiine4}) holds since $\frac{\phi(x)}{\phi_{\mu,v}(x)}$ is monotonically decreasing on $(-\infty,\lambda)$ from Lemma \ref{monotone} and $\lambda<0<\frac{\mu}{1-v}$.
Moreover, the equation
\begin{eqnarray}
\lim_{x\to-\infty}\frac{\phi(x)}{\phi_{\mu,v}(x)}=\infty
\Label{phieq4}
\end{eqnarray}
holds due to $v<1$.
%
Since $\frac{\phi(x)}{\phi_{\mu,v}(x)}$ is continuous and strictly monotonically decreasing on $(-\infty, \lambda)$, 
there uniquely exists a real number $s\le \lambda$ which satisfies (\ref{defs4})
by  (\ref{phiine4}), (\ref{phieq4}) and the intermediate value theorem.
Thus, the proof is completed for the case (iv).
\endproof


~

\subsection{Proof of Propositions: Properties of The Rayleigh-normal Distributions}
\Label{app.properties}
We prove Propositions \ref{Sym}, \ref{F3lim} and \ref{cum} which shows basic properties of the Rayleigh-normal distributions.

~

\noindent{\bf Proof of Proposition \ref{Sym}:}
We first show
\begin{eqnarray}
\alpha_{\mu,v}
~=~-\sqrt{v}\beta_{\frac{\mu}{\sqrt{v}},\frac{1}{v}} + \mu
\Label{relation}
\end{eqnarray}
for $v>1$.
We define a function $g$ on $\R$ as $g(x):=\sqrt{v}x-\mu$.
Since the solution $\alpha_{\mu,v}$ of the equation (\ref{threshold1}) is unique,
we only have to show that $-g(\beta_{\frac{\mu}{\sqrt{v}},\frac{1}{v}})$ is a solution of the equation (\ref{threshold1}).
Note that the function $g$ is represented as 
\begin{eqnarray}
g(x)
~=~\Phi_{-\mu,v}^{-1}\circ\Phi(x)
~=~\Phi^{-1}\circ\Phi_{\frac{\mu}{\sqrt{v}},\frac{1}{v}}(x).
\Label{g}
\end{eqnarray}
Then, the following holds
\begin{eqnarray*}
\frac{\Phi(-g(\beta_{\frac{\mu}{\sqrt{v}},\frac{1}{v}}))}{\Phi_{\mu,v}(-g(\beta_{\frac{\mu}{\sqrt{v}},\frac{1}{v}}))}
&=\frac{1-\Phi(g(\beta_{\frac{\mu}{\sqrt{v}},\frac{1}{v}}))}{1-\Phi_{-\mu,v}(g(\beta_{\frac{\mu}{\sqrt{v}},\frac{1}{v}}))}
&=\frac{1-\Phi_{\frac{\mu}{\sqrt{v}},\frac{1}{v}}(\beta_{\frac{\mu}{\sqrt{v}},\frac{1}{v}})}{1-\Phi(\beta_{\frac{\mu}{\sqrt{v}},\frac{1}{v}})}\\
&=\frac{\phi_{\frac{\mu}{\sqrt{v}},\frac{1}{v}}(\beta_{\frac{\mu}{\sqrt{v}},\frac{1}{v}})}{\phi(\beta_{\frac{\mu}{\sqrt{v}},\frac{1}{v}})}
&=\frac{\frac{d\Phi\circ g}{dx}(\beta_{\frac{\mu}{\sqrt{v}},\frac{1}{v}})}{\frac{d\Phi_{-\mu,v}\circ g}{dx}(\beta_{\frac{\mu}{\sqrt{v}},\frac{1}{v}})}\\
&=\frac{\phi(g(\beta_{\frac{\mu}{\sqrt{v}},\frac{1}{v}}))}{\phi_{\mu,v}(g(\beta_{\frac{\mu}{\sqrt{v}},\frac{1}{v}}))}
&=\frac{\phi(-g(\beta_{\frac{\mu}{\sqrt{v}},\frac{1}{v}}))}{\phi_{\mu,v}(-g(\beta_{\frac{\mu}{\sqrt{v}},\frac{1}{v}}))}.
\end{eqnarray*}
Thus, we have (\ref{relation}).
%
From (\ref{g}),
we have
\begin{eqnarray}
\Phi(\alpha_{\mu,v})
&=&1-\Phi_{\frac{\mu}{\sqrt{v}},\frac{1}{v}}(\beta_{\frac{\mu}{\sqrt{v}},\frac{1}{v}}),\\
\Phi_{\mu,v}(\alpha_{\mu,v})
&=&1-\Phi(\beta_{\frac{\mu}{\sqrt{v}},\frac{1}{v}}).
\end{eqnarray}
From direct calculation,
\begin{eqnarray}
I_{\mu,v}(\infty) - I_{\mu,v}(\alpha_{\mu,v})
&=& I_{\frac{\mu}{\sqrt{v}},\frac{1}{v}}(\beta_{\frac{\mu}{\sqrt{v}},\frac{1}{v}})
\Label{I-eq}
\end{eqnarray}
holds, and thus,
we obtain (\ref{sym}).
\endproof

~

\noindent{\bf Proof of Proposition \ref{F3lim}:}
Since the first equation obviously holds because of Proposition \ref{Sym},
we show only $\lim_{v\to 0}Z_{v}(\mu)=\Phi\left(\mu\right)$ in the following.

First, 
the definition in (\ref{Ix}) implies 
\begin{eqnarray}
0
~\le~ I_{\mu,v}(\beta_{\mu,v})
~\le~ I_{\mu,v}(\infty)
\overset{v\to 0}{\longrightarrow}0.
\Label{limI2}
\end{eqnarray}
Moreover, as shown below,
\begin{eqnarray}
\lim_{v\to 0}\Phi(\beta_{\mu,v}) &=&\Phi\left(\mu\right)
\Label{lim3'} \\
\lim_{v\to 0}\Phi_{\mu,v}(\beta_{\mu,v})&=&0.
\Label{lim3}
\end{eqnarray}
The relations (\ref{limI2}), (\ref{lim3'}) and (\ref{lim3}) yield the equation $\lim_{v\to 0}Z_{v}(\mu)=\Phi\left(\mu\right)$.
%

Thus, all we have to do is to show (\ref{lim3'}) and (\ref{lim3}).
First, we will show (\ref{lim3'}).
In order to show (\ref{lim3'}), it is enough to prove that $\lim_{v\to0}\beta_{\mu,v}=\mu$. 
First, we have $\beta_{\mu,v}<\frac{\mu}{1-v}$ from Lemma \ref{sol2}.
Since $\lim_{v\to0}\frac{\mu}{1-v}=\mu$, we obtain ${\rm limsup}_{v\to0}\beta_{\mu,v}\le\mu$.
Next, we set the function $f_{\mu,v}(x)$ as $(1-\Phi_{\mu,v}(x))-(1-\Phi(x))\frac{\phi_{\mu,v}(x)}{\phi(x)}$ and take an arbitrary $x\in\R$ such that $x<\mu$.
Since $\lim_{v\to0}\phi_{\mu,v}(x)=0$ and $\lim_{v\to0}\Phi_{\mu,v}(x)\le\frac{1}{2}$, $\lim_{v\to0}f_{\mu,v}(x)>0$ holds, in other words, $x$ is not a zero point of $f_{\mu,v}$ when $v$ is close to $0$.
Thus, we obtain $\lim_{v\to0}\beta_{\mu,v}\ge \mu$ since $\beta_{\mu,v}$ is the unique zero point of $f_{\mu,v}$.
Therefore, $\lim_{v\to0}\beta_{\mu,v}=\mu$ holds.

Next, we will show (\ref{lim3}).
In order to show (\ref{lim3}), it is enough to prove that $\lim_{v\to0}\frac{\beta_{\mu,v}-\mu}{\sqrt{v}}=-\infty$ by the definition of $\Phi_{\mu,v}$.
Since $\beta_{\mu,v}<\frac{\mu}{1-v}$ and $\lim_{v\to0}\frac{\mu}{1-v}=\mu$, $\beta_{\mu,v}$ is bounded above by some constant $\gamma$ as $\beta_{\mu,v}<\gamma$ when $v$ is close to $0$, and then, we have the following inequality:
\begin{eqnarray}
\frac{\phi(\beta_{\mu,v})}{\phi_{\mu,v}(\beta_{\mu,v})}
~=~\frac{1-\Phi(\beta_{\mu,v})}{1-\Phi_{\mu,v}(\beta_{\mu,v})} 
~\ge~ \Phi(-\gamma).
\end{eqnarray}
Thus, the following holds
\begin{eqnarray}
2\log \Phi(-\gamma)
&\le&-\mathrm{log}\frac{\phi(\beta_{\mu,v})}{\phi_{\mu,v}(\beta_{\mu,v})}\nonumber\\
&=&(1-v)
\left(
\frac{\beta_{\mu,v}-\mu(1-v)^{-1}}{\sqrt{v}}
\right)^2
+\mathrm{log}v
-\frac{\mu^2}{1-v}.
\Label{inequality12}
\end{eqnarray}
Therefore, we have 
\begin{eqnarray}
\lim_{v\to0}\left(\frac{\beta_{\mu,v}-\mu(1-v)^{-1}}{\sqrt{v}}\right)^2
~=~\lim_{v\to0}(1-v)^{-1}\left(2\log \Phi(-\gamma)-\log v + \frac{\mu^2}{1-v}\right)
~=~\infty.
\end{eqnarray}
Since Lemma \ref{sol2} guarantees that
$\beta_{\mu,v}-\frac{\mu}{1-v}<0$, 
we obtain
\begin{eqnarray}
\lim_{v\to0}\frac{\beta_{\mu,v}-\mu}{\sqrt{v}}
~=~\lim_{v\to0}\frac{\beta_{\mu,v}-\mu(1-v)^{-1}}{\sqrt{v}}
~=~ -\infty.
\end{eqnarray}
\endproof

~

\noindent{\bf Proof of Proposition \ref{cum}:}
From Proposition \ref{Sym}, it is enough to treat the case when $0\le v<1$.
This proposition is obvious for $v=0$ by the definition $Z_0:=\Phi$. 
In the following, we fix $0<v<1$
and  show that $Z_v$ satisfies the definition of a cumulative distribution function, that is,
the right continuous, monotonically increasing, 
\begin{eqnarray}
\displaystyle\lim_{\mu\to-\infty}Z_{v}(\mu)=0,
\Label{-infty}
\end{eqnarray}
and 
\begin{eqnarray}
\displaystyle\lim_{\mu\to\infty}Z_{v}(\mu)=1.
\Label{+infty}
\end{eqnarray}

First, we show that $Z_v(\mu)$ is continuous.
From Lemma \ref{sol2}, $\beta_{\mu,v}$ is differentiable, especially continuous, with respect to $\mu$.
Thus, 
$Z_v(\mu)$ is continuous from Theorem \ref{Zform}.

Next we show (\ref{-infty}).
From Lemma \ref{sol2}, 
the inequality 
\begin{eqnarray*}
\beta_{\mu,v}<\frac{\mu}{1-v}
\end{eqnarray*}
holds and thus 
\begin{eqnarray*}
\displaystyle\lim_{\mu\to-\infty}\beta_{\mu,v}=-\infty.
\end{eqnarray*}
Similarly, from Lemma \ref{sol1},
the inequality 
\begin{eqnarray*}
\alpha_{\frac{\mu}{\sqrt{v}},\frac{1}{v}}>\frac{\mu}{\sqrt{v}(1-v^{-1})}
\end{eqnarray*}
holds and thus 
\begin{eqnarray*}
\displaystyle\lim_{\mu\to-\infty}\alpha_{\frac{\mu}{\sqrt{v}},\frac{1}{v}}=\infty.
\end{eqnarray*}
Therefore, we obtain 
\begin{eqnarray*}
&&\displaystyle\lim_{\mu\to-\infty}\Phi(\beta_{\mu,v})~=~0,\\
&&\displaystyle\lim_{\mu\to-\infty}\Phi_{\mu,v}(\beta_{\mu,v})
~=~ \displaystyle\lim_{\mu\to-\infty}\Phi(-\alpha_{\frac{\mu}{\sqrt{v}},\frac{1}{v}})
~=~0.
\end{eqnarray*}
Since 
\begin{eqnarray*}
\displaystyle\lim_{\mu\to-\infty}I_{\mu,v}(\beta_{\mu,v})
~\le~ \displaystyle\lim_{\mu\to-\infty}I_{\mu,v}(\infty)
~=~ 0,
\end{eqnarray*}
we have (\ref{-infty}) from Theorem \ref{Zform}.

Next we show (\ref{+infty}).
From (\ref{relation}), 
the equality 
\begin{eqnarray*}
\beta_{\mu,v}=-\sqrt{v}\alpha_{\frac{\mu}{\sqrt{v}},\frac{1}{v}}+\mu
\end{eqnarray*}
holds.
Since $\alpha_{\mu,v}$ monotonically decreases with respect to $\mu$, 
we have 
\begin{eqnarray*}
\displaystyle\lim_{\mu\to\infty}\beta_{\mu,v}=\infty.
\end{eqnarray*}
Therefore, we obtain 
\begin{eqnarray*}
\displaystyle\lim_{\mu\to\infty}\Phi(\beta_{\mu,v})=1.
\end{eqnarray*}
Since 
\begin{eqnarray*}
\displaystyle\lim_{\mu\to-\infty}I_{\mu,v}(\beta_{\mu,v})
~\le~ \displaystyle\lim_{\mu\to-\infty}I_{\mu,v}(\infty)
~=~ 0,
\end{eqnarray*}
we have (\ref{+infty}) from Theorem \ref{Zform}.

Finally, we show that $Z_{v}(\mu)$ is monotonically increasing.
We define a shift operator $S_{\mu}$ for a map $A:\R\to\R$ by 
\begin{eqnarray*}
(S_{\mu}A)(x):=A(x-\mu).
\end{eqnarray*}
Then we have 
\begin{eqnarray*}
{\cal F}(S_{\mu}p,S_{\mu}q)={\cal F}(p,q).
\end{eqnarray*}
Thus when we set as 
\begin{eqnarray*}
{\cal A}(\mu):=\left\{A:\R\to[0,1]\Big|
\begin{array}{l}
A \text{ is a continuous differentiable monotone increasing}\\
\text{function such that} ~\Phi_{\mu,1}\le A\le 1
\end{array}
\right\},
\end{eqnarray*}
we obtain the following form of the Rayleigh-normal distribution function
\begin{eqnarray*}
Z_{v}(\mu)
&:=&1-\sup_{A\in{\cal A}(0)}{\cal F}\left(\frac{dA}{dx}, \phi_{\mu,v}\right)^2
~=~1-\sup_{A\in{\cal A}(0)}{\cal F}\left(S_{-\mu}\frac{dA}{dx}, S_{-\mu}\phi_{\mu,v}\right)^2\\
&=&1-\sup_{A\in{\cal A}(0)}{\cal F}\left(\frac{d(S_{-\mu}A)}{dx}, \phi_{0,v}\right)^2
~=~1-\sup_{\tilde{A}\in{\cal A}(-\mu)}{\cal F}\left(\frac{d\tilde{A}}{dx}, \phi_{0,v}\right)^2.
\end{eqnarray*}
For $\mu<\tau$,  ${\cal A}(-\mu)\supset{\cal A}(-\tau)$ holds, and thus we obtain $Z_{v}(\mu)\le Z_{v}(\tau)$.
\endproof

\subsection{Proof of Proposition: Optimal Majorization Conversion in Non-Asymptotic Theory} 
\Label{app.OPT}
We prove Proposition \ref{opt con} which gives the maximum fidelity of  the optimal majorization conversion from a non-uniform distribution  to the uniform distribution.

~

\noindent{\bf Proof of Proposition \ref{opt con}:}
We first verify that $J_{P,m}$ is well defined, i.e., the set in the definition of $J_{P,m}$
\begin{eqnarray}
\left\{j \in \{2, \ldots, m\} 
~\left|~ \frac{\sum_{i=j}^{|\mathcal{X}|} P^{\downarrow}(i)}{m+1-j}<P^{\downarrow}(j-1)\right.\right\}
\Label{set220}
\end{eqnarray}
is not empty when $P^{\downarrow}(1)> \frac{1}{m}$.
Indeed, $j=2$ is included in the set (\ref{set220}) as follows: 
\begin{eqnarray}
\frac{\sum_{i=2}^{|\mathcal{X}|} P^{\downarrow}(i)}{m+1-2}
~=~ \frac{1-P^{\downarrow}(1)}{m-1}
~<~ P^{\downarrow}(1)
~=~ P^{\downarrow}(2-1)
\end{eqnarray}
where the inequality follows from $P^{\downarrow}(1)> \frac{1}{m}$.

We show the followings:
\begin{eqnarray}
P&\prec& {\cal C}_m(P),
\Label{precC}\\
F^{\mathcal{M}}(P\to U_m)
&=&F(\mathcal{C}_m(P), U_m).
\Label{eqU}
\end{eqnarray} 

We prove (\ref{precC}).
It is enough to show that $P^{\downarrow}(j) \le {\cal C}_m(P)(j)$ for $J_{P,m}\le j \le m$.
Since 
\begin{eqnarray}
P^{\downarrow}(J_{P,m})
~\le~ \frac{\sum_{i=J_{P,m}+1}^{|\mathcal{X}|} P^{\downarrow}(i)}{m-J_{P,m}}
\nonumber
\end{eqnarray}
holds by the definition of $J_{P,m}$,
we have
\begin{eqnarray}
P^{\downarrow}(j)
~\le~ P^{\downarrow}(J_{P,m})
~\le~ \frac{\sum_{i=J_{P,m}}^{|\mathcal{X}|} P^{\downarrow}(i)}{m+1-J_{P,m}}
~=~{\cal C}_m(P)(j)
\nonumber
\end{eqnarray}
for $J_{P,m}\le j \le m$.

Next, we show (\ref{eqU}).
To do so,
we take a probability distribution $Q=(Q(1), ..., Q(m))$ which satisfies $P\prec Q$, $Q(i)\ge Q(i+1)$ and 
\begin{eqnarray}
F^{\mathcal{M}}(P\to U_m)
=F(Q, U_m).\Label{Qmax}
\end{eqnarray}
In the following, we assume that $Q\ne \mathcal{C}_m(P)$ and derive contradiction.
Then, we will have $Q=\mathcal{C}_m(P)$, 
and thus,
$\mathcal{C}_m(P)$ satisfies (\ref{eqU}) by (\ref{Qmax}).

Note that $Q^{\downarrow}(i)=Q(i)$ by the definition of $Q$.
Since $Q\ne \mathcal{C}_m(P)$, the set $\{1\le k\le L ~|~Q(k)>\mathcal{C}_m(P)(k)\}$ is not empty.
We introduce the integer 
\begin{eqnarray}
k_0:=\min\{1\le k\le m ~|~ Q(k)>\mathcal{C}_m(P)(k)\}.
\end{eqnarray}
We separately give proofs for two cases; $k_0\ge J_{P,m}$ and $k_0< J_{P,m}$.

We first treat the case that $k_0\ge J_{P,m}$.
Then, $Q(j)\le \mathcal{C}_m(P)(j)=P^{\downarrow}(j)$ for $j=1, ..., J_{P,m}-1$.
Since $P\prec Q$ by the definition of $Q$,
the equation $Q(j)=\mathcal{C}_m(P)=P^{\downarrow}(j)$ holds for $j=1, ..., J_{P,m}-1$.
Then we have 
\begin{eqnarray}
\sum_{j=J_{P,m}}^{m} Q(j)
~=~1-\sum_{j=1}^{J_{P,m}-1} Q(j)
~=~1-\sum_{j=1}^{J_{P,m}-1} P(j)
~=~\sum_{j=J_{P,m}}^{|{\cal X}|} P(j).
\nonumber
\end{eqnarray}
Since $\mathcal{C}_m(P)$ is uniform on $\{J_{P,m}, ..., L\}$ and $\mathcal{C}_m(P)\ne Q$, 
$Q$ is not uniform on $\{J_{P,m}, ..., m\}$.
Thus, the following strict inequality holds by the Schwarz inequality:
\begin{eqnarray}
F(Q, U_m)
&=&\sum_{j=1}^{J_{P,m}-1}\sqrt{\mathcal{C}_m(P)(j)}\sqrt{U_m(j)} + \sum_{j=J_{P,m}}^{m}\sqrt{Q(j)}\sqrt{U_m(j)}\nonumber\\
&<&\sum_{j=1}^{J_{P,m}-1}\sqrt{\mathcal{C}_m(P)(j)}\sqrt{U_m(j)} + \sqrt{\sum_{j=J_{P,m}}^{m} Q(j)}\sqrt{\sum_{j=J_{P,m}}^{m} U_m(j)}\nonumber\\
&=&\sum_{j=1}^{J_{P,m}-1}\sqrt{\mathcal{C}_m(P)(j)}\sqrt{U_m(j)} +\sqrt{\sum_{j=J_{P,m}}^{|{\cal X}|} P(j)}\sqrt{(m+1-J_{P,m})U_m(j)}\nonumber\\
&=&\sum_{j=1}^{J_{P,m}-1}\sqrt{\mathcal{C}_m(P)(j)}\sqrt{U_m(j)} +\sum_{j=J_{P,m}}^{m}\sqrt{\frac{\sum_{j=J_{P,m}}^{|{\cal X}|} P(j)}{m+1-J_{P,m}}}\sqrt{U_m(j)}\nonumber\\
&=&\sum_{j=1}^{J_{P,m}-1}\sqrt{\mathcal{C}_m(P)(j)}\sqrt{U_m(j)} + \sum_{j=J_{P,m}}^{m}\sqrt{\mathcal{C}_m(P)(j)}\sqrt{U_m(j)}\nonumber\\
&=&F(\mathcal{C}_m(P), U_m).
\Label{fid>1}
\end{eqnarray}
On the other hand, since $P\prec \mathcal{C}_m(P)$, the following holds:
\begin{eqnarray}\Label{fidle1}
F(\mathcal{C}_m(P), U_m)
~\le~ F^{\mathcal{M}}(P\to U_m)
~=~ F(Q, U_m).
\end{eqnarray}
The inequalities (\ref{fid>1}) and (\ref{fidle1}) are contradictory to each other. 

Next, we treat the case that $k_0< J_{P,m}$. 
In the following, we show that there exists a probability distribution $Q'$ such that  the following contradictory inequalities hold:
\begin{eqnarray}\Label{fidle2}
F(Q', U_m)\le F^{\mathcal{M}}(P\to U_m)
\end{eqnarray}
and
\begin{eqnarray}\Label{fid>2}
F(Q', U_m)> F^{\mathcal{M}}(P\to U_m).
\end{eqnarray}

We first define a probability distribution $Q'$.
Let 
\begin{eqnarray}
l_0:={\rm max}\{1\le k\le m ~|~ Q(k_0)=Q(k)\}.
\end{eqnarray}
Then, $Q(l_0)>Q(l_0+1)$ holds.
For 
\begin{eqnarray}
\epsilon:=\min\left\{\frac{1}{2}(Q(l_0)-Q(l_0+1)),~ Q(k_0)-\mathcal{C}_m(P)(k_0)\right\}>0,
\nonumber
\end{eqnarray} 
we define a probability distribution $Q'$ as 
\begin{eqnarray}
Q'(j)
=\left\{
\begin{array}{ll}
Q(l_0)-\epsilon & \hbox{ if }~j=l_0 \\
Q(l_0+1)+\epsilon &\hbox{ if }~j=l_0+1 \\
Q(j) & \hbox{ otherwise}. 
\end{array}
\right.
\end{eqnarray}
Note that $Q'(j)\ge Q'(j+1)$ and thus $Q'^{\downarrow}=Q'$.

We show (\ref{fidle2}) for $Q'$.
To do so,
we show $P\prec Q'$, 
that is,
\begin{eqnarray}
\sum_{j=1}^l P^{\downarrow}(j) 
~\le~ \sum_{j=1}^l Q'(j)
\Label{Q'ineq}
\end{eqnarray}
for any $l\in\N$.
For $l\ne l_0$,
$\sum_{j=1}^l Q'(j) = \sum_{j=1}^l Q(j)$ holds by the definition of $Q'$,
and thus, (\ref{Q'ineq}) holds by the majorization condition $P\prec Q$.
For $l= l_0$,
we obtain (\ref{Q'ineq}) as follows:
\begin{eqnarray*}
\sum_{j=1}^{l_0} Q'(j)
&
\ge& \sum_{j=1}^{l_0} Q(j) - (Q(k_0)-\mathcal{C}_m(P)(k_0))
~=~ \sum_{j=1}^{l_0} Q(j) - (Q(l_0)-\mathcal{C}_m(P)(k_0))
\\
&\ge& \sum_{j=1}^{l_0-1} P^{\downarrow}(j) + \mathcal{C}_m(P)(k_0)
~=~ \sum_{j=1}^{l_0-1} P^{\downarrow}(j) + P(k_0)\\
&\ge&\sum_{j=1}^{l_0} P^{\downarrow}(j),
\end{eqnarray*}
where the second inequality comes from that
$\mathcal{C}_m(P)(k_0)=P(k_0)$ holds by $k_0<J_{P,m}$ and the definition $\mathcal{C}_m(P)$.
Since $P\prec Q'$, (\ref{fidle2}) holds by the definition of $F^{\mathcal{M}}$.

Then we show (\ref{fid>2}).
We note taht $Q(j)=\mathcal{C}_m(P)(j)=P(j)$ for $j<k_0$ by the definition of $k_0$ and $\mathcal{C}_m(P)$.
Thus we obtain  (\ref{fid>2}) as follows:
\begin{eqnarray}
F(Q', U_m)
&=&\sum_{j\ne l_0, l_0+1}\sqrt{Q'(j)}\sqrt{U_m(j)} + \sqrt{Q'(l_0)}\sqrt{U_m(l_0)} + \sqrt{Q'(l_0+1)}\sqrt{U_m(l_0+1)}\nonumber\\
&=&\sum_{j\ne l_0, l_0+1}\sqrt{Q(j)}\sqrt{U_m(j)} + \sqrt{Q(l_0)-\epsilon}\sqrt{U_m(l_0)} + \sqrt{Q(l_0+1)+\epsilon}\sqrt{U_m(l_0+1)}\nonumber\\
&=&\sum_{j\ne l_0, l_0+1}\sqrt{Q(j)}\sqrt{U_m(j)} + (\sqrt{Q(l_0)-\epsilon} + \sqrt{Q(l_0+1)+\epsilon})\sqrt{1/m}\nonumber\\
&>&\sum_{j\ne l_0, l_0+1}\sqrt{Q(j)}\sqrt{U_m(j)} + (\sqrt{Q(l_0)} + \sqrt{Q(l_0+1)})\sqrt{1/m}\nonumber\\
&=&\sum_{j\ne l_0, l_0+1}\sqrt{Q(j)}\sqrt{U_m(j)} + \sqrt{Q(l_0)}\sqrt{U_m(l_0)} + \sqrt{Q(l_0+1)}\sqrt{U_m(l_0+1)}\nonumber\\
&=&F(Q, U_m)\nonumber\\
&=&F^{\mathcal{M}}(P\to U_m)
\Label{fid>2}
\end{eqnarray}
where the inequality is obtained by a simple calculation\footnote{A function $f(x):=\sqrt{Q(l_0)-x} + \sqrt{Q(l_0+1)+x}$ is strictly increasing when $-Q(l_0+1)<x< \frac{1}{2}(Q(l_0)-Q(l_0+1))$.}.
\endproof

\subsection{Proof of Lemmas: Limit of Tail Probability}
\Label{app.tail}

We prove Lemmas \ref{lem.central} and \ref{lem.central2} which gives the upper tail probability of $Q^{n\downarrow}$.

~

\noindent{\bf Proof of Proposition \ref{lem.central}:}
Let  $\tilde{S}_n^Q(x):=\{i\in\N~|~{Q^n}^{\downarrow}(i) ~\ge~ 2^{-H(Q)n-x\sqrt{n}}\}$ and $\tilde{S}_n^Q(x, x'):= \tilde{S}_n(x')\setminus \tilde{S}_n(x)$.
Then, the followings are obtained by the central limit theorem:
\begin{eqnarray}
\displaystyle\lim_{n\to\infty} Q^{n\downarrow}(\tilde{S}_n^Q(x))&=&\Phi\left(\frac{x}{\sqrt{V(Q)}}\right).\Label{central1'}
\end{eqnarray}
Next, we will show that the following holds for an arbitrary $\delta>0$:
\begin{eqnarray}
\tilde{S}_n^Q(x) ~\subset~ S_n^Q(x) ~\subset~ \tilde{S}_n^Q(x+\delta)
\Label{hougan}
\end{eqnarray}
when $n\in \N$ is large enough.
Since $\tilde{S}_n^Q(x)\subset S_n^Q(x)$ obviously holds for any $n\in\N$, it is enough to show $S_n^Q(x)\subset\tilde{S}_n^Q(x+\delta)$.

We note that ${Q^n}^{\downarrow}(i)- 2^{-H(Q)n-(x+\frac{\delta}{2})\sqrt{n}}\ge0$ holds if and only if $i\in \tilde{S}_n^Q(x+\frac{\delta}{2})$.
Thus, we have the following inequality for an arbitrary subset $\mathcal{S}\subset\N$:
\begin{eqnarray}
&&{Q^n}^{\downarrow}\left(\tilde{S}_n^Q\left(x+\frac{\delta}{2}\right)\right)- 2^{-H(Q)n-(x+\frac{\delta}{2})\sqrt{n}}\left|\tilde{S}_n^Q\left(x+\frac{\delta}{2}\right)\right|\nonumber\\
&\ge&
{Q^n}^{\downarrow}(\mathcal{S})- 2^{-H(Q)n-(x+\frac{\delta}{2})\sqrt{n}}|\mathcal{S}|.
\Label{subset}
\end{eqnarray}
In particular, when $\mathcal{S}=\tilde{S}_n^Q(x+\delta)$, we obtain 
\begin{eqnarray}
\left|\tilde{S}_n^Q\left(x+\frac{\delta}{2},x+\delta\right)\right|
~\ge~
2^{H(Q)n+(x+\frac{\delta}{2})\sqrt{n}}{Q^n}^{\downarrow}\left(\tilde{S}_n^Q\left(x+\frac{\delta}{2},x+{\delta}\right)\right)
\Label{subset2}
\end{eqnarray}
from (\ref{subset}).
Since 
\begin{eqnarray}
\lim_{n\to\infty}{Q^n}^{\downarrow}\left(\tilde{S}_n^Q\left(x+\frac{\delta}{2},x+{\delta}\right)\right)
~=~\Phi\left(\frac{x+\delta}{\sqrt{V(Q)}}\right)-\Phi\left(\frac{x+\frac{\delta}{2}}{\sqrt{V(Q)}}\right)
~>~0
\Label{subset3}
\end{eqnarray}
 from (\ref{central1'}), the following holds for large enough $n\in\N$:
\begin{eqnarray}
2^{H(Q)n+(x+\frac{\delta}{2})\sqrt{n}}{Q^n}^{\downarrow}\left(\tilde{S}_n^Q\left(x+\frac{\delta}{2},x+{\delta}\right)\right)
~\ge~
2^{H(Q)n+x\sqrt{n}}.
\Label{subset4}
\end{eqnarray}
Therefore, the inequality 
\begin{eqnarray}
\left|\tilde{S}_n^Q\left(x+\delta\right)\right|
~\ge~ \left|\tilde{S}_n^Q\left(x+\frac{\delta}{2},x+\delta\right)\right|
~\ge~ 2^{H(Q)n+x\sqrt{n}}
~=~ |S_n^Q(x)|
\Label{subset5}
\end{eqnarray}
holds from (\ref{subset2}) and (\ref{subset4}) for large enough $n\in\N$, and it implies (\ref{hougan}).
Thus, (\ref{central1}) is derived by (\ref{central1'}) and (\ref{hougan}).
\endproof

~

\noindent{\bf Proof of Proposition \ref{lem.central2}:}
For an arbitrarily small $\epsilon>0$, 
the following holds for large enough $n$:
\begin{eqnarray}
S_{\frac{H(P)}{H(Q)}n+b\sqrt{n}}^Q\left(\frac{x-bH(Q)}{\sqrt{\frac{H(P)}{H(Q)}}}-\epsilon\right)
~\subset~ S_n^P(x)
~\subset~ S_{\frac{H(P)}{H(Q)}n+b\sqrt{n}}^Q\left(\frac{x-bH(Q)}{\sqrt{\frac{H(P)}{H(Q)}}}+\epsilon\right).\Label{PtoQ}
\end{eqnarray}
Thus, we obtain  (\ref{central2}) as follows
\begin{eqnarray}
\displaystyle\lim_{n\to\infty} Q^{\frac{H(P)}{H(Q)}n+b\sqrt{n}\downarrow}(S_n^P(x))
&=&\displaystyle\lim_{n\to\infty} Q^{\frac{H(P)}{H(Q)}n+b\sqrt{n}\downarrow}\left(S_{\frac{H(P)}{H(Q)}n+b\sqrt{n}}^Q\left(\frac{x-bH(Q)}{\sqrt{\frac{H(P)}{H(Q)}}}\right)\right)\nonumber\\
&=&\Phi\left(\frac{x-bH(Q)}{\frac{H(P)}{H(Q)}\sqrt{V(Q)}}\right)\\
&=&\Phi_{P,Q,b}\left(\frac{x}{\sqrt{V(P)}}\right),
\end{eqnarray}
where the second equation is obtained from (\ref{central1}).
When $P$ is also a non-uniform distribution,
$V(P)$ is non-zero and thus $C_{P,Q}$ is well-defined.
Then, the right hand side of (\ref{central2}) coincides with the right hand side of (\ref{central3}) from the definition.
\endproof

\subsection{Proof of Lemma: The Existence of A Suitable Deterministic Conversion}\Label{app.Exist}


\noindent{\bf Proof of Lemma \ref{Wlem}:}
Let $g:S_1\to \N$ and $h:S_2\to \N$ be injective maps such that $W_g(B)(i)\ge W_g(B)(i+1)$ and $W_h(C)(j)\ge W_j(C)(j+1)$.
Then, we denote $W_g(B)$ and $W_h(C)$ by $B^{\downarrow}$ and $C^{\downarrow}$.
Let $k_0=0$ and  $k_i$ be a non-negative integer for $i\ge1$ such that 
\begin{eqnarray}
\sum_{k=k_{i-1}+1}^{k_i} C^{\downarrow} (k)
~\le~ B^{\downarrow}(i)
~\le~ \sum_{k=k_{i-1}+1}^{k_i+1} C^{\downarrow} (k).
\end{eqnarray}
The above $\{k_i\}_{i=1}^{\infty}$ always exist but may not be unique. 
For $i\ge1$ and $k_{i-1}+1\le k \le k_i$,
a map $f':\N\to\N$ defined by $g'(k)=i$  satisfies the following inequalities:
\begin{eqnarray}
B^{\downarrow}(i) 
&\le& \sum_{k=k_{i-1}+1}^{k_i+1} C^{\downarrow} (k)\nonumber\\
&=& \sum_{k=k_{i-1}+1}^{k_i} C^{\downarrow} (k) + C^{\downarrow} (k_i+1)\nonumber\\
&=& W_{g'}(C^{\downarrow})(i) + C^{\downarrow} (k_i+1)\nonumber\\
&\le& W_{g'}(C^{\downarrow})(i) + \max_{j\in S_2}C(j).
\Label{Wmap3}
\end{eqnarray}

In the following, we construct the map $f$ which satisfies (\ref{Wmap2}).
For $j\in S_2$ in which $g'\circ h(j)$ is in the image of $g$,
we define as $f(j):=g^{-1}\circ W_{g'}\circ h(j)$.
For other $j\in S_2$, 
there is no constraint for the value of $f(j)$ as long as $f(j)$ is in $S_2$.  
Then we obtain (\ref{Wmap2}) as follows:
\begin{eqnarray*}
B(i) 
&=&B^{\downarrow}(g(i)) \\
&\le& W_{g'}(C^{\downarrow})(g(i)) + \max_{j\in S_2}C(j)\\
&=& W(C)(i) + \max_{j\in S_2}C(j),
\end{eqnarray*}
where the inequality comes from (\ref{Wmap3}).
\endproof

\subsection{Proof of Lemma: Converse Part of Non-Uniform Case in Asymptotic Theory}
\Label{app.Converse}

\noindent{\bf Proof of Lemma \ref{20}:}
\noindent{\it Sketch of proof:}\quad
It is enough to show that 
\begin{eqnarray}
&&{\underset{n\to\infty}{\rm limsup}}F(P'^{\downarrow}_n, Q^{\frac{H(P)}{H(Q)}n+b\sqrt{n}\downarrow})\nonumber\\
&\le&\sqrt{\Phi(t)}\sqrt{\Phi_{P, Q, b}(t)}
+\int_{t}^{t'}\sqrt{\phi(x)} \sqrt{\phi_{P, Q, b}(x)}dx\nonumber\\
&&+\sqrt{1-{\Phi}(t')} \sqrt{1-{\Phi}_{P, Q, b}(t')}
\Label{Fineq}
\end{eqnarray}
for an arbitrary sequence $\{P'_n\}_{n=1}^{\infty}$ of probability distributions such that $P'_n\succ P_n$.

To show (\ref{Fineq}), we choose  $x_i^I:=t+\frac{t'-t}{I}i$ for a natural number $I$.
We set as 
\begin{eqnarray}
a_0^I:={\Phi}({x}_{0}^I), 
&a_i^I:={\Phi}({x}_i^I)-{\Phi}({x}_{i-1}^I),& 
a_{I+1}^I=1-{\Phi}({x}_{I}^I),\\
b_0^I={\Phi}_{P, Q, b}({x}_{0}), 
&b_i^I={\Phi}_{P, Q, b}({x}_{i}^I)-{\Phi}_{P, Q, b}({x}_{i-1}^I),&
b_{I+1}^I=1-{\Phi}_{P,Q,b}({x}_{I}^I).
\end{eqnarray}
Then, 
we will show that there is a sequence ${\bf c}:=\{c_i^I\}_{i=0}^{I+1}$
such that
\begin{eqnarray}
{\underset{n\to\infty}{\rm limsup}}F(P'^{\downarrow}_n, Q^{\downarrow}_n)
~\le~ \sum_{i=0}^{I+1}\sqrt{c_i^I} \sqrt{b_i^I}
\Label{cineq}
\end{eqnarray}
holds and 
${\bf a}:=\{a_i^I\}_{i=0}^{I+1}$, ${\bf b}:=\{b_i^I\}_{i=0}^{I+1}$ and ${\bf c}$ satisfy
the assumptions of Lemma \ref{naiseki2}.
Then, using Lemma \ref{naiseki2},
we have (\ref{Fineq}) as follows:
\begin{eqnarray}
{\underset{n\to\infty}{\rm limsup}}F(P'^{\downarrow}_n, Q^{\downarrow}_n)
&\le& \lim_{I\to\infty}\sum_{i=0}^{I+1}\sqrt{c_i^I} \sqrt{b_i^I}\nonumber\\
&\le&\lim_{I\to\infty}\sum_{i=0}^{I+1}\sqrt{a_i^I} \sqrt{b_i^I}\nonumber\\
&=&\sqrt{{\Phi}(t)}\sqrt{{\Phi}_{P, Q, b}(t)}\Label{nai3}\\
&&+\lim_{I\to\infty}\sum_{i=1}^{I}\sqrt{{\Phi}({x}_i^I)-{\Phi}({x}_{i-1}^I)} \sqrt{{\Phi}_{P, Q, b}({x}_{i}^I)-{\Phi}_{P, Q, b}({x}_{i-1}^I)}\nonumber\\
&&+\sqrt{1-{\Phi}(t')} \sqrt{1-{\Phi}_{P, Q, b}(t')}\nonumber\\
&=&\sqrt{{\Phi}(t)}\sqrt{{\Phi}_{P, Q, b}(t)}
+\int_{t}^{t'}\sqrt{\phi(x)}\sqrt{\phi_{P, Q, b}(x)}dx\nonumber\\
&&+\sqrt{1-{\Phi}(t')} \sqrt{1-{\Phi}_{P, Q, b}(t')}.\nonumber 
\end{eqnarray}

\noindent{\it Detailed proof:}\quad
From the above sketch of proof,
all we have to do is to show the existence of ${\bf c}:=\{c_i^I\}_{i=0}^{I+1}$
such that (\ref{cineq}) holds and ${\bf a}$, ${\bf b}$ and ${\bf c}$ satisfy the assumptions of Lemma \ref{naiseki2}.

First, we construct the sequence ${\bf c}=\{c_i^I\}_{i=0}^{I+1}$
such (\ref{cineq}) holds.
For simplicity, we denote $\sqrt{V(P)}x$ by $\tilde{x}$ for an arbitrary real number $x$ in the following.
Let a map $g_n^I:\N\to\{0,1,...I,I+1\}$ satisfy 
$g_n^I(l)=0$ for $l\in S_n^P(\tilde{x}_{0}^I)$,
$g_n^I(l)=i$ for $l\in S_n^P(\tilde{x}_{i-1}^I, \tilde{x}_{i}^I)$,
and 
$g_n^I(l)=I+1$ for $l\in \N\setminus S_n^P(\tilde{x}_{I}^I)$.
Then we have the following inequality  by the monotonicity of the fidelity for $g_n^I$: 
\begin{eqnarray}
&&F(P'^{\downarrow}_n, Q^{\frac{H(P)}{H(Q)}n+b\sqrt{n}\downarrow})\nonumber\\
&\le&F(g_n^I(P'^{\downarrow}_n), g_n^I(Q^{\frac{H(P)}{H(Q)}n+b\sqrt{n}\downarrow}))\nonumber\\
&=&\sqrt{P'^{\downarrow}_n( S_n^P(\tilde{x}_{0}^I))} \sqrt{Q^{\frac{H(P)}{H(Q)}n+b\sqrt{n}\downarrow}( S_n^P(\tilde{x}_{0}^I))}\Label{H6-18-1}\\
&&+\sum_{i=1}^{I}\sqrt{P'^{\downarrow}_n( S_n^P(\tilde{x}_{i-1}^I, \tilde{x}_{i}^I))} \sqrt{Q^{\frac{H(P)}{H(Q)}n+b\sqrt{n}\downarrow}( S_n^P(\tilde{x}_{i-1}^I, \tilde{x}_{i}^I))}\nonumber\\
&&+\sqrt{1-P'^{\downarrow}_n( S_n^P(\tilde{x}_{I}^I))} \sqrt{1-Q^{\frac{H(P)}{H(Q)}n+b\sqrt{n}\downarrow}(S_n^P(\tilde{x}_{I}^I))}.\nonumber
\end{eqnarray}
Here, we denote the right hand side of (\ref{H6-18-1}) by $R_I(n)$.
Then, we can choose a subsequence $\{n_l\}_l\subset\{n\}$ 
such that 
$\lim_{l \to \infty} R_I(n_l)
={\underset{n\to\infty}{\rm limsup}} R_I(n)$
and the limits 
\begin{eqnarray}
c_0^I&:=&\displaystyle\lim_{l\to\infty} P'^{\downarrow}_{n_l}( S_{n_l}(\tilde{x}_{0}^I)),\\
c_i^I&:=&\displaystyle\lim_{l\to\infty} P'^{\downarrow}_{n_l}( S_{n_l}(\tilde{x}_{i-1}^I, \tilde{x}_{i}^I)),\\
c_{I+1}^I&:=&1-\displaystyle\lim_{l\to\infty} P'^{\downarrow}_{n_l}( S_{n_l}(\tilde{x}_{I}^I))
\end{eqnarray}
exist for $i=1, \ldots, I$.
Hence, we obtain (\ref{cineq}) as follows:
\begin{eqnarray}
{\underset{n\to\infty}{\rm limsup}}F(P'^{\downarrow}_n, Q^{\downarrow}_n)
&\le&{\underset{n\to\infty}{\rm limsup}}R_I(n)
=\lim_{l\to\infty} R_I(n_l)\nonumber\\
&=&\sqrt{c_0^I} \sqrt{{\Phi}_{P, Q, b}({x}_{0})}\Label{c_0eq}\\
&&+\sum_{i=1}^{I}\sqrt{c_i^I} \sqrt{{\Phi}_{P, Q, b}({x}_{i}^I)-{\Phi}_{P, Q, b}({x}_{i-1}^I)}\nonumber\\
&&+\sqrt{c_{I+1}^I} \sqrt{1-{\Phi}_{P, Q, b}({x}_{I}^I)},\nonumber\\
&=&\sum_{i=0}^{I+1}\sqrt{c_i^I} \sqrt{b_i^I}\nonumber
\end{eqnarray} 
where the equality (\ref{c_0eq}) follows from Lemma \ref{lem.central}.

Next,
we show that ${\bf a}$, ${\bf b}$ and ${\bf c}$ satisfy the assumptions of Lemma \ref{naiseki2}.
First,
${\bf a}$, ${\bf b}$ and ${\bf c}$ are probability distributions by the definitions. 
Second, (\ref{gijimajo}) is obtained as follows:
\begin{eqnarray}
&&\sum_{i=0}^k a_i
~=~\Phi({x}_{k}^I)
~=~\lim_{l\to\infty} P^{n_l\downarrow}(S_{n_l}^P(\tilde{x}_{k}^I))
~\le~ \lim_{l\to\infty} P'^{\downarrow}_{n_l}(S_{n_l}^P(\tilde{x}_{k}^I))
~=~\sum_{i=0}^k c_i^I
\end{eqnarray} 
holds for $k=0,1,...,I$ since $P^n\prec P'_n$, and $\sum_{i=0}^{I+1} a_i
=1=\sum_{i=0}^{I+1} c_i^I$ holds.
Finally, we show (\ref{ratio}) for $i=1,2,...,I$.
The equations $a_0/b_0=\phi(s)/\phi_{P,Q,b}(s)$ and $a_{I+1}/b_{I+1}=\phi(s')/\phi_{P,Q,b}(s')$ hold by the assumption (II).
Moreover, there exist $z_i\in[x_{i-1}^I,x_i^I]$ for $i=1,...,I$ such that 
 $a_i/b_i=\phi(z_i)/\phi_{P,Q,b}(z_i)$ for $i=1,...,I$ due to the mean value theorem. 
Then $z_i\in(s,s')$ holds because of the relation $t=x_{0}^I\le x_{i-1}^I\le z_i\le x_{i}^I \le x_I^I= t' $ and the assumption (I).
Since $\phi(x)/\phi_{P,Q,b}(x)$ is monotonically decreasing on $(s,s')$ by the assumption (III), we have (\ref{ratio}) for $i=1,...,I+1$.
\endproof

\vspace{0.5em}



\begin{thebibliography}{99}

\bibitem{AT89}
M. P. Allen and D. J. Tildesley,
{\em Computer Simulation of Liquids},
Clarendon Press (1989).

\bibitem{ACP04}
F. Anselmi, A. Chefles, and M. B. Plenio, 
``Local copying of orthogonal entangled quantum states,"
{\it New J. Phys.}, {\bf 6(1)}, 164, (2004).

\bibitem{Arn86}
B. C. Arnold, 
\textit{Majorization and the Lorenz Order: A Brief Introduction}, 
Springer-Verlag, (1986).

\bibitem{BZ06}
I. Bengtsson and K. Zyczkowski, 
\textit{Geometry of Quantum States}, 
Cambridge University Press, (2006).

\bibitem{BBPS96} 
C. H. Bennett, H. J. Bernstein, S. Popescu and B. Schumacher, 
``Concentrating partial entanglement by local operations,"
{\it Phys. Rev. A}, {\bf 53(4)}, 2046, (1996).

\bibitem{BDSW96}
C. H. Bennett, D. P. DiVincenzo, J. Smolin and W. K. Wootters, 
``Mixed-state entanglement and quantum error correction,"
{\it Phys. Rev. A}, {\bf 54(5)}, 3824, (1996).

\bibitem{BFS12}
P. Bratley, B. L. Fox, and L. E. Schrage,
{\em A Guide to Simulation},
Springer Science \& Business Media (2012).

\bibitem{BH98}
V. Buzek, and M. Hillery, 
``Universal optimal cloning of arbitrary quantum states: From qubits to quantum registers,"
{\it Phys. Rev. Lett.}, {\bf 81(22)}, 5003, (1998).

\bibitem{BPRST00} 
C. H. Bennett, S. Popescu, D. Rohrlich, J. A. Smolin and A. V. Thapliyal, 
``Exact and asymptotic measures of multipartite pure-state entanglement,"
{\it Phys. Rev. A}, {\bf 63(1)}, 012307 (2000).

\bibitem{BS98}
C. H. Bennett and P. W. Shor, 
``Quantum information theory,"
{\it IEEE Trans. Inform. Theory}, {\bf 44(6)}, 2724, (1998).

\bibitem{BSC14}
M. Bourguignon, R. B. Silva and G. M. Cordeiro,
``The Weibull-G Family of Probability Distributions,"
{\it J. Data Sci.}, {\bf 12} (2014).

\bibitem{CYY13}
G. Chiribella, Y. Yang and A. C. C. Yao,
``Quantum replication at the Heisenberg limit,"
{\it Nature comm.}, {\bf 4} (2013).

\bibitem{CP02}
D. Collins, and S. Popescu,
``Classical analog of entanglement," 
{\it Physical Review A}, {\bf 65(3)}, 032321, (2002).

\bibitem{DL14}
N. Datta and F. Leditzky, 
``Second-order asymptotics for source coding, dense coding and pure-state entanglement conversions,"
{\it IEEE Trans. Inform. Theory}, {\bf 61(1)}, 582-608 (2015).

\bibitem{Dev86}
L. Devroye
{\em Non-Uniform Random Variate Generation},
Springer-Verlag, New York (1986).

\bibitem{DHR02}
M. J. Donald, M. Horodecki and O. Rudolph, 
``The uniqueness theorem for entanglement measures,"
{\it J. Math. Phys.}, {\bf 43(9)}, 4252, (2002).

\bibitem{Fil04}
R. Filip, 
``Quantum partial teleportation as optimal cloning at a distance," 
{\it Phys. Rev. A}, {\bf 69(5)}, 052301, (2004).

\bibitem{GFE09}
D. Gross, S. T. Flammia and J. Eisert,
 ``Most quantum states are too entangled to be useful as computational resource," 
{\it Phys. Rev. Lett.}, {\bf 102(19)}, 190501, (2009).

\bibitem{Han03}
T. S. Han, 
\textit{Information-Spectrum Methods in Information Theory}, 
Springer, New York, (2003).

\bibitem{Han05}
T. S. Han, 
``Folklore in source coding: Information-spectrum approach,''
{\em IEEE Trans. Inf. Theory}, {\bf 51(2)}, 747-753 (2005).

\bibitem{HL04}
A. Harrow and  H. K. Lo,   
``A tight lower bound on the classical communication cost of entanglement dilution," 
{\it IEEE Trans. Inform. Theory}, {\bf 50(2)}, 319, (2004).

\bibitem{HHT01} P. Hayden, M. Horodecki and B. M. Terhal, 
``The asymptotic entanglement cost of preparing a quantum state,"
{\it J. Phys. A}, {\bf 34(35)}, 6891, (2001).

\bibitem{HW03}
P. Hayden and  A. Winter, 
``Communication cost of entanglement transformations,"
{\it Phys. Rev. A}, {\bf 67(1)}, 012326, (2003).%

\bibitem{Hay06}
M. Hayashi, 
``General formulas for fixed-length quantum entanglement concentration,"
{\it IEEE Trans. Inform. Theory}, {\bf 52(5)}, 1904, (2006).

\bibitem{Hay08}
M. Hayashi, 
``Second-order asymptotics in fixed-length source coding and intrinsic randomness,"
{\it IEEE Trans. Inform. Theory}, {\bf 54(10)}, 4619, (2008).

\bibitem{Hay09}
M. Hayashi, 
``Information-Spectrum Approach to Second-Order Coding Rate in Channel Coding," 
{\it IEEE Trans. Inform. Theory}, {\bf 55(11)}, 4947, 2009.

\bibitem{Hay11}
M. Hayashi, 
``Comparison between the Cramer-Rao and the mini-max approaches in quantum channel estimation," 
{\it Comm. Math. Phys.}, {\bf 304(3)}, 689, (2011).


\bibitem{HLD04}
W. Hormann, J. Leydold, and G. Derflinger, 
{\em Automatic Nonuniform Random Variate Generation},
Springer-Verlag, Berlin Heidelberg (2004).

\bibitem{HHHH09}
R. Horodecki, P. Horodecki, M. Horodecki and K. Horodecki, 
``Quantum Entanglement,"
{\it Rev. Mod. Phys.}, {\bf 81(2)}, 865, (2009).

\bibitem{IH09}
S. Ishizaka and T. Hiroshima, 
``Quantum teleportation scheme by selecting one of multiple output ports,''
{\it Phys. Rev. A}, {\bf 79}, 042306, (2009).



\bibitem{KH13}
W. Kumagai, M. Hayashi,  
``Entanglement Concentration is Irreversible," 
{\it Phys. Rev. Lett.}, {\bf 111(13)}, 130407, (2013).

\bibitem{Li12}
K. Li, 
``Second order asymptotics for quantum hypothesis testing,"
{\em Annals of Statistics}, {\bf 42(1)}, 171-189 (2014).

\bibitem{LP01}
H. K. Lo, and S. Popescu,  
``Concentrating entanglement by local actions: Beyond mean values," 
{\it Physical Review A}, {\bf 63(2)}, 022301,  (2001).

\bibitem{MO79}
A. W. Marshall and  I. Olkin, 
\textit{Inequalities: Theory of Majorization and Its Applications},
 Academic Press, New York, (1979).
 
 \bibitem{MH07}
K. Matsumoto, and M. Hayashi, 
"Universal distortion-free entanglement concentration," 
{\it Physical Review A}, {\bf 75}, 062338 (2007).

\bibitem{MPSVW10}
K. Modi, T. Paterek, W. Son, V. Vedral, and M. Williamson, 
``Unified view of quantum and classical correlations,"
{\it Phys. Rev. Lett.}, {\bf 104(8)}, 080501, (2010).

\bibitem{Nie99} 
M. A. Nielsen, 
``Conditions for a class of entanglement transformations,"
{\it Phys. Rev. Lett.}, {\bf 83(2)}, 436, (1999).

\bibitem{NC00}
M. A. Nielsen and  I. L. Chuang, 
\textit{Quantum Computation and Quantum Information}, 
Cambridge University Press, Cambridge, (2000).

\bibitem{NH13}
R. Nomura and  T. S. Han, 
``Second-Order Resolvability, Intrinsic Randomness, and Fixed-Length Source Coding for Mixed Sources,"
{\it IEEE Trans. Inform. Theory}, {\bf 59(1)}, 1, (2013).

\bibitem{OP93}
M. Ohya and D. Petz, 
\textit{Quantum entropy and its use}, 
Springer-Verlag, Heidelberg, (1993).

\bibitem{OH06}
M. Owari and  M. Hayashi,
``Local copying and local discrimination as a study for nonlocality of a set of states,"
{\it Phys. Rev. A}, {\bf 74(3)}, 032108, (2006).

\bibitem{PPV10}
Y. Polyanskiy, H.V. Poor and  S. Verd\'{u} , 
``Channel Coding Rate in the Finite Blocklength Regime", 
{\it IEEE Trans. Inform. Theory}, {\bf 56(5)}, 2307 (2010).


\bibitem{Ren15}
J. M. Renes,   
``Relative submajorization and its use in quantum resource theories" 
{\it arXiv}:1510.03695, (2015).


\bibitem{TH12}
M. Tomamichel and  M. Hayashi, 
``A hierarchy of information quantities for finite block length analysis of quantum tasks,''  
{\it IEEE Trans. Inform. Theory}, {\bf 59(11)}, 7693-7710 (2013).

\bibitem{TT13}
M. Tomamichel and  V. Y. F. Tan,  
``Second-Order Asymptotics of Classical-Quantum Channels,"
{\it Comm. Math. Phys.}, {\bf 338(1)}, 103-137 (2015).

\bibitem{Vaa98}
A. W. Van der Vaart. 
{\it Asymptotic Statistics}, 
Cambridge University Press, (1998).%

\bibitem{VV95}
S. Vembu and S. Verd\'{u},  
``Generating random bits from an arbitrary source: fundamental limits,"
{\it IEEE Trans. Inform. Theory}, {\bf 41(5)}, 1322, (1995).

\bibitem{Vid99} 
G. Vidal, 
``Entanglement of pure states for a single copy,"
{\it Phys. Rev. Lett.}, {\bf 83(5)}, 1046, (1999).

\bibitem{VJN00}
G. Vidal, D. Jonathan and  M. A. Nielsen, 
``Approximate transformations and robust manipulation of bipartite pure-state entanglement,"
{\it Phys. Rev. A}, {\bf 62(1)}, 012304, (2000).

\bibitem{VKV98}
 K. Visweswariah, S. R. Kulkarni, and S. Verd\'{u}, 
``Source codes as random number generators,'' 
{\em IEEE Trans. Inf. Theory}, {\bf 44(2)}, 462-471 (1998).

\bibitem{Wer98}
R. F. Werner, 
``Optimal cloning of pure states,''
{\it Phys. Rev. A}, {\bf 58(3)}, 1827, (1998).
\end{thebibliography}
\end{document}